\documentclass[11pt]{article}
\usepackage[utf8]{inputenc}
\usepackage{style}
\usepackage{framed}
\usepackage{nicefrac}
\usepackage[parfill]{parskip}
\usepackage{setspace}
\usepackage{hyperref}
\usepackage{tikz}
\usetikzlibrary{arrows.meta}
\usetikzlibrary{decorations.pathreplacing}
\usetikzlibrary{patterns}
\usetikzlibrary{calc}
\tikzset{
    ncbar angle/.initial=90,
    ncbar/.style={
        to path=(\tikztostart)
        -- ($(\tikztostart)!#1!\pgfkeysvalueof{/tikz/ncbar angle}:(\tikztotarget)$)
        -- ($(\tikztotarget)!($(\tikztostart)!#1!\pgfkeysvalueof{/tikz/ncbar angle}:(\tikztotarget)$)!\pgfkeysvalueof{/tikz/ncbar angle}:(\tikztostart)$)
        -- (\tikztotarget)
    },
    ncbar/.default=0.5cm,
}

\tikzset{square left brace/.style={ncbar=0.5cm}}
\tikzset{square right brace/.style={ncbar=-0.5cm}}
\tikzset{small square left brace/.style={ncbar=0.2cm}}
\tikzset{small square right brace/.style={ncbar=-0.2cm}}
\usepackage[nottoc,numbib]{tocbibind}
\usetikzlibrary{matrix}

\allowdisplaybreaks

\title{Parameterized Intractability of Even Set and\\ Shortest Vector Problem from Gap-ETH}

\author{
Arnab Bhattacharyya\thanks{Indian Institute of Science, India. Email: \texttt{arnabb@iisc.ac.in}. This work was partially supported by Ramanujan Fellowship DSTO 1358.}
\and
Suprovat Ghoshal\thanks{Indian Institute of Science, India. Email: \texttt{suprovat@iisc.ac.in}.}
\and
Karthik C.\ S.\thanks{Weizmann Institute of Science, Israel. Email: \texttt{karthik.srikanta@weizmann.ac.il}. This work was supported by  Irit Dinur's ERC-CoG grant 772839.}
\and
Pasin Manurangsi\thanks{University of California, Berkeley, USA. Email: \texttt{pasin@berkeley.edu}. Some part of this work was done while the author was visiting Indian Institute of Science and supported by the Indo-US Joint Center for Pseudorandomness in Computer Science.}
}
\date{}
\begin{document}

\maketitle
\thispagestyle{empty}

\begin{abstract}
The \emph{$k$-Even Set} problem is a parameterized variant of the \emph{Minimum Distance Problem} of linear codes over $\mathbb F_2$, which can be stated as follows: given a generator matrix $\mathbf A$ and an integer $k$, determine whether the code generated by $\mathbf A$ has distance at most $k$. Here, $k$ is the parameter of the problem. The question of whether $k$-Even Set is fixed parameter tractable (FPT) has been repeatedly raised in literature and has earned its place in Downey and Fellows' book (2013) as one of the ``most infamous'' open problems in the field of Parameterized Complexity. \vspace{0.1cm}

In this work, we show that $k$-Even Set does not admit FPT algorithms under the (randomized) Gap Exponential Time Hypothesis (Gap-ETH) [Dinur'16, Manurangsi-Raghavendra'16]. In fact, our result rules out not only exact FPT algorithms, but also any constant factor FPT approximation algorithms for the problem. Furthermore, our result holds even under the following weaker assumption, which is also known as the \emph{Parameterized Inapproximability Hypothesis } (PIH) [Lokshtanov et al.'17]: no (randomized) FPT algorithm can distinguish a satisfiable 2CSP instance from one which is only $0.99$-satisfiable (where the parameter is the number of variables).\vspace{0.1cm}

We also consider the parameterized  \emph{$k$-Shortest Vector Problem (SVP)}, in which we are given a lattice whose basis vectors are integral and an integer $k$, and the goal is to determine whether the norm of the shortest vector (in the $\ell_p$ norm for some fixed $p$) is at most $k$. Similar to $k$-Even Set, this problem is also a long-standing open problem in the field of Parameterized Complexity. We show that, for any $p > 1$, $k$-SVP is hard to approximate (in FPT time) to some constant factor, assuming PIH. Furthermore, for the case of $p = 2$, the inapproximability factor can be amplified to any constant. 
%
\end{abstract}


\newpage

\section{Introduction} \label{sec:intro}

The study of error-correcting codes gives rise to many interesting computational problems. One of the most fundamental among these is the problem of computing the distance of a linear code. In this problem, which is commonly referred to as the \emph{Minimum Distance Problem (\MDP)}, we are given as input a generator matrix $\bA \in \mathbb{F}_2^{n \times m}$ of a binary\footnote{Note that \MDP\ can be defined over larger fields as well; we discuss more about this in Section~\ref{sec:open}.} linear code and an integer $k$. The goal is to determine whether the code has distance at most $k$. Recall that the distance of a linear code is $\underset{\vzero \ne \bx \in \mathbb{F}_2^m}{\min}\ \|\bA\bx\|_0$ where $\| \cdot \|_0$ denote the 0-norm (aka the Hamming norm).

The study of this problem dates back to at least 1978 when Berlekamp \etal~\cite{BMT78} conjectured that it is \NP-hard. This conjecture remained open for almost two decades until it was positively resolved by Vardy~\cite{Var97a,Var97b}. Later, Dumer \etal~\cite{DMS03} strengthened this intractability result by showed that, even \emph{approximately} computing the minimum distance of the code is hard. Specifically, they showed that, unless $\NP = \RP$, no polynomial time algorithm can distinguish between a code with distance at most $k$ and one whose distance is greater than $\gamma \cdot k$ for any constant $\gamma \geqs 1$. Furthermore, under stronger assumptions, the ratio can be improved to superconstants and even almost polynomial. Dumer \etal's result has been subsequently derandomized by Cheng and Wan~\cite{CW12} and further simplified by Austrin and Khot~\cite{AK14} and Micciancio~\cite{Mic14}.

While the aforementioned intractability results rule out not only efficient algorithms but also efficient approximation algorithms for \MDP, there is another popular technique in coping with \NP-hardness of problems which is not yet ruled out by the known results: \emph{parameterization}.

In parameterized problems, part of the input is an integer that is designated as the parameter of the problem, and the goal is now not to find a polynomial time algorithm but a \emph{fixed parameter tractable (FPT)} algorithm. This is an algorithm whose running time can be upper bounded by some (computable) function of the parameter in addition to some polynomial in the input length. Specifically, for \MDP, its parameterized variant\footnote{Throughout Sections~\ref{sec:intro}  and \ref{sec:overview}, for a computational problem $\Pi$, we denote its parameterized variant by $k$-$\Pi$, where $k$ is the parameter of the problem.} $k$-\MDP\ has $k$ as the parameter and the question is whether there exists an algorithm that can decide if the code generated by $\bA$ has distance at most $k$ in time $T(k) \cdot \poly(mn)$ where $T$ can be any computable function that depends only on $k$.

The parameterized complexity of $k$-\MDP\ was first questioned by Downey \etal~\cite{DFVW99}\footnote{$k$-\MDP\ is formulated slightly differently in~\cite{DFVW99}. There, the input contains a parity-check matrix instead of the generator matrix, but, since we can efficiently compute one given the other, the two formulations are equivalent.}$^,$\footnote{$k$-\MDP\ is commonly referred to in the area of parameterized complexity as the \emph{$k$-Even Set} problem due to its graph theoretic interpretation (see ~\cite{DFVW99}).} who showed that parameterized variants of several other coding-theoretic problems, including the Nearest Codeword Problem and the Nearest Vector Problem\footnote{The Nearest Vector Problem is also referred to in the literature as the Closest Vector Problem.} which we will discuss in more details in Section~\ref{sec:ncp-cvp},  are $\W[1]$-hard. Thereby, assuming the widely believed $\W[1] \ne \FPT$ hypothesis, these problems are rendered intractable from the parameterized perspective. Unfortunately, Downey \etal~fell short of proving such hardness for $k$-\MDP\ and left it as an open problem: 

\begin{question}\label{openq1}
Is $k$-\MDP\ fixed parameter tractable? 
\end{question}

Although almost two decades have passed, the above question remains unresolved to this day, despite receiving significant attention from the community. In particular, the problem was listed as an open question in the seminal book of Downey and Fellows~\cite{DF99} and has been reiterated numerous times over the years~\cite{DGMS07,FGMS12,GKS12,DF13,CFJKLMPS14,CFKLMPPS15,BGGS16,CFHW17,Maj17}. In fact, in their second book~\cite{DF13}, Downey and Fellows even include this problem as one of the six\footnote{So far, two of the six problems have been resolved: that of parameterized complexity of $k$-Biclique~\cite{Lin15} and that of parameterized approximability of $k$-Dominating Set~\cite{
KLM18}.} ``most infamous'' open questions in the area of Parameterized Complexity.

Another question posted in Downey \etal's work~\cite{DFVW99} that remains open is the parameterized \emph{Shortest Vector Problem ($k$-\SVP)} in lattices. The input of $k$-\SVP\ (in the $\ell_p$ norm) is an integer $k \in \N$ and a matrix $\bA \in \Z^{n \times m}$ representing the basis of a lattice, and we want to determine whether the shortest (non-zero) vector in the lattice has length at most $k$, i.e., whether $\underset{\vzero \ne \bx \in \Z^{m}}{\min}\ \|\bA\bx\|_p \leqs k$. Again, $k$ is the parameter of the problem. It should also be noted here that, similar to~\cite{DFVW99}, we require the basis of the lattice to be integer-value, which is sometimes not enforced in literature (e.g.~\cite{VEB,Ajt98}). This is because, if $\bA$ is allowed to be any matrix in $\mathbb{R}^{n \times m}$, then parameterization is meaningless because we can simply scale $\bA$ down by a large multiplicative factor.

The (non-parameterized) Shortest Vector Problem (\SVP) has been intensively studied, motivated partly due to the fact that both algorithms and hardness results for the problem have numerous applications. Specifically, the celebrated LLL algorithm for \SVP~\cite{LLL} can be used to factor rational polynomials, and to solve integer programming (parameterized by the number of unknowns)~\cite{Len83} and many other computational number-theoretic problems (see e.g.~\cite{NV10}). Furthermore, the hardness of (approximating) \SVP\ has been used as the basis of several cryptographic constructions~\cite{Ajt98,AD97,Reg03,Reg05}. Since these topics are out of scope of our paper, we refer the interested readers to the following surveys for more details:~\cite{Reg06,MR-survey,NV10,Reg10}.

On the computational hardness side of the problem, van Emde-Boas~\cite{VEB} was the first to show that \SVP\ is \NP-hard for the $\ell_{\infty}$ norm, but left open the question of whether \SVP\ on the $\ell_p$ norm for $1 \leqs p < \infty$ is \NP-hard. It was not until a decade and a half later that Ajtai~\cite{Ajt96} showed, under a randomized reduction, that \SVP\ for the $\ell_2$ norm is also \NP-hard; in fact, Ajtai's hardness result holds not only for exact algorithms but also for $(1 + o(1))$-approximation algorithms as well. The $o(1)$ term in the inapproximability ratio was then improved in a subsequent work of Cai and Nerurkar~\cite{CN99}. Finally, Micciancio~\cite{Mic00} managed to achieve a factor that is bounded away from one. Specifically, Micciancio~\cite{Mic00} showed (again under randomized reductions) that \SVP\ on the $\ell_p$ norm is \NP-hard to approximate to within a factor of $\sqrt[p]{2}$ for every $1 \leqs p < \infty$. Khot~\cite{Khot05} later improved the ratio to any constant, and even to $2^{\log^{1/2 - \varepsilon}(nm)}$ under a stronger assumption. Haviv and Regev~\cite{HR07} subsequently simplified the gap amplification step of Khot and, in the process, improved the ratio to almost polynomial. We note that both Khot's and Haviv-Regev reductions are also randomized and it is still open to find a deterministic \NP-hardness reduction for \SVP\ in the $\ell_p$ norms for $1 \leqs p < \infty$ (see~\cite{Mic12}); we emphasize here that such a reduction is not known even for the exact (not approximate) version of the problem. For the $\ell_\infty$ norm, the following stronger result due to Dinur is known~\cite{Din02}: \SVP\ in the $\ell_\infty$ norm is \NP-hard to approximate to within $n^{\Omega(1/\log \log n)}$ factor (under a \emph{deterministic} reduction).

Very recently, fine-grained studies of \SVP\ have been initiated~\cite{BGS17,AD17}. The authors of~\cite{BGS17,AD17} showed that \SVP\ for any $\ell_p$ norm cannot be solved (or even approximated to some constant strictly greater than one) in subexponential time assuming the existence of a certain family of lattices\footnote{This additional assumption is only needed for $1 \leqs p \leqs 2$. For $p > 2$, their hardness is conditional only on Gap-ETH.} and the (randomized) \emph{Gap Exponential Time Hypothesis (Gap-ETH)}~\cite{D16,MR16}, which states that no randomized subexponential time algorithm can distinguish between a satisfiable 3CNF formula and one which is only 0.99-satisfiable. 
(See Hypothesis~\ref{hyp:gap-eth}.)

As with \MDP, Downey \etal~\cite{DFVW99} were the first to question the parameterized tractability of $k$-\SVP\ (for the $\ell_2$ norm). Once again, Downey and Fellows included $k$-\SVP\ as one of the open problems in both of their books~\cite{DF99,DF13}, albeit, in their second book, $k$-\SVP\ was in the ``tough customers'' list instead of the ``most infamous'' list that $k$-\MDP\ belonged to. And again, as with Open Question~\ref{openq1}, this question remains unresolved to this day:

\begin{question}\label{openq2}
Is $k$-\SVP\ fixed parameter tractable? 
\end{question}

\subsection{Our Results}\label{sec:results}

The main result of this paper is a resolution to the previously mentioned Open~Question~\ref{openq1}~and~\ref{openq2}: more specifically, we prove that $k$-\MDP\ and $k$-\SVP\  (on $\ell_p$ norm for any $p > 1$) do not admit any FPT algorithm, assuming the aforementioned (randomized) Gap-ETH (Hypothesis~\ref{hyp:gap-eth}). In fact, our result is slightly stronger than stated here in a couple of ways:
\begin{enumerate}[(1)]
\item First, our result rules out not only exact FPT algorithms but also FPT approximation algorithms as well.
\item Second, our result works even under the so-called \emph{Parameterized Inapproximability Hypothesis (PIH)}~\cite{LRSZ17}, which asserts that no (randomized) FPT algorithm can distinguish between a satisfiable 2CSP instance and one which is only 0.99-satisfiable, where the parameter is the number of variables (See Hypothesis~\ref{hyp:pih}). It is known (and simple to see) that Gap-ETH implies PIH; please refer to Section~\ref{subsec:pih}, for more details regarding the two assumptions.
\end{enumerate}

With this in mind, we can state our results starting with the parameterized intractability of $k$-\MDP, more concretely (but still informally), as follows:

\begin{theorem}[Informal; see Theorem~\ref{thm:MDPmain}]
Assuming PIH, for any $\gamma \geqs 1$ and any computable function $T$, no $T(k) \cdot \poly(nm)$-time algorithm, on input $(\bA, k) \in \F^{n \times m} \times \N$, can distinguish between
\begin{itemize}
\item the distance of the code generated by $\bA$ is at most $k$  $\left(\text{i.e., } \underset{\vzero \ne \bx \in \F^m}{\min}\ \|\bA\bx\|_0 \leqs k\right)$, and,
\item the distance of the code generated by $\bA$ is more than $\gamma \cdot k$  $\left(\text{i.e., }\underset{\vzero \ne \bx \in \F^m}{\min}\ \|\bA\bx\|_0 > \gamma \cdot k\right)$.
\end{itemize} 
\end{theorem}


Notice that our above result rules out FPT approximation algorithms with \emph{any} constant approximation ratio for $k$-\MDP. In contrast, we can only prove FPT inapproximability with \emph{some} constant ratio for $k$-\SVP\ in $\ell_p$ norm for $p > 1$, with the exception of $p = 2$ for which the inapproximability factor in our result can be amplified to any constant. These are stated more precisely below.

\begin{theorem}[Informal; see Theorem~\ref{thm:SVPmain}]
For any $p > 1$, there exists a constant $\gamma_p > 1$ such that, assuming PIH, for any computable function $T$, no $T(k) \cdot \poly(nm)$-time algorithm, on input $(\bA, k) \in \Z^{n \times m} \times \N$, can distinguish between
\begin{itemize}
\item the $\ell_p$ norm of the shortest vector of the lattice generated by $\bA$ is $\leqs k$  $\left(\text{i.e., }\underset{\vzero \ne \bx \in \Z^m}{\min}\ \|\bA\bx\|_p \leqs k\right)$, and,
\item the $\ell_p$ norm of the shortest vector of the lattice generated by $\bA$ is $> \gamma_p \cdot k$  $\left(\text{i.e., }\underset{\vzero \ne \bx \in \Z^m}{\min}\ \|\bA\bx\|_p > \gamma_p \cdot k\right)$.
\end{itemize} 
\end{theorem}


\begin{theorem}[Informal; see Theorem~\ref{thm:SVPmain-boost}]
	Assuming PIH, for any computable function $T$ and constant $\gamma \geqs 1$, no $T(k) \cdot \poly(nm)$-time algorithm, on input $(\bA, k) \in \Z^{n \times m} \times \N$, can distinguish between
	\begin{itemize}
		\item the $\ell_2$ norm of the shortest vector of the lattice generated by $\bA$ is $\leqs k$  $\left(\text{i.e., }\underset{\vzero \ne \bx \in \Z^m}{\min}\ \|\bA\bx\|_2 \leqs k\right)$, and,
		\item the $\ell_2$ norm of the shortest vector of the lattice generated by $\bA$ is $> \gamma \cdot k$  $\left(\text{i.e., }\underset{\vzero \ne \bx \in \Z^m}{\min}\ \|\bA\bx\|_2 > \gamma \cdot k\right)$.
	\end{itemize} 
\end{theorem}

We remark that our results do not yield hardness for \SVP\ in the $\ell_1$ norm and this remains an interesting open question. Section~\ref{sec:open} contains discussion on this problem. We also note that, for Theorem~\ref{thm:SVPmain} and onwards, we are only concerned with $p \ne \infty$; this is because, for $p = \infty$, the problem is \NP-hard to approximate even when $k = 1$~\cite{VEB}!

\subsubsection{Nearest Codeword Problem and Nearest Vector Problem}
\label{sec:ncp-cvp}

As we shall see in Section~\ref{sec:overview}, our proof proceeds by first showing FPT hardness of approximation of the non-homogeneous variants of $k$-\MDP\ and $k$-\SVP\ called the \emph{$k$-Nearest Codeword Problem} ($k$-\NCP) and the \emph{$k$-Nearest Vector Problem } ($k$-\CVP) respectively. For both $k$-\NCP\ and $k$-\CVP, we are given a target vector $\by$ (in $\F^n$ and $\Z^n$, respectively) in addition to $(\bA, k)$, and the goal is to find whether there is any $\bx$ (in $\F^m$ and $\Z^m$, respectively) such that the (Hamming and $\ell_p$, respectively) norm of $\bA\bx - \by$ is at most $k$. 

As an intermediate step of our proof, we show that the $k$-\NCP\ and $k$-\CVP\ problems are hard to approximate\footnote{While our $k$-\MDP\ result only applies for $\F$, it is not hard to see that our intermediate reduction for $k$-\NCP\ actually applies for any finite field $\mathbb{F}_q$ too.} (see Theorem~\ref{thm:MLDmain} and Theorem~\ref{thm:CVPmain} respectively). This should be compared to~\cite{DFVW99}, in which the authors show that both problems are $\W[1]$-hard. The distinction here is that our result rules out not only exact algorithms but also approximation algorithms, at the expense of the stronger assumption than that of~\cite{DFVW99}. Indeed, if one could somehow show that $k$-\NCP\ and $k$-\CVP\ are $\W[1]$-hard to approximate (to some constant strictly greater than one), then our reduction would imply $\W[1]$-hardness of $k$-\MDP\ and $k$-\SVP\ (under randomized reductions). Unfortunately, no such $\W[1]$-hardness of approximation of $k$-\NCP\ and $k$-\CVP\ is known yet.

We end this section by remarking that the computational complexity of both (non-parameterized) \NCP\ and \CVP\ are also thoroughly studied (see e.g.~\cite{Mic01,DKRS03,Ste93,ABSS97,GMSS99} in addition to the references for \MDP\ and \SVP), and indeed the inapproximability results of these two problems form the basis of hardness of approximation for \MDP\ and \SVP.

\subsection{Organization of the paper}

In the next section, we give an overview of our reductions and proofs. After that, in Section~\ref{sec:prelim}, we define additional notations and preliminaries needed to fully formalize our proofs. In Section~\ref{sec:csp-to-gapvec}  we show the constant inapproximability of $k$-\NCP. Next, in Section~\ref{sec:main-reduction}, we establish the constant inapproximability of $k$-\MDP. Section~\ref{sec:svp} provides the constant inapproximability of $k$-\CVP\ and $k$-\SVP. Finally, in Section~\ref{sec:open}, we conclude with a few open questions and research directions.

\section{Proof Overview}
\label{sec:overview}

In the non-parameterized setting, all the aforementioned inapproximability results for both \MDP\ and \SVP\ are shown in two steps: first, one proves the inapproximability of their inhomogeneous counterparts (i.e. \NCP\ and \CVP), and then reduces them to \MDP\ and \SVP. We follow this general outline. That is, we first show, via relatively simple reductions from PIH, that both $k$-\NCP\ and $k$-\CVP\ are hard to approximate. Then, we reduce $k$-\NCP\ and $k$-\CVP\ to $k$-\MDP\ and $k$-\SVP\ respectively. In this second step, we employ Dumer \etal's reduction~\cite{DMS03} for $k$-\MDP\ and Khot's reduction~\cite{Khot05} for $k$-\SVP. While the latter works almost immediately in the parameterized regime, there are several technical challenges in adapting Dumer \etal's reduction to our setting. The remainder of this section is devoted to presenting all of our reductions and to highlight such technical challenges and changes in comparison with the non-parameterized settings.

The starting point of all the hardness results in this paper is Gap-ETH (Hypothesis~\ref{hyp:gap-eth}). As mentioned earlier, it is well-known that Gap-ETH implies PIH (Hypothesis~\ref{hyp:pih}), i.e., PIH is weaker than Gap-ETH. Hence, for the rest of this section, we may start from PIH instead of Gap-ETH.

\subsection{Parameterized Intractability of $k$-\MDP\ from PIH}

We start this subsection by describing the Dumer \etal's (henceforth DMS) reduction~\cite{DMS03}. The starting point of the DMS reduction is the \NP-hardness of approximating \NCP\ to any constant factor \cite{ABSS97}. Let us recall that in \NCP\ we are given a matrix $\bA\in\F^{n\times m}$, an integer $k$, and a target vector $\by\in\F^n$, and the goal is to determine whether there is any $\bx\in\F^m$ such that $\|\bA\bx - \by\|_0$ is at most $k$.  Arora et al. \cite{ABSS97} shows that for any constant $\gamma\geqs 1$, it is \NP-hard to distinguish the case when there exists $\bx$ such that $\|\bA\bx - \by\|_0\leqs k$ from the case when for all $\bx$ we have that $\|\bA\bx - \by\|_0> \gamma k$. Dumer at al.\ introduce the notion of ``locally dense codes'' to enable a gadget reduction from $\NCP$ to $\MDP$. Informally, a locally dense code is a linear code $\bL$ with minimum distance $d$
admitting a ball $\cB(\bs, r)$ centered at $\bs$ of radius\footnote{Note that for the ball to contain more than a single codeword, we must have $r\geqs \nicefrac{d}{2}$.} $r < d$ and containing a large (exponential in the dimension) number of codewords. Moreover, for the gadget reduction to \MDP\ to go through, we require not only the knowledge of the code, but also the center $\bs$ and  a linear transformation $\bT$ used to index the codewords in $\cB(\bs, r)$, i.e., $\bT$ maps $\cB(\bs, r)\cap \bL$ onto a smaller subspace. Given an instance $(\bA,\by,k)$ of \NCP, and a  locally dense code $(\bL,\bT,\bs)$ whose parameters (such as dimension and distance) we will fix later, Dumer et al.\ build the following matrix:
\begin{align}
\bB =
\begin{bmatrix}
\bA \bT \bL & - \by \\
\vdots & \vdots \\
\bA \bT \bL & - \by \\
\bL & - \bs \\
\vdots & \vdots \\
\bL & - \bs
\end{bmatrix}
\begin{tikzpicture}[baseline=11ex]
\draw [decorate,decoration={brace,amplitude=10pt}]
(-1.5,1.8) -- (-1.5,0) node [black,midway,xshift=35pt]{$b$ copies};
\draw [decorate,decoration={brace,amplitude=10pt}]
(-1.5,4.0) -- (-1.5,2.2) node [black,midway,xshift=35pt]{$a$ copies};
\end{tikzpicture}
,
\label{eq:LDCGadget}
\end{align}

where $a,b$ are some appropriately chosen positive integers. If there exists $\bx$ such that $\|\bA\bx-\by\|_0\leqs k$ then consider $\bz'$ such that $\bT\bL \bz' = \bx$ (we choose the parameters of $(\bL,\bT,\bs)$, in particular the dimensions of $\bL$ and $\bT$ such that all these computations are valid). Let $\bz=\bz'\circ 1$, and note that $\|\bB\bz\|_0=a\|\bA\bx-\by\|_0+b\|\bL\bz-\bs\|_0\leqs ak+br$. In other words, if $(\bA,\by,k)$ is a YES instance of \NCP\ then $(\bB,ak+br)$ is a YES instance of \MDP. On the other hand if we had that for all $\bx$, the norm of $\|\bA\bx-\by\|_0$ is more than $\gamma k$ for some constant\footnote{Note that in the described reduction, we need the inapproximability  of \NCP\ to a factor greater than two, even to just reduce to the \emph{exact} version of \MDP.} $\gamma>2$, then it is possible to show that for all $\bz$ we have that $\|\bB\bz\|_0>\gamma' (ak+br)$ for any $\gamma'<\frac{2\gamma}{2+\gamma}$. The proof is based on a case analysis of the last coordinate of $\bz$. If that coordinate is 0, then, since $\bL$ is a code of distance $d$, we have $\|\bB\bz\|_0 \geqs bd > \gamma'(ak+br)$; if that coordinate is 1, then the assumption that $(\bA,\by,k)$ is a NO instance of \NCP\ implies that $\|\bB\bz\|_0 > a\gamma k >\gamma'(ak+br)$. Note that this gives an inapproximability for \MDP\ of ratio $\gamma' < 2$; this gap is then further amplified by a simple tensoring procedure. 


We note that Dumer at al.\ were not able to find a deterministic construction of locally dense code with all of the above described properties. Specifically, they gave an efficient deterministic construction of a code $\bL$, but only gave a randomized algorithm that finds a linear transformation $\bT$ and a center $\bs$ w.h.p. Therefore, their hardness result relies on the assumption that $\NP\neq\RP$, instead of the more standard $\NP\neq \P$ assumption. Later, Cheng and Wan \cite{CW12} and Micciancio \cite{Mic14} provided constructions for such (families of) locally dense codes with an explicit center, and thus showed the constant ratio inapproximability of \MDP\ under the assumption of $\NP\neq\P$.

Trying to follow the DMS reduction in order to show the parameterized intractability of $k$-\MDP, we face the following three immediate obstacles. First, there is no inapproximability result known for  $k$-\NCP, for any constant factor greater than 1. Note that to use the DMS reduction, we need the parameterized inapproximability of $k$-\NCP, for an approximation factor which is greater than two. Second, the construction of locally dense codes of Dumer et al. only works when the distance is linear in the block length (which is a function of the size of the input). However, we need codes whose distance are bounded above by a function of the parameter of the problem (and not depend on the input size). This is because the DMS reduction converts an instance $(\bA,\by,k)$ of $k$-\NCP\ to an instance $(\bB,ak+br)$ of $(ak+br)$-\MDP, and for this reduction to be an FPT reduction, we need $ak+br$ to be a function only depending on $k$, i.e., $d$, the distance of the code $\bL$ (which is at most $2r$), must be a function only of $k$. Third, recall that the DMS reduction needs to identify the vectors in the ball $\cB(\bs, r) \cap \bL$ with all the potential solutions of $k$-\NCP. Notice that the number of vectors in the ball is at most $(nm)^{O(r)}$ but the number of potential solutions of $k$-\NCP\ is exponential in $m$ (i.e. all $\bx \in \F^m$). However, this is impossible since $r \leqs d$ is bounded above by a function of $k$!


We overcome the first obstacle by proving the constant inapproximability of $k$-\NCP\ under PIH. Specifically, assuming PIH, we first show the parameterized inapproximability of $k$-\NCP\ for some constant factor greater than 1, and then boost the gap using a composition operator (self-recursively). Note that in order to follow the DMS reduction, we need the inapproximability of $k$-\NCP\ for some constant factor greater than 2; in other words, the gap amplification for $k$-\NCP\ is necessary, even if we are not interested in showing the inapproximability of $k$-\NCP\ for all constant factors.

We overcome the third obstacle by introducing an intermediate problem in the DMS reduction, which we call the \emph{sparse nearest codeword problem}. The sparse nearest codeword problem is a promise problem which differs from  $k$-\NCP\ in only one way: in the YES case, we want to find $\bx\in \cB(\mathbf{0},k)$ (instead of the entire space $\F^m$), such that $\|\bA\bx-\by\|_0\leqs k$. In other words, we only allow sparse $\bx$ as a solution. We show that $k$-\NCP\ can be reduced to the sparse nearest codeword problem.

Finally, we overcome the second obstacle by introducing a variant of locally dense codes, which we call \emph{sparse covering codes}. Roughly speaking, we show that any code which nears the sphere-packing bound (aka Hamming bound) in the high rate regime is a sparse covering code. Then we follow the DMS reduction with the new ingredient of sparse covering codes (replacing locally dense codes) to reduce the sparse nearest codeword problem to $k$-\MDP.

We remark here that overcoming the second and third obstacles are our main technical contributions. In particular, our result on sparse covering codes might be of independent interest.

The full reduction goes through several intermediate steps, which we will describe in more detail in the coming paragraphs. The high-level summary of these steps is also provided in Figure~\ref{fig:overviewMDC}. Throughout this section, for any gap problem, if we do not specify the gap in the subscript, then it implies that the gap can be any arbitrary constant.
For every $\varepsilon\geqs 0$, we denote by $\csp_\varepsilon$ the gap problem where we have to determine if a given 2CSP instance $\Gamma$, i.e., a graph $G = (V, E)$ and a set of constraints  $\{C_{uv}\}_{(u, v) \in E}$ over an alphabet set $\Sigma$,  has an assignment to its vertices that satisfies all the constraints or if every assignment violates more than $\varepsilon$ fraction of the constraints. Here each $C_{uv}$ is simply the set of all $(\sigma_u, \sigma_v) \in \Sigma \times \Sigma$ that satisfy the constraint. The parameter of the problem is $|V|$. PIH asserts that there exists some constant $\varepsilon>0$ such that  no randomized FPT algorithm can solve $\csp_{\varepsilon}$. (See Hypothesis~\ref{hyp:pih} for a more formal statement.)

\begin{figure}[h!]
    \centering
    \resizebox{\textwidth}{!}{\begin{tikzpicture}

\node (gapeth) [draw=red!80!black,thick] at (0,0) {Gap-ETH$_{\delta}$};

\node (PIH) [draw=red!80!black,thick] at (5, 0) {PIH$_{ \varepsilon}$};

\node (gapvec) [draw=red!80!black,thick] at (9, 0) {\gapvec$_{ \gamma}$};
\node (ampgapvec) [draw=red!80!black,thick] at (14, 0) {\gapvec};

\node (gapsnc) [draw=red!80!black,thick] at (1.5, -3) {\sncp};

\node (gapmdc) [draw=red!80!black,thick] at (7.5, -3) {\mdp$_{1.99}$};
\node (ampgapmdc) [draw=red!80!black,thick] at (13, -3) {\mdp};

\draw [-{Latex[length=1.5mm, width=1.5mm]}] (gapeth) -- (PIH);
\draw [-{Latex[length=1.5mm, width=1.5mm]}] (PIH) -- (gapvec);
\draw [-{Latex[length=1.5mm, width=1.5mm]}] (gapvec) -- (ampgapvec);
\draw [-{Latex[length=1.5mm, width=1.5mm]}] (ampgapvec) -- (gapsnc);
\draw [-{Latex[length=1.5mm, width=1.5mm]}] (gapsnc) -- (gapmdc);
\draw [-{Latex[length=1.5mm, width=1.5mm]}] (gapmdc) -- (ampgapmdc);

\node [above, align=center] at (2.75, 0.5) {\footnotesize Folklore Reduction};
\node [above, align=center] at (2.75, 0) {\footnotesize (Section~\ref{sec:gapETHtoPIH})};

\node [above, align=center] at (6.7, 0) {\footnotesize Section~\ref{sec:base-reduction}};
\node [above, align=center,rotate=12] at (7.5, -1.52) {\footnotesize Section~\ref{sec:gapvec-to-snc}};

\node [above, align=center] at (11.55,0.45) {\footnotesize Gap Amplification};
\node [above, align=center] at (11.55, 0) {\footnotesize (Section~\ref{sec:gap-amplification})};
\node [above, align=center] at (4.3, -3.6) {\footnotesize Introducing \SCC to};
\node [above, align=center] at (4.31, -4) {\footnotesize DMS reduction};
\node [above, align=center] at (4.3, -4.45) {\footnotesize (Sections~\ref{sec:dense-code} and \ref{sec:main-red})};

\node [above, align=center] at (10.4,-3.6) {\footnotesize Gap Amplification};
\node [above, align=center] at (10.4, -4.05) {\footnotesize (Proposition~\ref{prop:gap-amplification})};

\end{tikzpicture}}
    \caption{An overview of the reduction from Gap-ETH to the parameterized Minimum Distance problem. First, PIH follows from Gap-ETH due to known reductions (see Section~\ref{sec:gapETHtoPIH}). 
Next, we reduce $\csp_\varepsilon$ to $\gapvec_\gamma$ in Section~\ref{sec:base-reduction}, and then amplify the gap using a composition operator in Section~\ref{sec:gap-amplification}. Via a simple reduction from \gapvec, in Section~\ref{sec:gapvec-to-snc} we obtain the constant parameterized inapproximability of \sncp. In Section~\ref{sec:dense-code}, we formally introduce sparse covering codes and show how to efficiently (but probabilistically) construct them. These codes are then used in Section~\ref{sec:main-reduction} to obtain the parameterized innapproximability of $\mdp_{1.99}$. The final step is a known gap amplification by tensoring (Proposition~\ref{prop:gap-amplification}).
} \label{fig:overviewMDC}
\end{figure}
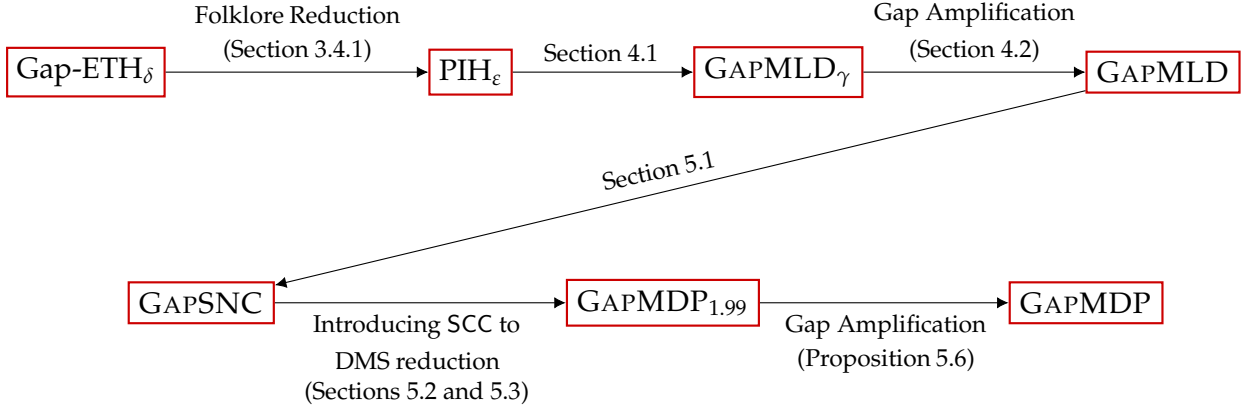

\noindent\textbf{ Reducing $\csp_\epsilon$ to $\gapvec_{\gamma}$.}
 We start by showing the parameterized inapproximability of $k$-\NCP\ for some constant ratio. Instead of working with $k$-\NCP, we work with its equivalent formulation (by converting the generator matrix given as input into a parity-check matrix) which  in the literature is referred to as the \emph{maximum likelihood decoding problem}\footnote{The two formulations are equivalent but we use different names for them to avoid confusion when we use \emph{Sparse Nearest Codeword Problem} later on.}. 
We define the gap version of this problem (i.e., a promise problem), denoted by $\gapvec_{\gamma}$ (for some constant $\gamma\geqs 1$) as follows: on input $(\bA,\by,k)$, distinguish between the YES case where there exists $\bx\in \cB(\mathbf{0},k)$ such that $\bA\bx=\by$, and the NO case where for all $\bx\in \cB(\mathbf{0},\gamma k)$ we have $\bA\bx\neq\by$. It is good to keep in mind that this is equivalent to asking whether there exist $k$ columns of $\bA$ whose sum is equal to $\by$ or whether any $\leqs \gamma k$ columns of $\bA$ do not sum up to $\by$.

Now, we will sketch the reduction from an instance $\Gamma=(G=(V,E),\Sigma,\{C_{uv}\}_{(u, v) \in E})$ of $\csp_{\varepsilon}$ to an instance $(\bA,\by,k)$ of $\gapvec_{1+\varepsilon/3}$. The matrix $\bA$ will have $|V||\Sigma|+\underset{(u,v)\in E}{\sum}\ |C_{uv}|$ columns and $|V|+|E|+2|E||\Sigma|$ rows. The first $|V||\Sigma|$ columns of $\bA$ are labelled with $(u, \sigma_u) \in V \times \Sigma$, and the remaining columns of $\bA$ are labeled by $(e, \sigma_u, \sigma_v)$ where $e = (u, v) \in E$ and $(\sigma_u, \sigma_v) \in C_{uv}$.

Before we continue with our description of $\bA$, let us note that, in the YES case where there is an assignment $\phi: V \to \Sigma$ that satisfies every constraint, our intended solution for our $\gapvec$ instance is to pick the $(u, \phi(u))$-column for every $u \in V$ and the $((u, v), \phi(u), \phi(v))$-column for every $(u, v) \in E$. Notice that $|V| + |E|$ columns are picked, and indeed we set $k = |V| + |E|$. Moreover, we set the first $|V| + |E|$ coordinates of $\by$ to be one and the rest to be zero.

We also identify the first $|V|$ rows of $\bA$ with $u \in V$, the next $|E|$ rows of $\bA$ with $e \in E$, and the remaining $2|E||\Sigma|$ rows of $\bA$ with $(e, \sigma, b) \in E \times \Sigma \times \{0, 1\}$. Figure~\ref{fig:mld-matrix} provides an illustration of the matrix $\bA$. The rows of $\bA$ will be designed to serve the following purposes: the first $|V|$ rows will ensure that, for each $u \in V$, at least one column of the form $(u, \cdot)$ is picked, the next $|E|$ rows will ensure that, for each $e \in E$, at least one column of the form $(e, \cdot, \cdot)$ is picked, and finally the remaining $2|E||\Sigma|$ rows will ``check'' that the constraint is indeed satisfied.

Specifically, each $u$-row for $u \in V$ has only $|\Sigma|$ non-zero entries: those in column $(u, \sigma_u)$ for all $\sigma_u \in \Sigma$. Since our target vector $\by$ has $\by_u = 1$, we indeed have that at least one column of the form $(u, \cdot)$ must be selected for every $u \in V$. Similarly, each $e$-row for $e = (u, v) \in E$ has $|C_{uv}|$ non-zero entries: those in column $(e, \sigma_u, \sigma_v)$ for all $(\sigma_u, \sigma_v) \in C_{uv}$. Again, these make sure that at least one column of the form $(e, \cdot, \cdot)$ must be picked for every $e \in E$.

Finally, we will define the entries of the last $2|E||\Sigma|$ rows. To do so, let us recall that, in the YES case, we pick the columns $(u, \phi(u))$ for all $u \in V$ and $((u, v), \phi(u), \phi(v))$ for all $(u, v) \in E$. The goal of these remaining rows is to not only accept such a solution but also prevent any solution that picks columns $(u, \sigma_u), (v, \sigma_v)$ and $((u, v), \sigma'_u, \sigma'_v)$ where $\sigma_u \ne \sigma'_u$ or $\sigma_v \ne \sigma'_v$. In other words, these rows serve as a ``consistency checker'' of the solution. Specifically, the $|\Sigma|$ rows of the form $((u, v), \cdot, 0)$ will force $\sigma_u$ and $\sigma'_u$ to be equal whereas the $|\Sigma|$ rows of the form $((u, v), \cdot, 1)$ will force $\sigma_v$ and $\sigma'_v$ to be equal. For convenience, we will only define the entries for the $((u, v), \cdot, 0)$-rows; the $((u, v), \cdot, 1)$-rows can be defined similarly. Each $((u, v), \sigma, 0)$-row has only one non-zero entry within the first $|V||\Sigma|$ rows: the one in the $(u, \sigma)$-column. For the remaining columns, the entry in the $((u, v), \sigma, 0)$-row and the $(e, \sigma_0, \sigma_1)$-column is non-zero if and only if $e = (u, v)$ and $\sigma_0 = \sigma$.

It should be clear from the definition that our intended solution for the YES case is indeed a valid solution because, for each $((u, v), \phi(u), 0)$-row, the two non-zero entries from the columns $(u, \phi(u))$ and $((u, v), \phi(u), \phi(v))$ cancel each other out. On the other hand, for the NO case, the main observation is that, for each edge $(u, v) \in E$, if only one column of the form $(u, \cdot)$, one of the form $(v, \cdot)$ and one of the form $((u, v), \cdot, \cdot)$ are picked, then the assignment corresponding to the picked columns satisfy the constraint $C_{uv}$. In particular, it is easy to argue that, if we can pick $(1 + \varepsilon/3)(|V| + |E|)$ columns that sum up to $\by$, then all but $\varepsilon$ fraction of all constraints fulfill the previous conditions, meaning that we can find an assignment that satisfies $1 - \varepsilon$ fraction of the constraints. Thus, we have also proved the soundness of the reduction.

\begin{figure}
\begin{align*}
\ \ \ \ \ \ \ \ \ \ \bA=\begin{tikzpicture}[baseline=10ex]
\fill [color=black!10] (1.7,2.9) rectangle (5,4);
\fill [color=black!10] (1.7,2.9) rectangle (-0.5,1.7);
\fill [pattern=north west lines, pattern color=blue] (1.7,2.9) rectangle (-0.5, 4);
\fill [pattern=north west lines, pattern color=blue] (1.7,2.9) rectangle (5, 1.7);
\fill [pattern=bricks, pattern color=orange!70!black] (1.7,0) rectangle (-0.5,1.65);
\fill [pattern=fivepointed stars, pattern color=red!50] (1.7,0) rectangle (5, 1.65);
\draw [ thick] (0,0) to [square left brace ] (0,4);
\draw [ thick] (4.5,0) to [square right brace ] (4.5,4);
\draw [decorate,decoration={brace,amplitude=10pt},rotate=90] (4.15,0.5)--
(4.15,-1.7) node [black,midway,yshift=17pt] { $|V|\times |\Sigma|$};
\draw [decorate,decoration={brace,amplitude=10pt},rotate=90](4.15,-1.7)--
(4.15,-5)  node [black,midway,yshift=20pt] { $\underset{(u, v) \in E}{\sum}|C_{uv}|$};
\draw [decorate,decoration={brace,amplitude=10pt}]
(5.15,1.7) -- (5.15,0) node [black,midway,xshift=55pt] { $|E|\times |\Sigma|\times \{0,1\}$};
\draw [decorate,decoration={brace,amplitude=9pt}]
(5.15,2.9) -- (5.15,1.7) node [black,midway,xshift=20pt] { $|E|$};
\draw [decorate,decoration={brace,amplitude=8pt}]
(5.15,4) -- (5.15,2.9) node [black,midway,xshift=20pt] { $|V|$ };
\draw [thick, dotted] (-0.5,2.9) -- (5,2.9);
\draw [thick, dotted] (-0.5,1.7) -- (5,1.7);
\draw [thick, dashed] (1.7,0) -- (1.7,4);
\end{tikzpicture},
\by = \begin{tikzpicture}[baseline=10ex]
\draw [ thick] (0,0) to [small square left brace ] (0,4);
\draw [ thick] (0.5,0) to [small square right brace ] (0.5,4);
\draw [decorate,decoration={brace,amplitude=10pt}]
(0.85,1.7) -- (0.85,0) node [black,midway,xshift=55pt] { $|E|\times |\Sigma|\times \{0,1\}$};
\draw [decorate,decoration={brace,amplitude=10pt}] (0.85,4)--
(0.85,1.7) node [black,midway,xshift=35pt] { $|V| + |E|$};
\node at (0.25, 3.75) {$1$};
\node at (0.25, 3) {$\vdots$};
\node at (0.25, 2.10) {$1$};
\node at (0.25, 1.5) {$0$};
\node at (0.25, 1) {$\vdots$};
\node at (0.25, 0.3) {$0$};  
\end{tikzpicture}
\end{align*}
\caption{An illustration of $\bA$ and $\by$. All entries in the shaded area are zero. Each \emph{row} in the brick pattern area has one non-zero entry in that area, and each \emph{column} in the star pattern area has two non-zero entries in the area. Finally, each column has one non-zero entry in the lines pattern area.}
\label{fig:mld-matrix}
\end{figure}
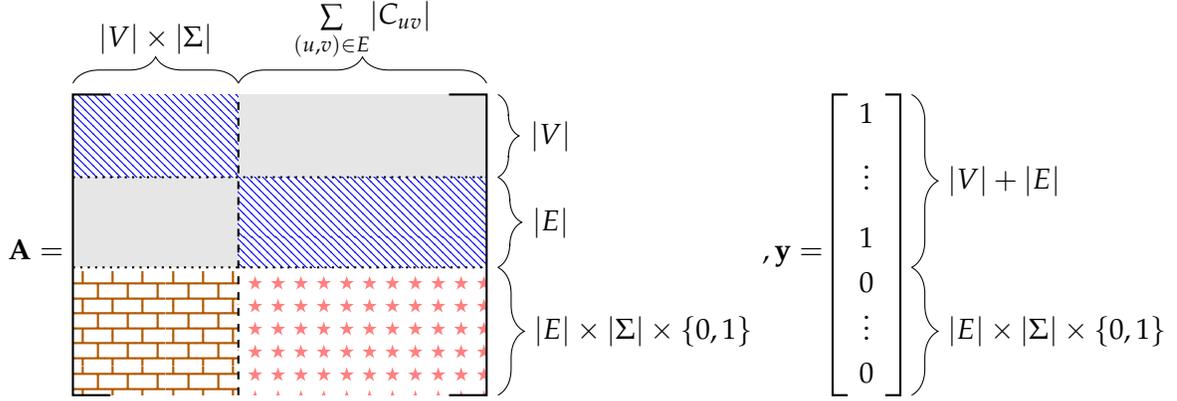

\noindent\textbf{Gap Amplification for $\gapvec_{\gamma}$.} We have sketched the proof of the hardness of $\gapvec_\gamma$ for \emph{some} constant $\gamma\geqs 1$, assuming PIH. The next step is to amplify the gap and arrive at the hardness for $\gapvec_\gamma$ for \emph{every} constant $\gamma \geqs 1$. To do so, we define an operator $\oplus$ over every pair of instances of $\gapvec_\gamma$ with the following property: if two instances $(\bA_1,\by_1,k_1)$ and $(\bA_2,\by_2,k_2)$ are both YES instances, then $(\bA,\by,k):=(\bA_1,\by_1,k_1)\oplus(\bA_2,\by_2,k_2)$ is a YES instance for $\gapvec_{\gamma'}$ where $\gamma' \approx \gamma^2$. On the other hand, if both $(\bA_1,\by_1,k_1)$ and $(\bA_2,\by_2,k_2)$ are NO instances, then $(\bA,\by,k)$ is a NO instance for $\gapvec_{\gamma'}$. Hence, we can apply $\oplus$ repeatedly to the $\gapvec_\gamma$ instance from the previous step (with itself) and amplify the gap to be any arbitrarily large constant. The definition of $\oplus$, while simple, is slightly tedious to formalized and we defer it to Section~\ref{sec:gap-amplification}.

\noindent\textbf{Reducing $\gapvec$ to $\sncp$.}
Now we introduce the \emph{sparse nearest codeword problem} that we had briefly talked about. We define the gap version of this problem, denoted by $\sncp_{\gamma}$ (for some constant $\gamma\geqs 1$) as follows: on input $(\bA',\by',k)$, distinguish between the YES case where there exists $\bx\in \cB(\mathbf{0},k)$ such that $\|\bA'\bx-\by'\|_0\leqs k$, and the NO case where for all $\bx$ (in the entire space), we have $\|\bA'\bx-\by'\|_0>\gamma k$. We highlight that the difference between $k$-\NCP\ and $\sncp_{\gamma}$ is that, in the YES case of the latter, we are promised that $\bx\in\cB(\mathbf 0,k)$. We sketch below the reduction from an instance $(\bA,\by,k)$ of $\gapvec_{\gamma}$ to an instance $(\bA',\by',k)$ of $\sncp_{\gamma}$. Given $\bA,\by$, let
$$
\bA'=\begin{bmatrix} 
    \bA  \\
    \vdots \\
    \bA\\
    \Id 
    \end{bmatrix}
    \begin{tikzpicture}[baseline=2.2ex]
\draw [decorate,decoration={brace,amplitude=10pt}]
(-1.5,1.7) -- (-1.5,0) node [black,midway,xshift=45pt] { $\gamma k + 1$ copies};
\end{tikzpicture},\\
\by'=\begin{bmatrix} 
    \by  \\
    \vdots \\
    \by\\
    \bzero 
    \end{bmatrix}
    \begin{tikzpicture}[baseline=2.2ex]
\draw [decorate,decoration={brace,amplitude=10pt}]
(-1.5,1.7) -- (-1.5,0) node [black,midway,xshift=45pt] { $\gamma k + 1$ copies};
\end{tikzpicture}
.$$

Notice that for any $\bx$ (in the entire space), we have
$$\|\bA'\bx - \by'\|_0 = (\gamma k + 1)\|\bA\bx-\by\|_0+\|\bx\|_0,$$ 
and thus both the completeness and soundness of the reduction easily follow.

\noindent\textbf{Sparse Covering Codes.}
Before reducing $\sncp$\ to $\mdp_{1.99}$ we need to introduce in more detail (but still informally) the notion of \emph{sparse covering codes} that we previously mentioned. 

A sparse covering code (SCC) is a linear code $\bL$ of block length $h$ with minimum distance $d$ admitting a ball $\cB(\bs, r)$ centered at $\bs$ of radius $r < d$ and containing a large (i.e., about $h^{k}$, where $k=\Omega(d)$) number of codewords. Moreover, for the reduction to $k$-\MDP\ to go through, we require not only the knowledge of the code, but also the center $\bs$ and  a linear transformation $\bT$ used to index the codewords in $\cB(\bs, r)$, i.e., $\bT(\cB(\bs, r)\cap \bL)$ needs to contains the ball of radius $k$ centered at $\mathbf 0$. Similar to how Dumer et al.\ only managed to show the probabilistic existence of the center, we too cannot find an explicit $\bs$ for the SCCs that we construct, but instead provide an efficiently samplable distribution such that, for any $\mathbf{x} \in \cB(\bzero,k)$, the 
probability (over $\bs$ sampled from the distribution) that $\mathbf{x} \in \bT(\mathcal{B}(\bs,r) \cap \bL)$ is non-negligible.
This is what makes our reduction from \sncp\ to \mdp$_{1.99}$ randomized. We will not  elaborate more on this issue here, but focus on the (probabilistic) construction of such codes. For convenience, we will assume throughout this overview that $k$ is much smaller than $d$, i.e., $k = 0.001d$.

Recall that the sphere-packing bound (aka Hamming bound) states that a binary code of block length $h$ and distance $d$ can have at most $2^h/|\cB(\vzero, \lceil \frac{d - 1}{2} \rceil)|$ codewords; this is simply because the balls of radius $\lceil \frac{d - 1}{2} \rceil$ at the codewords do not intersect. Our main theorem regarding the existence of sparse covering code is that any code that is ``near'' the sphere-packing bound is a sparse covering code with $r = \lceil \frac{d - 1}{2} \rceil + k \approx 0.501d$. Here ``near'' means that the number of codewords must be at least $2^h/|\cB(\vzero, \lceil \frac{d - 1}{2} \rceil)|$ divided by $f(d) \cdot \poly(h)$ for some function $f$ that depends only on $d$. (Equivalently, this means that the message length must be at least $h - (d/2 + O(1))\log h$.) The BCH code over binary alphabet is an example of a code satisfying such a condition.

While we will not sketch the proof of the existence theorem here, we note that the general idea is to set $\bT$ and the distribution over $\bs$ in such a way that the probability that $\bx$ lies in $\bT(\cB(\bs, r) \cap \bL)$ is at least the probability that a random point in $\F^h$ is within distance $r - k = \lceil \frac{d - 1}{2} \rceil$ of some codeword. The latter is non-negligible from our assumption that $\bL$ nears the sphere-packing bound.

Finally, we remark that our proof here is completely different from the DMS proof of existence of locally dense codes. Specifically, DMS uses a group-theoretic argument to show that, when a code exceeds the Gilbert–Varshamov bound, there must be a center $\bs$ such that $\cB(\bs, r)$ contains many codewords. Then, they pick a random linear map $\bT$ and show that w.h.p. $\bT(\cB(\bs, r) \cap \bL)$ is the entire space. Note that this second step does not use any structure of $\cB(\bs, r) \cap \bL$; their argument is simply that, for any sufficiently large subset $Y$, a random linear map $\bT$ maps $Y$ to an entire space w.h.p. However, such an argument fails for us, due to the fact that, in SCC, we want to cover a ball $\cB(\vzero, k)$ rather than the whole space, and it is not hard to see that there are very large subsets $Y$ such that no linear map $\bT$ satisfies $\bT(Y) \supseteq \cB(\vzero, k)$. A simple example of this is when $Y$ is a subspace of $\F^h$; in this case, even when $Y$ is as large as $\exp(\poly(h))$, no desired linear map $\bT$ exists.

\noindent\textbf{Reducing $\sncp_\gamma$ to  \mdp$_{1.99}$.} Next, we prove the hardness of $\mdp_{\gamma'}$ for all constant $\gamma'\in [1,2)$, assuming PIH, using a gadget constructed from sparse covering codes.

Given an instance $(\bA,\by,k)$ of $\sncp_{\gamma}$ for some $\gamma>2$ and a sparse covering code $(\bL,\bT,\bs)$ we build an instance $(\bB,ak+br)$ of $\mdp_{\gamma'}$ where $\gamma'<\frac{2\gamma}{2+\gamma}$, by following the DMS reduction (which was previously described, and in particular see \eqref{eq:LDCGadget}).
If there exists $\bx\in \cB(\mathbf 0,k)$ such that $\|\bA\bx-\by\|_0\leqs k$ then consider $\bz'$ such that $\bT\bL \bz' = \bx$. Note that the existence of such a $\bz'$ is guaranteed by the definition of SCC. Consider $\bz=\bz'\circ 1$, and note that $\|\bB\bz\|_0=a\|\bA\bx-\by\|_0+b\|\bL\bz-\bs\|_0\leqs ak+br$. In other words, as in the DMS reduction, if $(\bA,\by,k)$ is a YES instance of \NCP, then $(\bB,ak+br)$ is a YES instance of \MDP. On the other hand, similar to the DMS reduction, if we had that $\|\bA\bx-\by\|_0 > \gamma k$ for all $\bx$, then $\|\bB\bz\|_0>\gamma' (ak+br)$ for all $\bz$. The parameterized intractability of \mdp$_{1.99}$ is obtained by setting $\gamma=400$ in the above reduction.

\noindent\textbf{Gap Amplification for $\mdp_{1.99}$.} It is  well known that the distance of the tensor product of two linear codes is the product of the distances of the individual codes (see Proposition~\ref{prop:gap-amplification} for a formal statement). We can use this proposition to reduce $\mdp_{\gamma}$ to $\mdp_{\gamma^2}$ for any $\gamma\geqs 1$. In particular, we can obtain, for any constant $\gamma$, the intractability of $\mdp_\gamma$ starting from $\mdp_{1.99}$ by just recursively tensoring the input code $\lceil \log_{1.99} \gamma\rceil $ times.

\subsection{Parameterized Intractability of $k$-\SVP\ from PIH}

We begin this subsection by briefly describing Khot's reduction. The starting point of Khot's reduction is the \NP-hardness of approximating \CVP\ in every $\ell_p$ norm to any constant factor \cite{ABSS97}. Let us recall that in \CVP\ in the $\ell_p$ norm, we are given a matrix $\bA \in \Z^{n\times m}$, an integer $k$, and a target vector $\by\in\Z^n$, and the goal is to determine whether there is any $\bx\in\Z^m$ such that\footnote{Previously, we use $\|\bA\bx - \by\|_p$ instead of $\|\bA\bx - \by\|_p^p$. However, from the fixed parameter perspective, these two versions are equivalent since the parameter $k$ is only raised to the $p$-th power, and $p$ is a constant in our setting.} $\|\bA\bx - \by\|_p^p$ is at most $k$.  The result of Arora et al. \cite{ABSS97} states that for any constant $\gamma\geqs 1$, it is \NP-hard to distinguish the case when there exists $\bx$ such that $\|\bA\bx - \by\|_p^p\leqs k$ from the case when for all (integral) $\bx$ we have that $\|\bA\bx - \by\|_p^p> \gamma k$. 
Khot's reduction proceeds in four steps. First, he constructs a gadget lattice called the  ``BCH Lattice'' using BCH Codes. Next, he reduces \CVP\ in the $\ell_p$ norm (where $p\in(1,\infty)$) to an instance of \SVP\ on an intermediate lattice by using the BCH Lattice.
This intermediate lattice has the following property. For any YES instance of \CVP\ the intermediate lattice contains multiple copies of the witness of the YES instance; For any NO instance of \CVP\ there are also many ``annoying vectors'' (but far less than the total number of YES instance witnesses)  which look like witnesses of a YES instance. However, since the annoying vectors are outnumbered, Khot reduces this intermediate lattice to a proper \SVP\ instance, by randomly picking a sub-lattice via a random homogeneous linear constraint on the coordinates of the lattice vectors (this annihilates all the annoying vectors while retaining at least one witness for the YES instance). Thus he obtains some constant factor hardness for \SVP. Finally, the gap is amplified via ``Augmented Tensor Product''. It is important to note that Khot's reduction is randomized, and thus his result of inapproximability of \SVP\ is based on $\NP\neq \RP$.

Trying to follow Khot's reduction, in order to show the parameterized intractability of $k$-\SVP, we face only one obstacle: there is no known parameterized inapproximability of $k$-\CVP\  for any constant factor greater than 1. Let us denote by $\snvp_{p,\eta}$ for any constant $\eta\geqs 1$ the gap version of $k$-\CVP\ in the $\ell_p$ norm. Recall that in $\snvp_{p,\eta}$ we are given a matrix $\bA \in \Z^{n\times m}$, a target vector $\by\in\Z^n$, and a parameter $k$, and we would like to distinguish the case when there exists $\bx \in \Z^m$ such that $\|\bA\bx - \by\|_p^p\leqs k$ from the case when for all $\bx \in \Z^m$ we have that $\|\bA\bx - \by\|_p^p> \eta k$.
As it turns out, our reduction from $\csp_\varepsilon$ to $\sncp$ (with arbitrary constant gap), having $\gapvec_{\gamma}$ and  $\gapvec$ as intermediate steps, can be translated to show the constant inapproximability of $\snvp_p$ (under PIH) in a straightforward manner. We will not elaborate on this part of the proof any further here and defer the detailed proof to Appendix~\ref{sec:csp-to-snvp}.

Once we have established the constant parameterized inapproximability of $\snvp_p$, we follow Khot's reduction, and everything goes through as it is to establish the inapproximability for some factor of the gap version of $k$-\SVP\ in the $\ell_p$ norm (where $p\in (1,\infty)$). We denote by $\svp_{p,\gamma}$ for some constant $\gamma(p)\geqs 1$ the the gap version of $k$-\SVP\ (in the $\ell_p$ norm) where we are given a matrix $\bB \in \Z^{n\times m}$ and a parameter $k \in \N$, and we would like to distinguish the case when there exists a non-zero $\bx \in \Z^m$ such that $\|\bB\bx\|_p^p\leqs k$ from the case when for all $\bx \in \Z^m \setminus \{\bzero\}$ we have that $\|\bB\bx\|_p^p> \gamma k$. 
Let $\gamma^*:=\frac{2^p}{2^{p-1}+1}$.
Following Khot’s reduction, we obtain the inapproximability of $\svp_{p,\gamma^*}$ (under PIH). To obtain inapproximability of $\svp_2$ for all constant ratios, we use the tensor product of lattices; the argument needed here is slightly more subtle than the similar step in \MDP\ because, unlike distances of codes, the $\ell_2$ norm of the shortest vector of the tensor product of two lattices is not necessarily equal to the product of the $\ell_2$ norm of the shortest vector of each lattice. Fortunately, Khot's construction is tailored so that the resulting lattice is ``well-behaved'' under tensoring~\cite{Khot05,HR07}, and gap amplification is indeed possible for such instances.

We remark here that, for the (non-parameterized) inapproximability of SVP, the techniques of~\cite{Khot05,HR07} allow one to successfully amplify gaps for $\ell_p$ norm where $p \ne 2$ as well. Unfortunately, this does not work in our settings, as it requires the distance $k$ to be dependent on $nm$ which is not possible for us since $k$ is the parameter of the problem.

Summarizing, in Figure~\ref{fig:overviewSV}, we provide the proof outline of our reduction from Gap-ETH to $\svp_p$ with some constant gap, for every $p\in (1,\infty)$ (with the additional gap amplification to constant inapproximability for $p=2$).

\begin{figure}[h!]
    \centering
    \resizebox{0.85\textwidth}{!}{\begin{tikzpicture}

\node (gapeth) [draw=red!80!black,thick] at (0,0) {Gap-ETH$_\delta$};

\node (PIH) [draw=red!80!black,thick] at (5, 0) {PIH$_\varepsilon$};

\node (gapvec) [draw=red!80!black,thick] at (10, 0) {\snvp$_p$};
\node (ampgapvec) [draw=red!80!black,thick] at (2, -3) {\svp$_{p,\frac{2^p}{2^{p-1}+1} }$};

\node (gapsnc) [draw=red!80!black,thick] at (8, -3) {\svp$_p$};

\draw [-{Latex[length=1.5mm, width=1.5mm]}] (gapeth) -- (PIH);
\draw [-{Latex[length=1.5mm, width=1.5mm]}] (PIH) -- (gapvec);
\draw [-{Latex[length=1.5mm, width=1.5mm]}] (gapvec) -- (ampgapvec);
\draw [-{Latex[length=1.5mm, width=1.5mm]}] (ampgapvec) -- (gapsnc);

\node [above, align=center] at (2.75, 0.5) {\footnotesize Folklore Reduction};
\node [above, align=center] at (2.75, 0) {\footnotesize (Section~\ref{sec:gapETHtoPIH})};

\node [above, align=center] at (7.2, 0) {\footnotesize Theorem~\ref{thm:CVPmain}};

\node [above, align=center,rotate=19.5] at (5.15, -1.4) {\footnotesize Khot's Reduction};
\node [above, align=center,rotate=19.5] at (5.4, -1.75) {\footnotesize Lemma~\ref{lem:snvp-to-svp}};
\node [above, align=center] at (5.23, -3.6) {\footnotesize Gap Amplification};
\node [above, align=center] at (5.23, -4.05) {\footnotesize for $p = 2$ (Section~\ref{sec:svp-gap-amp})};

\end{tikzpicture}}
    \caption{The figure provides an overview of the reduction from Gap-ETH to the parameterized Shortest Vector problem in the $\ell_p$ norm, where $p\in(1,\infty)$. First, recall that Gap-ETH implies PIH (see Section~\ref{sec:gapETHtoPIH}). Next, we reduce $\csp_\varepsilon$ to $\snvp_p$ in Appendix~\ref{sec:csp-to-snvp}. Lemma~\ref{lem:snvp-to-svp} (i.e., Khot's reduction)  then implies  the parameterized inapproximability of $\svp_{p,\frac{2^p}{2^{p-1}+1}}$. The final step (for $p= 2$) is the Haviv-Regev gap amplification via tensor product, which is described in Section~\ref{sec:svp-gap-amp}.} \label{fig:overviewSV}
\end{figure}
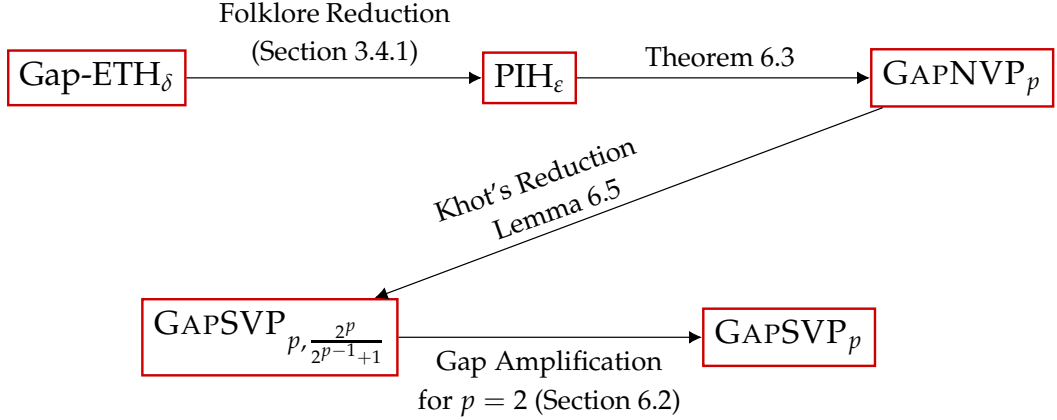

\section{Preliminaries} \label{sec:prelim}

We use the following notations throughout the paper.

\noindent\textbf{Notations.} For $p \in \mathbb{N}$, we use $\vone_p$ (respectively, $\vzero_p$) to denote the all ones (respectively, all zeros) vector of length $p$.
We sometimes drop the subscript if the dimension is clear from the context. 

For $p,q \in \mathbb{N}$, we use $\vzero_{p\times q}$ to denote the all zeroes matrix of $p$ rows and $q$ columns. We use Id$_q$ to denote the identity matrix of $q$ rows and $q$ columns.

For any vector $x\in\mathbb{R}^d$, the $\ell_p$ norm of $x$
is defined as
$\ell_p(x) = \|x\|_p = \left(\sum_{i=1}^d|x_i|^p\right)^{1/p}$.
Thus, $\ell_{\infty}(x) = \|x\|_{\infty} = \max_{i\in[d]}\{|x_i|\}$.
The $\ell_0$ norm of $x$ is defined as
$\ell_0(x) = \|x\|_0=|\{x_i\neq 0: i\in [d]\}|$, i.e.,
the number of non-zero entries of $x$. We note that the $\ell_0$ norm is also referred to as the Hamming norm. For $a\in\mathbb N$, $t\in\mathbb{N}\cup \{0\}$, and $\bs\in\{0,1\}^a$, we use $\mathcal{B}_a(\bs,t)$ to denote the Hamming ball of radius $t$ centered at $s$, i.e., $\mathcal{B}_a(\bs,t)=\{\bx\in\{0,1\}^a\mid \|\bs-\bx\|_0\leqs t\}$. Finally, given two vectors $\bx \in \F^m$ and $\by \in \F^n$, we use $\bx \circ \by \in \F^{m + n}$ to denote the concatenation of vectors $\bx$ and $\by$.

\subsection{Parameterized Promise Problems and (Randomized) FPT Reductions}
In this subsection, we briefly describe the various kinds of fixed-parameter reductions that are used in this paper. We start by defining the notion of promise problems in the fixed-parameter world, which is naturally analogues to promise problems in the NP world (see e.g.~\cite{Gol06}).

\begin{definition}
A parameterized promise problem $\Pi$ is a pair of parameterized languages $(\Pi_{YES}, \Pi_{NO})$ such that $\Pi_{YES} \cap \Pi_{NO} = \emptyset$.
\end{definition}

Next, we formalize the notion of algorithms for these parameterized promise problems:

\begin{definition}
A deterministic algorithm $\cA$ is said to be an \emph{FPT algorithm for $\Pi$} if the following holds:
\begin{itemize}
\item On any input $(x, k)$, $\cA$ runs in time $f(k)|x|^c$ for some computable function $f$ and constant $c$.
\item (YES) For all $(x, k) \in \Pi_{YES}$, $\cA(x, k) = 1$.
\item (NO) For all $(x, k) \in \Pi_{NO}$, $\cA(x, k) = 0$.
\end{itemize}
\end{definition}

\begin{definition}
A Monte Carlo algorithm $\cA$ is said to be a \emph{randomized FPT algorithm for $\Pi$} if the following holds:
\begin{itemize}
\item $\cA$ runs in time $f(k)|x|^c$ for some computable function $f$ and constant $c$ (on every randomness).
\item (YES) For all $(x, k) \in \Pi_{YES}$, $\Pr[\cA(x, k) = 1] \geqs 2/3$.
\item (NO) For all $(x, k) \in \Pi_{NO}$, $\Pr[\cA(x, k) = 0] \geqs 2/3$.
\end{itemize}
\end{definition}

Finally, we define deterministic and randomized reductions between these problems.

\begin{definition}
A (deterministic) FPT reduction from a parameterized promise problem $\Pi$ to a parameterized promise problem $\Pi'$ is a (deterministic) procedure that transforms $(x, k)$ to $(x', k')$ that satisfies the following:
\begin{itemize}
\item The procedure runs in $f(k) |x|^c$ for some computable function $f$ and constant $c$.
\item There exists a computable function $g$ such that $k' \leqs g(k)$ for every input $(x, k)$.
\item For all $(x, k) \in \Pi_{YES}$, $(x', k') \in \Pi_{YES}'$.
\item For all $(x, k) \in \Pi_{NO}$, $(x', k') \in \Pi_{NO}'$.
\end{itemize}
\end{definition}

\begin{definition}					\label{def:rand-fpt-red}
A randomized (one sided error) FPT reduction from a parameterized promise problem $\Pi$ to a parameterized promise problem $\Pi'$ is a randomized procedure that transforms $(x, k)$ to $(x', k')$ that satisfies the following:
\begin{itemize}
\item The procedure runs in $f(k) |x|^c$ for some computable function $f$ and constant $c$ (on every randomness).
\item There exists a computable function $g$ such that $k' \leqs g(k)$ for every input $(x, k)$.Mi
\item For all $(x, k) \in \Pi_{YES}$, $\Pr[(x', k') \in \Pi_{YES}'] \geqs 1/(f'(k)|x|^{c'})$ for some computable function $f'$ and constant $c'$.
\item For all $(x, k) \in \Pi_{NO}$, $\Pr[(x', k') \in \Pi_{NO}'] = 1$.
\end{itemize}
\end{definition}

Note that the above definition corresponds to the notion of \emph{Reverse Unfaithful Random (RUR) reductions} in the classical world~\cite{J90}. The only difference (besides the allowed FPT running time) is that the above definition allows the probability that the YES case gets map to the YES case to be as small as $1/(f'(k)\poly(|x|))$, whereas in the RUR reductions this can only be $1/\poly(|x|)$. The reason is that, as we will see in Lemma~\ref{lem:red-intract} below, FPT algorithms can afford to repeat the reduction $f'(k)\poly(|x|)$ times, whereas polynomial time algorithms can only repeat $\poly(|x|)$ times.

We also consider randomized two-sided error FPT reductions, which are defined as follows. 

\begin{definition}
	A randomized \emph{two sided error} FPT reduction from a parameterized promise problem $\Pi$ to a parameterized promise problem $\Pi'$ is a randomized procedure that transforms $(x, k)$ to $(x', k')$ that satisfies the following:
	\begin{itemize}
		\item The procedure runs in $f(k) |x|^c$ for some computable function $f$ and constant $c$ (on every randomness).
		\item There exists a computable function $g$ such that $k' \leqs g(k)$ for every input $(x, k)$.
		\item For all $(x, k) \in \Pi_{YES}$, $\Pr[(x', k') \in \Pi_{YES}'] \geqs 2/3$.
		\item For all $(x, k) \in \Pi_{NO}$, $\Pr[(x', k') \in \Pi_{NO}']  \geqs 2/3$.
	\end{itemize}
\end{definition}

Note that this is not a generalization of the standard randomized FPT reduction (as defined in Definition \ref{def:rand-fpt-red}), since the definition requires the success probabilities for the YES and NO cases to be constants independent of the parameter. In both cases, using standard techniques randomized FPT reductions, can be used to transform randomized FPT algorithms for $\Pi'$ to randomized FPT algorithm for $\Pi$, as stated by the following lemma:

\begin{lemma}
	Suppose there exists a randomized (one sided/ two sided) error FPT reduction from a parameterized promise problem $\Pi$ to a parameterized promise problem $\Pi'$. If there exists a randomized FPT algorithm $\mathcal{A}$ for $\Pi'$, there there also exists a randomized FPT algorithm for $\Pi$.
	\label{lem:red-intract}
\end{lemma}
\begin{proof}
	We prove this for one sided error reductions, the other case follows using similar arguments. Suppose there exists a randomized one sided error reduction from $\Pi$ to $\Pi'$. 
	Let $f'(\cdot),c'$ be as in Definition \ref{def:rand-fpt-red}. We consider the following subroutine. Given instance $(x,k)$ of promise problem $\Pi$, we apply the randomized reduction on $(x,k)$ to get instance $(x',k')$ of promise problem $\Pi'$. 
	We run $\mathcal{A}$ on $(x',k')$ repeatedly $100\log (f'(k)|x|^c)$ times, and output the majority of the outcomes. 
	
	If $(x,k)$ is a YES instance, then with probability at least $1/(f'(k)|x|^{c'})$, $(x',k')$ is also a YES instance for $\Pi'$. Using Chernoff bound, conditioned on $(x',k')$ being a YES instance, the majority of the outcomes is YES with probability at least $1  - e^{-10\log(f'(k)|x|^{c'})}$. Therefore using union bound, the output of the above algorithm is YES with probability at least $1/(f'(k)|x|^{c'}) - e^{-10\log(f'(k)|x|^{c'})} \geqs 1/2(f'(k)|x|^{c'})$. Similarly, if $(x,k)$ is a NO instance, then the subroutine outputs YES with probability at most $ e^{-10\log(f'(k)|x|^{c'})}$.
	
	Equipped with the above subroutine, our algorithm is simply the following: given $(x,k)$, it runs the subroutine $10f'(k)|x|^{c'}$ times. If at least one of the outcomes is YES, then the algorithm outputs YES, otherwise it outputs NO. Again we can analyze this using elementary probability. If $(x,k)$ is a YES instance, then the algorithm outputs NO only if outcomes of all the trials is NO. Therefore, the algorithm outputs YES with probability at least $1 - ( 1 - 1/2(f'(k)|x|^{c'}))^{10f'(k)|x|^{c'}} \geqs 0.9$. Conversely, if $(x,k)$ is a NO instance, then by union bound, the algorithm outputs NO with probability at least $1 - 10f'(k)|x|^{c'}  e^{-10\log(f'(k)|x|^{c'})} \geqs 0.9$. Finally, if $\mathcal{A}$ is FPT, then the running time of the proposed algorithm is also FPT. Hence the claim follows{\footnote{For the case of $2$-sided error, we change the final step of the algorithm as follows; we invoke the subroutine $O(\log 1/\delta)$-times (where $\delta$ is a constant) and again output the majority of the outcomes. The guarantees again follow by a Chernoff bound argument. }}.
\end{proof}	

Since the conclusion of the above proposition holds for both types of randomized reductions, we will not be distinguishing between the two types in the rest of the paper. 

\subsection{Minimum Distance Problem}
In this subsection, we define the fixed-parameter variant of the minimum distance problem and other relevant parameterized problems. We actually define them as gap problems -- as later in the paper, we show the constant inapproximability of these problems. 

For every $\gamma\geqs 1$, we define the $\gamma$-gap minimum distance problem\footnote{In the parameterized complexity literature, this problem is referred to as the $k$-Even set problem \cite{DFVW99} and the input to the problem is (equivalently) given through the parity-check matrix, instead of the generator matrix as described in this paper.} as follows:

\begin{framed}
$\gamma$-Gap Minimum Distance Problem ($\mdp_{\gamma}$)

{\bf Input: } A matrix $\bA \in \mathbb{F}_2^{n \times m}$ and a positive integer $k \in \mathbb{N}$

{\bf Parameter: } $k$

{\bf Question: } Distinguish between the following two cases:
\begin{itemize}
\item (YES) there exists $\bx \in \mathbb{F}_2^m \setminus \{\bzero\}$ such that $\|\bA\bx\|_0 \leqs k$
\item (NO) for all $\bx \in \mathbb{F}_2^m \setminus \{\bzero\}$, $\|\bA\bx\|_0 > \gamma \cdot k$
\end{itemize}
\end{framed}

Next, for every $\gamma\geqs 1$, we define the $\gamma$-gap maximum likelihood decoding problem\footnote{The maximum likelihood decoding problem is also equivalently known in the literature as the nearest codeword problem.} as follows:

\begin{framed}
$\gamma$-Gap Maximum Likelihood Decoding Problem ($\gapvec_{\gamma}$)

{\bf Input: } A matrix $\bA \in \mathbb{F}_2^{n \times m}$, a vector $\by \in \mathbb{F}_2^n$ and a positive integer $k \in \mathbb{N}$

{\bf Parameter: } $k$

{\bf Question: } Distinguish between the following two cases:
\begin{itemize}
\item (YES) there exists $\bx \in \cB_{m}(\mathbf{0}, k)$ such that $\bA\bx = \by$
\item (NO) for all $\bx \in \cB_{m}(\mathbf{0}, \gamma k)$, $\bA\bx \ne \by$
\end{itemize}
\end{framed}

For brevity, we shall denote the exact version (i.e., $\gapvec_1$) of the problem as $\kvec$.
Finally, we introduce a ``sparse'' version of the \gapvec\ problem called the  sparsest nearest codeword problem, and later in the paper we show a reduction from \gapvec\ to this problem, followed by a reduction from this problem to \mdp. Formally, for every $\gamma\geqs 1$, we define the $\gamma$-gap sparsest nearest codeword problem as follows:

\begin{framed}
$\gamma$-Gap Sparse Nearest Codeword Problem ($\sncp_{\gamma}$)

{\bf Input: } A matrix $\bA \in \mathbb{F}_2^{n \times m}$, a vector $\by \in \mathbb{F}_2^n$ and a positive integer $k \in \mathbb{N}$

{\bf Parameter: } $k$

{\bf Question: } Distinguish between the following two cases:
\begin{itemize}
\item (YES) there exists $\bx \in \cB_{m}(\mathbf{0}, k)$ such that $\|\bA\bx - \by\|_0 \leqs k$
\item (NO) for all $\bx \in \mathbb{F}_2^m$, $\|\bA\bx - \by\|_0 > \gamma \cdot k$
\end{itemize}
\end{framed}

\subsection{Shortest Vector Problem and Nearest Vector Problem}

In this subsection, we define the fixed-parameter variants of the shortest vector and nearest vector problems. As in the previous subsection, we define them as gap problems, for the same reason that later in the paper, we show the constant inapproximability of these two problems. 

Fix $p\in\mathbb{R}_{\geqs 1}$. For every $\gamma\geqs 1$, we define the $\gamma$-gap shortest vector problem in the $\ell_p$-norm{\footnote{Note that we define $\snvp$ and $\svp$ problems in terms of $\ell^p_p$, whereas traditionally, it is defined in terms of $\ell_p$. However, it is sufficient for us to work with the $\ell^p_p$ variant, since an $\alpha$-factor inapproximability in $\ell^p_p$ translates to an $\alpha^{1/p}$-factor inapproximabillity in the $\ell_p$ norm, for any $\alpha \geqs 1$}} as follows:
\begin{framed}
	$\gamma$-Gap Shortest Vector Problem ($\svp_{p,\gamma}$)
	
	{\bf Input: } A matrix $\bA \in \mathbb{Z}^{n \times m}$ and a positive integer $k \in \mathbb{N}$
	
	{\bf Parameter: } $k$
	
	{\bf Question: } Distinguish between the following two cases:
	\begin{itemize}
		\item (YES) there exists $\bx \in \mathbb{Z}^{m} \setminus \{\bzero\}$ such that $\|\bA\bx\|^p_p \leqs k$
		\item (NO) for all $\bx \in \mathbb{Z}^m \setminus \{\bzero\}$, $\|\bA\bx\|^p_p > \gamma \cdot k$
	\end{itemize}
\end{framed}

For every $\gamma\geqs 1$, we define the $\gamma$-gap nearest vector problem in the $\ell_p$-norm as follows:

\begin{framed}
	$\gamma$-Gap Nearest Vector Problem ($\snvp_{p,\gamma}$)
	
	{\bf Input: } A matrix $\bA \in \mathbb{Z}^{n \times m}$, vector $\by \in \mathbb{Z}^n$ and a positive integer $k \in \mathbb{N}$
	
	{\bf Parameter: } $k$
	
	{\bf Question: } Distinguish between the following two cases:
	\begin{itemize}
		\item (YES) there exists $\bx \in \Z^{m}$ such that $\|\bA\bx - \by\|^p_p \leqs k$
		\item (NO) for all $\bx \in \Z^m$, $\|\bA\bx - \by\|^p_p > \gamma \cdot k$
	\end{itemize}
\end{framed}

\subsection{CSPs and Parameterized Inapproximability Hypothesis}
\label{subsec:pih}


In this section, we will formally state the Parameterized Inapproximability Hypothesis (PIH). To do so, we first have to define 2CSP and its corresponding gap problem, starting with the former:

\begin{definition}[2CSP]
An instance $\Gamma$ of 2CSP consists of
\begin{itemize}
\item an undirected graph $G = (V, E)$, which is referred to as the \emph{constraint graph},
\item an \emph{alphabet set} $\Sigma$,
\item for each edge $e = (u, v) \in E$, a \emph{constraint} $C_{uv} \subseteq \Sigma \times \Sigma$.
\end{itemize} 
An \emph{assignment} of $\Gamma$ is simply a function from $V$ to $\Sigma$. An edge $e = (u, v) \in E$ is said to be \emph{satisfied} by an assignment $\psi: V \to \Sigma$ if $(\psi(u), \psi(v)) \in C_{uv}$. A \emph{value} of an assignment $\psi$, denoted by $\val(\psi)$, is the fraction of edges satisfied by $\psi$, i.e., $\val(\psi) = \frac{1}{|E|} \cdot \{(u, v) \in E \mid (\psi(u), \psi(v)) \in C_{uv}\}$. The value of the instance $\Gamma$, denoted by $\val(\Gamma)$, is the maximum value among all possible assignments, i.e., $\val(\Gamma) = \max_{\psi: V \to \Sigma} \val(\psi)$.
\end{definition}

The gap problem for 2CSP can then be defined as follows:

\begin{framed}
$\varepsilon$-Gap 2CSP ($\csp_{\varepsilon}$)

{\bf Input: } A 2CSP instance $\Gamma = (G = (V, E), \Sigma, \{C_{uv}\}_{(u, v) \in E})$.

{\bf Parameter: } $|V|$

{\bf Question: } Distinguish between the following two cases:
\begin{itemize}
\item (YES) $\val(\psi) = 1$.
\item (NO) $\val(\psi) < 1 - \varepsilon$.
\end{itemize}
\end{framed}

Note that, when $\varepsilon = 0$, $\csp_0$ is simply asking whether the input instance $\Gamma$ is fully satisfiable. It is easy to see that this problem generalizes the $k$-Clique problem, which is well known to be $\W[1]$-hard; in other words, $\csp_0$ is $\W[1]$-hard, and hence does not admit FPT algorithm unless $\W[1] = \FPT$. It is believed that $\csp_{\varepsilon}$ remains $\W[1]$-hard even for some constant $\varepsilon > 0$; this belief has recently been formalized by Lokshtanov \etal~\cite{LRSZ17} as the \emph{Parameterized Inapproximability Hypothesis (PIH)}. For the purpose of this work, we will use an even weaker version of the hypothesis than Lokshtanov \etal's. Namely, we will only assume that $\csp_{\varepsilon}$ is not in $\FPT$ for some $\varepsilon > 0$ (rather than assuming that it is $\W[1]$-hard), as stated below.

\begin{hypothesis}[Parameterized Inapproximability Hypothesis (PIH)~\cite{LRSZ17}] \label{hyp:pih}
There exists $\varepsilon > 0$ such that there is no randomized FPT algorithm for $\csp_{\varepsilon}$.
\end{hypothesis}

\subsubsection{Relation to Gap Exponential Time Hypothesis}\label{sec:gapETHtoPIH}


There are several supporting evidences for PIH. One such evidence is that it follows from the Gap Exponential Time Hypothesis hypothesis, which can be stated as follows.

\begin{hypothesis}[Randomized Gap Exponential Time Hypothesis (Gap-ETH)~\cite{D16,MR16}] \label{hyp:gap-eth}
There exist constants $\varepsilon,\delta> 0$ such that any randomized algorithm that, on input a 3CNF formula $\varphi$ on $n$ variables and $O(n)$ clauses, can distinguish between $\sat(\varphi)=1$ and $\sat(\varphi)<1-\varepsilon$, with probability at least $\nicefrac{2}{3}$, must run in time at least $2^{\delta n}$.
\end{hypothesis}

Gap-ETH itself is a strengthening of the Exponential Time Hypothesis (ETH)~\cite{IP01,IPZ01}, which states that deciding whether a 3CNF formula is satisfiable cannot be done in subexponential time. We remark here that Gap-ETH would follow from ETH if a linear-size PCP exists; unfortunately, no such PCP is known yet, with the shortest known PCP having quasi-linear size~\cite{BS08,D07}. It has also recently been shown that Gap-ETH would follow from ETH if we assume certain ``smoothness'' conditions on instances from ETH~\cite{App17}. We will not discuss the evidences supporting Gap-ETH in details here; we refer interested readers to~\cite{D16}.

As stated earlier, it is known that PIH follows from Gap-ETH (see e.g.~\cite{DM18,CFM17}\footnote{Note that, in~\cite{CFM17}, 2CSP is referred to as \emph{Maximum Colored Subgraph Isomorphism}. Note also that the hardness proof of~\cite{CFM17} relies heavily on the parameterized inapproximability of Densest $k$-Subgraph from~\cite{CCKLMNT17}, which in turns relies on the reduction and the main lemma from~\cite{M17}.}). Unfortunately, the proofs in literature so far have been somewhat complicated, since the previous works on the topic (e.g.~\cite{DM18,CCKLMNT17,M17}) put emphasis on achieving as large a factor hardness of approximation for 2CSP as possible. On the other hand, some constant inapproximability factor strictly greater than zero suffices to show PIH. For this regime, there is a (folklore) proof that is much simpler than those in the literature. Since we are not aware of this proof being fully written down anywhere, we provide it in Appendix~\ref{app:gap-eth}. Note that this proof also yields a running time lower bound of $T(k) \cdot |\Sigma|^{\Omega(k)}$ of solving $\csp_{\varepsilon}$ for some $\varepsilon > 0$; this running time lower bound is better than those provided by the aforementioned previous proofs and is essentially optimal, since one can solve 2CSP (even exactly) in time $|\Sigma|^{O(k)}$.

\subsection{Error-Correcting Codes}
An error correcting code $C$ over alphabet $\Sigma$ is a function $C: \Sigma^m \to \Sigma^h$ where $m$ and $h$ are positive integers which are referred to as the {\em message length} and {\em block length} of $C$ respectively. Intuitively, the function $C$ encodes an original message of length $m$ to an encoded message of length $h$.
The {\em distance} of a code, denoted by $d(C)$, is defined as $\underset{x \ne y \in \Sigma^m}{\min} \|C(x)- C(y)\|_0$, i.e., the number of coordinates on which $C(x)$ and $C(y)$ disagree.
We also define the systematicity of a code as follows: Given $s \in \mathbb N$, a code $C:\Sigma^m\to \Sigma^{h}$ is {\em $s$-systematic} if there exists a size-$s$ subset of $[h]$, which for convenience we identify with $[s]$, such that for every $x \in \Sigma^{s}$ there exists $w \in \Sigma^m$ in which $x = C(w)\mid_{[s]}$.
We use the shorthand $[h,m,d]_{|\Sigma|}$ to denote a code of message length $m$, block length $h$, and distance $d$. 

Additionally, we will need the following existence and efficient construction of BCH codes for every message length and distance parameter.

\begin{theorem}[BCH Code~\cite{H59,BR60}] \label{thm:bch}
For any choice of $h, d \in \mathbb{N}$ such that $h + 1$ is a power of two and that $d \leqs h$, there exists a linear code over $\F$ with block length $h$, message length $h - \left\lceil\frac{d-1}{2}\right\rceil\cdot\log (h + 1)$ and distance $d$. Moreover, the generator matrix of this code can be computed in $\poly(h)$ time. 
\end{theorem}

Finally, we define the tensor product of codes which will be used later in the paper. Consider two linear codes $C_1 \subseteq \F^m$ (generated by $\mathbf{G}_1 \in \F^{m \times m'}$) and ${C}_2 \subseteq \F^n$  (generated by $\mathbf{G}_2 \in \F^{n \times n'}$). Then the tensor product of the two codes ${C}_1 \otimes {C}_2 \subseteq \F^{m \times n}$ is defined as 
\begin{equation*}
{C}_1 \otimes {C}_2 = \{\mathbf{G}_1\mathbf{X}\mathbf{G}^\top_2  | \mathbf{X} \in \F^{m' \times n '}\}.
\end{equation*}
%
%

\section{Parameterized Intractability of \gapvec} \label{sec:csp-to-gapvec}

In this section, we will show the parameterized intractibility of $\gapvec$ as stated below.

\begin{theorem}\label{thm:MLDmain}
	Assuming PIH, there is no randomized FPT algorithm for $\gapvec_\gamma$ for any $\gamma \geqs 1$.
\end{theorem}

The proof proceeds in two steps. First, we reduce \csp\ to $\gapvec_{\gamma}$ for some $\gamma > 1$ in Section~\ref{sec:base-reduction}. Then, we boost the gap in the \gapvec\ problem in Section~\ref{sec:gap-amplification}.

\subsection{Reducing \csp\ to \gapvec}\label{sec:base-reduction}

In this subsection, we show an FPT reduction from \csp\ to \gapvec, as stated below:

\begin{lemma}	\label{lem:csp-to-gapvec}
For any $\varepsilon \geqs 0$, there is an FPT reduction from $\csp_{\varepsilon}$ to $\gapvec_{\gamma}$ where $\gamma = 1 + \varepsilon/3$.
\end{lemma}

Combining the reduction with PIH, we obtain the following hardness of \gapvec:

\begin{theorem}\label{thm:MLDbase}
	Assuming PIH, for some $\gamma >1 $, there is no randomized FPT algorithm for $\gapvec_\gamma$.
\end{theorem}

\begin{proof}[Proof of Lemma~\ref{lem:csp-to-gapvec}]
Let $\Gamma = (G = (V, E), \Sigma, \{C_{uv}\}_{(u, v) \in E})$ be the input for $\csp_{\varepsilon}$. We produce an instance $(\bA, \by, k)$ of $\gapvec_{\gamma}$ as follows. Let $n = |V| + |E| + 2 |E| |\Sigma|$ and $m = |V||\Sigma| + \sum_{(u, v) \in E} |C_{uv}|$; our matrix $\bA$ will be of dimension $(n \times m)$. For convenience, let us label the first $|V|$ rows of $\bA$ by the vertices $u \in V$, the next $|E|$ rows by the edges $e \in E$, and the last $2 |E| |\Sigma|$ rows by a tuple $(e, \sigma, b) \in E \times \Sigma \times \{0, 1\}$. Furthermore, we label the first $|V||\Sigma|$ columns of $\bA$ by $(u, \sigma_u)$ where $u \in V$ and $\sigma_u \in \Sigma$, and the rest of the columns by $(e, \sigma_0, \sigma_1)$ where $e \in E$ and $(\sigma_0, \sigma_1) \in C_e$. The entries of our matrix $\bA \in \F^{n \times m}$ can now be defined as follows.
\begin{itemize}
\item For each column of the form $(u, \sigma_u)$, let the following entries be one: 
\begin{itemize}
\item $\bA_{u, (u, \sigma_u)}$.
\item $\bA_{((u, v), \sigma_u, 0), (u, \sigma_u)}$ for every $v \in V$ such that $(u, v) \in E$.
\item $\bA_{((v, u), \sigma_u, 1), (u, \sigma_u)}$ for every $v \in V$ such that\footnote{For the clarity of presentation, we assume here that $G$ is directed.} $(v, u) \in E$.
\end{itemize}
The rest of the entries are set to zero.
\item For each column of the form $(e, \sigma_0, \sigma_1)$, let the following three entries be one: $\bA_{e, (e, \sigma_0, \sigma_1)}$, $\bA_{(e, \sigma_0, 0), (e, \sigma_0, \sigma_1)}$ and $\bA_{(e, \sigma_1, 1), (e, \sigma_0, \sigma_1)}$.
The rest of the entries are set to zero.
\end{itemize}
Finally, we set $\by = \vone_{|V| + |E|} \circ \vzero_{2|E||\Sigma|}$ and $k = |V| + |E|$.

{\bf Parameter Dependency.} The new parameter $k$ is simply $|V| + |E| = O(|V|^2)$.

Before we move on to prove the completeness and soundness of the reduction, let us state an observation that will be useful in the analysis:
\begin{observation}
For any $\bx \in \F^{m}$, the following properties hold.
\begin{itemize}
\item For every row of the form $u \in V$, we have 
\begin{align}
(\bA\bx)_{u} = \sum_{\sigma \in \Sigma} \bx_{(u, \sigma)}.
\label{eq:row-1}
\end{align}
\item For every row of the form $e \in E$, we have
\begin{align}
(\bA\bx)_e = \sum_{(\sigma_0, \sigma_1) \in C_e} \bx_{(e, \sigma_0, \sigma_1)}.
\label{eq:row-2}
\end{align}
\item For every row of the form $(e = (u_0, u_1), \sigma, b)$, we have
\begin{align}
(\bA\bx)_{(e, \sigma, b)} = \bx_{(u_b, \sigma)} + \sum_{(\sigma_0, \sigma_1) \in C_e \atop \sigma_b = \sigma} \bx_{(e, \sigma_0, \sigma_1)}.
\label{eq:row-3}
\end{align}
\end{itemize}
\end{observation}
Note that the above observation follows trivially from our definition of $\bA$.

{\bf Completeness.} Suppose that $\val(\Gamma) = 1$. That is, there exists an assignment $\psi: V \to \Sigma$ that satisfies all the edges. We define the vector $\bx \in \F^m$ as follows.
\begin{itemize}
\item For each $(u, \sigma)$, let $\bx_{(u, \sigma)} = \ind[\psi(u) = \sigma]$
\item For each $(e = (u, v), \sigma_0, \sigma_1)$, let $\bx_{(e, \sigma_0, \sigma_1)} = \ind[(\psi(u) = \sigma_0) \wedge (\psi(v) = \sigma_1)]$.
\end{itemize}
We claim that $\bA\bx = \by$. To see that this is the case, consider the following three cases of rows.
\begin{itemize}
\item For each row of the form $u \in V$, we have
\begin{align*}
(\bA\bx)_{u} \overset{(\ref{eq:row-1})}{=} \sum_{\sigma \in \Sigma} \bx_{(u, \sigma)} = \sum_{\sigma \in \Sigma} \ind[\psi(u) = \sigma] = 1.
\end{align*}
\item For each row of the form $e = (u, v) \in E$, we have 
\begin{align*}
(\bA\bx)_e \overset{(\ref{eq:row-2})}{=} \sum_{(\sigma_0, \sigma_1) \in C_e} \bx_{(e, \sigma_0, \sigma_1)} = \sum_{(\sigma_0, \sigma_1) \in C_e} \ind[(\psi(u) = \sigma_0) \wedge (\psi(v) = \sigma_1)] = 1.
\end{align*}
Note that, in the last equality, we use the fact that the edge $e$ is satisfied (i.e. $(\psi(u), \psi(v)) \in C_e$).
\item For each row $(e = (u_0, u_1), \sigma, b) \in E \times \Sigma \times \{0, 1\}$, we have 
\begin{align*}
(\bA\bx)_{(e, \sigma, b)} \overset{(\ref{eq:row-3})}{=} \bx_{(u_b, \sigma)} + \sum_{(\sigma_0, \sigma_1) \in C_e \atop \sigma_b = \sigma} \bx_{(e, \sigma_0, \sigma_1)} = \ind[\psi(u_b) = \sigma] + \ind[\psi(u_b) = \sigma] = 0,
\end{align*}
where the second equality uses the fact that exactly one of $\bx_{(e, \sigma_0, \sigma_1)}$ is not zero among $(\sigma_0, \sigma_1) \in C_e$, which is $\bx_{(e, \psi(u_0), \psi(u_1))}$.
\end{itemize}

Hence, $\bA\bx$ is indeed equal to $\by$. Finally, observe that $\|\bx\|_0 = |V| + |E|$ as desired.

{\bf Soundness.} We will prove this by contrapositive. Suppose that the constructed instance $(\bA, \by, k)$ is not a NO instance of $\gapvec_\gamma$, i.e., for some $x \in \cB_{m}(\bzero, \gamma k)$, $\bA\bx = \by$. We will show that $\Gamma$ is not a NO instance of $\csp_{\varepsilon}$, i.e., that $\val(\Gamma) \geqs 1 - \varepsilon$.

First, for each vertex $u \in V$, let $S_u = \{\sigma \in \Sigma \mid \bx_{(u, \sigma)} = 1\}$ and, for each $e \in E$, let $T_e = \{(\sigma_0, \sigma_1) \in C_e \mid \bx_{(e, \sigma_0, \sigma_1)} = 1\}$. From (\ref{eq:row-1}) and from $(\bA\bx)_u = \by_u = 1$, we can conclude that $|S_u| \geqs 1$ for all $u \in V$. Similarly, from (\ref{eq:row-2}) and from $(\bA\bx)_e = \by_e = 1$, we have $|T_e| \geqs 1$ for all $e \in E$.

We define an assignment $\psi: V \to \Sigma$ of $\Gamma$ by setting $\psi(u)$ to be an arbitrary element of $S_u$ for all $u \in V$. We will show that $\val(\psi) \geqs 1 - \varepsilon$, which indeed implies $\val(\Gamma) \geqs 1 - \varepsilon$.

To do so, let $E_{\uni}$ denote the set of $e \in E$ such that $|T_e| = 1$. Notice that
\begin{align}
(\gamma - 1) k &\geqs \|\bx\|_0 - k 						\label{eq:start}\\ 
&= \sum_{u \in V} |S_u| + \sum_{e \in E} |T_e| - k 			\nonumber	\\
&= \sum_{u \in V} (|S_u| - 1) + \sum_{e \in E} (|T_e| - 1) 	\nonumber   \\
&\geqs \sum_{e \in E \setminus E_{\uni}} (|T_e| - 1) 		\nonumber	\\
&\geqs |E \setminus E_{\uni}|,								\label{eq:end}
\end{align}
which implies that $|E_{\uni}| \geqs |E| - (\gamma - 1)k$, which, from our choice of $\gamma$, is at least $|E| - \frac{\varepsilon k}{3} \geqs (1 - \varepsilon)|E|$. Note that the last inequality follows from $|V| \leqs 2|E|$, which can be assumed w.l.o.g.

Since $|E_{\uni}| \geqs (1 - \varepsilon)|E|$, to show that $\val(\psi) \geqs 1 - \varepsilon$, it suffices to show that $\psi$ satisfies every edge in $E_{\uni}$. To see that this is the case, let $e = (u, v)$ be any edge in $E_{\uni}$. Let $(\sigma^*_0, \sigma^*_1)$ be the only element of $T_e$. Observe that, from (\ref{eq:row-3}) with $\sigma = \psi(u)$ and $b = 0$ and from $(\bA\bx)_{e, \sigma, b} = \by_{e, \sigma, b} = 0$, we can conclude that $\sigma^*_0 = \psi(u)$. Similarly, we can conclude that $\sigma^*_1 = \psi(v)$. As a result, $(\psi(u), \psi(v))$ must be in $C_e$, meaning that $\psi$ satisfies $e$, which completes our proof.
\end{proof}

\subsection{Gap Amplification}\label{sec:gap-amplification}

In this subsection, we describe the gap amplification step for $\gapvec$. Towards that end, we define a \emph{composition} operation $\oplus$ on $\gapvec$ instances, which can be used to efficiently amplify the gap to any constant factor with only a polynomial blowup in the instance size. 

\noindent{\bf The Composition Operator}: Consider two $\gapvec$ instances given by coefficient matrices ${\bf A} \in \mathbb{F}^{u \times v}_2, \mathbf{B} \in \mathbb{F}^{u' \times v'}_2$ and non-zero target vectors $\mathbf{z} \in \mathbb{F}^u_2, \mathbf{z}' \in \mathbb{F}^{u'}_2$. Their composition is a $\gapvec$ instance  $(\mathbf{C},\mathbf{w},k_2(k_1 + 1)) = (\mathbf{A},\mathbf{z},k_1) \oplus (\mathbf{B},\mathbf{z}',k_2)$, given by coefficient matrix $\mathbf{C} \in \mathbb{F}^{(u' + uv') \times (v' + vv')}_2$, and target vector $\mathbf{w} \in \mathbb{F}^{u' + uv'}_2$. They are constructed as follows. Consider the following partition of $[u' + uv']$ into blocks $S_0,S_1,\cdots,S_{v'}$, where $S_0$ consists of the first $u'$ coordinates, and for every $i \in [v']$ the block $S_i$ consists of coordinates $u' + (i-1)v+ 1,\ldots,u' + iv$. Similarly, consider the partition $T_0,T_1,\ldots,T_{v'}$ of $[v' + vv']$, where $T_0$ consists of the first $v'$ coordinates, and blocks $T_1,\ldots,T_{v'}$ is the contiguous equipartition of the remaining coordinates. Furthermore, for any choice of $S_i,T_j$, let $\mathbf{C}_{S_i,T_j}$ be the sub-matrix of $\mathbf{C}$ consisting of rows and columns indexed by $S_i$ and $T_j$ respectively. Extending the notation to vectors $\mathbf{x} \in \mathbb{F}^{v' + vv'}_2$, we use $\mathbf{x}^i$ to denote the restriction of the vectors along the coordinates in $T_i$. Now we describe the construction of $\mathbf{C}$ and $\mathbf{w}$ in a row block wise order. 

\begin{itemize}
	\item[1.] {\bf Row block} $S_0$: We set $\mathbf{C}_{S_0,T_0} = \mathbf{B}$, and $\mathbf{C}_{S_0,T_i} = \mathbf{0}_{u' \times v}$ for every $i \in [v']$. Additionally, we set the corresponding sub-vector $\mathbf{w}_{S_0} = \mathbf{z'}$.
	\item[2.] {\bf Row block} $S_i$ (for $i \geqs 1$): Here $\mathbf{C}_{S_i,T_0}$ contains $\mathbf{z}$ in the $i^{th}$ column and is zero everywhere else. The sub-matrix $\mathbf{C}_{S_i,T_i}$ is set to the coefficient matrix $\mathbf{A}$, and for all blocks $j \notin \{0,i\}$, we set the sub-matrices $\mathbf{C}_{S_i,T_j} = \mathbf{0}_{u \times v}$. Finally, for the target vector we set the corresponding block $\mathbf{w}_{S_i} = \mathbf{0}_u$.
\end{itemize}

The intuition underlying the above construction is as follows. For simplicity, let $k = k_1 = k_2$. Consider the row block $S_0$ in $(\mathbf{C},\mathbf{w},k^2 + k)$. By construction, the set of indices of non-zero column of $\mathbf{C}$ along $S_0$ is exactly $T_0$, and by construction $\mathbf{C}_{S_0,T_0} = \mathbf{B}$. In particular, for any vector $\mathbf{x} \in \mathbb{F}^{v' \times vv'}$ satisfying $\mathbf{C}\mathbf{x} = \mathbf{w}$, the constraints along rows in $S_0$ enforce $\mathbf{B}\mathbf{x}^0 = \mathbf{z}'$, and hence $\|\mathbf{x}^0\|_0 \geqs k$. Similarly, for any $i \in [v']$, the non-zero columns blocks corresponding to the row block $S_0$ are $T_0$ and $T_i$. Since $\mathbf{C}_{S_i,T_i} = \mathbf{A}$, the row block $S_i$ forces the constrains $\mathbf{A}\mathbf{x}^i = x^i_0\mathbf{z}$. Therefore, for any $i \in [v]$ such that $x^i_0 \ne 0$, the sub-vector $\mathbf{x}^i$ must satisfy $\mathbf{A}\mathbf{x}^i = \mathbf{z}$, and hence $\|\mathbf{x}^i\|_0 \geqs k$. Since this must happen for every $i \in [v']$ such that $x^0_i \neq 0$, the overall Hamming weight of the vector $\mathbf{x}$ must be at least $k+k^2$. These observations are made formal in the following lemma.

\begin{lemma}									\label{lem:gap_step}
	Consider coefficient matrices $\mathbf{A} \in \mathbb{F}^{u \times v}_2,\mathbf{B} \in \mathbb{F}^{u' \times v'}_2$, and target vectors $\mathbf{z} \in \mathbb{F}^v_2,\mathbf{z}' \in \mathbb{F}^{v'}_2$ corresponding to $\gapvec_\gamma$ instances $(\mathbf{A},\mathbf{z},k_1)$ and $(\mathbf{B},\mathbf{z}',k_2)$. Then $(\mathbf{C},\mathbf{w},k_2(k_1 + 1)) = (\mathbf{A},\mathbf{z},k_1)\oplus(\mathbf{B},\mathbf{z}',k_2)$ as constructed above satisfies the following properties:
	\begin{itemize}
		\item If $(\mathbf{A},\mathbf{z},k_1)$ and $(\mathbf{B},\mathbf{z}',k_2)$ are YES instances, then there exists $\bx \in \mathcal{B}_{v' + vv'}\big(\mathbf{0},k_2 + k_1k_2\big)$ such that $\mathbf{C}\bx = \mathbf{w}$.
		\item If $(\mathbf{A},\mathbf{z},k_1)$ and $(\mathbf{B},\mathbf{z}',k_2)$ are NO instances, then for every $\bx \in \mathcal{B}_{v' + vv'}\big(\mathbf{0},\gamma k_2 + \gamma^2k_1k_2\big)$ we have $\mathbf{C}\bx \neq \mathbf{w}$.
	\end{itemize} 
\end{lemma}

\begin{proof}
   
	Suppose $(\mathbf{A},\mathbf{z},k_1)$ and $(\mathbf{B},\mathbf{z}',k_2)$ are NO instances. Consider any vector $\mathbf{x} \in \mathbb{F}^{v' + vv'}_2$ satisfying $\mathbf{C}\mathbf{x} = \mathbf{w}$. By construction, it must satisfy the constraints corresponding to row-block $S_0$ i.e., $\mathbf{B}\mathbf{x}^0 = \mathbf{z}'$, and therefore $\|\mathbf{x}^0 \|_0 > \gamma k_2$. Let $\Omega = \{ i \in [v'] : x^0_i \neq 0\}$ be the set of non-zero indices in $\mathbf{x}^0$, and fix any $i \in \Omega$. It is easy to see that along row-block $S_i$ the only variables acting on non-zeros blocks are the variable $x^0_i$ (acting on the sub-matrix $\mathbf{C}_{S_i,T_0}$) and the vector $\mathbf{x}^i$ (acting on the sub-matrix $\mathbf{C}_{S_i,T_i} = \mathbf{A}$). All together, this induces constraints $\mathbf{A}\mathbf{x}^i = x^0_i\cdot \mathbf{z}$. Since $x^0_i \neq 0$, the vector $\mathbf{x}^i$ must satisfy $\mathbf{A}\mathbf{x}^i = \mathbf{z}$ and therefore $\|\mathbf{x}^i\|_0 > \gamma k_1$. Moreover, since this must happen for every $i \in \Omega$ and $|\Omega| > \gamma k_2$, we have $\|\mathbf{x}\|_0 > \gamma k_2 + \gamma^2 k_1k_2$ as desired. 
	
	It remains to be shown if $(\mathbf{A},\mathbf{z},k_1)$ and $(\mathbf{B},\mathbf{z}',k_2)$ are YES instances, then there exists a satisfying solution with Hamming weight at most $k_2(k_1 + 1)$. Let $\mathbf{a} \in \mathcal{B}_{v}(\mathbf{0},k_1)$ and $\mathbf{b} \in \mathcal{B}_{v'}(\mathbf{0},k_2)$ be such that $\mathbf{A}\mathbf{a} = \mathbf{z}$ and $\mathbf{B}\mathbf{b} = \mathbf{z}'$. We construct $\mathbf{x} \in \mathbb{F}^{v' + vv'}_2$ as follows. We set $\mathbf{x}^0 = \mathbf{b}$ and for every $i \in [v]$, we set $\mathbf{x}^i = b_i\cdot\mathbf{a}$. It can be easily verified that the vector $\mathbf{x}$ is a solution to $(\mathbf{C},\mathbf{w})$. Indeed, along row-block $S_0$, we have $\mathbf{B}\mathbf{x}^0 = \mathbf{B}\mathbf{b} = \mathbf{z}$. Furthermore, if $x^0_i = 1$ the corresponding row-block $S_i$ satisfies $\mathbf{A}\mathbf{x}^i = \mathbf{A}\mathbf{a} = \mathbf{z}$. Finally, if $x^0_i = 0$, the corresponding row-block evaluates to $\mathbf{A}\mathbf{x}^i = \mathbf{0}_u$ which is trivially satisfied by the given construction. Since the vector $\mathbf{x}$ has Hamming weight at most $k_2(k_1 + 1)$, the claim follows. 	
\end{proof}

The above lemma can be used to amplify the gap for $\gapvec_\gamma$ instances, as shown by the following corollary.

\begin{corollary}									\label{corr:gap-amplify}
	For all choices of $\gamma > 1$,$\eta > 0$ and $k \geqs \Big(\gamma^{\eta} - 1\Big)^{-1}$, there exists an FPT reduction from $k$-$\gapvec_\gamma$ to $(k^2 + k)$-$\gapvec_{\gamma^{2-\eta}}$.
\end{corollary}
\begin{proof}
	Let $(\mathbf{A},\mathbf{z},k)$ be a $\gapvec_\gamma$ instance, where $\mathbf{A} \in \F^{u \times v}$. Let $(\mathbf{C},\mathbf{w},k^2 + k) = (\mathbf{A},\mathbf{z},k) \oplus (\mathbf{A},\mathbf{z},k)$ be constructed as above. If $(\mathbf{A},\mathbf{z},k)$ is a YES instance, then there exists $\bx \in \mathcal{B}_{v + uv}(\mathbf{0},k + k^2)$ of Hamming such that $\mathbf{C}\bx = \mathbf{w}$. Conversely, if $(\mathbf{A},\mathbf{z},k)$ is a NO instance, then every solution to $\mathbf{C}\bx = \mathbf{w}$ has Hamming weight at least 
	 $ \gamma^2k^2 + \gamma k \geqs \gamma^{2-\eta}(k^2 + k)$ (by our choice of $k$). Hence the claim follows.
\end{proof}


Theorem~\ref{thm:MLDmain} is a direct consequence of the above lemma:

\begin{proof}[Proof of Theorem~\ref{thm:MLDmain}]
	Consider the hardness of \gapvec\ as given in Theorem~\ref{thm:MLDbase}.
	For any $\gamma \geqs 1 + \epsilon/3$, applying Corollary \ref{corr:gap-amplify} $i \geqs \lceil \log(\gamma)/\log(1 + \epsilon/3)\rceil$ times on $\gapvec_{(1+\epsilon/3)}$ instance can boost the gap to $\gapvec_{\gamma}$. Since all the steps in the reduction are FPT, the claim follows.
\end{proof}

\section{Parameterized Intractability of Minimum Distance Problem}			\label{sec:main-reduction}

Next, we will prove our main theorem regarding parameterized intractability of \mdp:

\begin{theorem}\label{thm:MDPmain}
	Assuming PIH, there is no randomized FPT algorithm for $\mdp_\gamma$ for any $\gamma \geqs 1$.
\end{theorem}

This again proceeds in two steps. First, we provide an FPT reduction from \gapvec\ to \sncp\ in Section~\ref{sec:gapvec-to-snc}. Next, we formalize the definition of Sparse Covering Codes (\SCC), prove their existence, and show how to use them to reduce \sncp\ to \mdp\ in Section~\ref{sec:dense-code}.

\subsection{Reducing $\gapvec$ to $\sncp$} \label{sec:gapvec-to-snc}
In this subsection, we show an FPT reduction \gapvec\ to \sncp: 

\begin{lemma}							\label{lem:gapvec-to-sncp}
For any constant $\gamma > 1$, there is an FPT reduction from $\gapvec_{\gamma}$ to $\sncp_{\gamma}$.
\end{lemma}

\begin{proof}
Given an instance $(\bA, \by, k)$ of $\gapvec_{\gamma}$ where $\bA \in \F^{n \times m}$. We create an instance $(\bA', \by', k')$ of $\sncp_{\gamma}$ by letting
\begin{align*}
\bA' =
\begin{bmatrix}
\vone_{\lceil \gamma k + 1\rceil} \otimes \bA \\
\Id_m
\end{bmatrix}
\in \F^{(\lceil \gamma k + 1\rceil n + m) \times m}, 
\by' =
\begin{bmatrix}
\vone_{\lceil \gamma k + 1\rceil} \otimes \by \\
\bzero_m
\end{bmatrix}
\in \F^{\lceil \gamma k + 1\rceil n + m}
\end{align*}
and $k' = k$.

{\bf Completeness} Suppose that $(\bA, \by, k)$ is a YES instance of $\gapvec_{\gamma}$, i.e., there exists $\bx \in \cB_m(\bzero, k)$ such that $\bA\bx = \by$. Consider $\bA'\bx - \by'$. We have
\begin{align*}
\bA'\bx - \by' =
\begin{bmatrix}
\vone_{\lceil \gamma k + 1\rceil} \otimes (\bA\bx - \by) \\
\bx
\end{bmatrix}
=
\begin{bmatrix}
\bzero_{\lceil \gamma k + 1\rceil n} \\
\bx
\end{bmatrix}.
\end{align*}
Hence, $\bx$ is a vector in the ball $\cB_m(\bzero, k)$ such that $\|\bA'\bx - \by'\|_0 = \|\bx\|_0 \leqs k$, meaning that $(\bA', \by', k')$ is a YES instance of $\sncp_{\gamma}$.

{\bf Soundness} Suppose that $(\bA, \by, k)$ is a NO instance of $\gapvec_{\gamma}$. Consider any $\bx \in \cB_m(\bzero, \gamma k)$. To show that $\|\bA'\bx - \by\|_0 > \gamma k$, let us consider two cases.
\begin{itemize}
\item Case 1: $\|\bx\|_0 > \gamma k$. In this case, we have $\|\bA'\bx - \by'\|_0 \geqs \|\bx\|_0 > \gamma k$.
\item Case 2: $\|\bx\|_0 \leqs \gamma k$. In this case, we have $\|\bA'\bx - \by'\|_0 \geqs \lceil \gamma k + 1 \rceil \|\bA\bx - \by\|_0 \geqs \lceil \gamma k + 1 \rceil > \gamma k$, where the third inequality comes from the fact that $(\bA, \by, k)$ is a NO instance of $\gapvec_{\gamma}$ and hence $\bA\bx \ne \by$.
\end{itemize}
As a result, we can conclude that $(\bA', \by', k')$ is a NO instance of $\sncp_{\gamma}$.
\end{proof}

\subsection{Sparse Covering Codes} \label{sec:dense-code}

Below, we provide the definition of Sparse Covering Codes\footnote{The definition can be naturally extended to fields of larger size.} over $\F$.

\begin{definition}\label{def:SCC}
A Sparse Covering Code (\SCC) over $\F$ with parameters\footnote{We remark that the parameter $h$ is implicit in specifying \SCC.} $(m,q,t,d,r,\delta)$ is a tuple  $(\mathbf L,\mathbf T,\mathcal D)$ where $\mathbf L\in \F^{h\times m}$ is the basis of an $m$ dimensional linear code with minimum distance (at least) $d$, $\mathbf T\in \F^{q\times h}$ is a linear transformation, and $\mathcal D$ is a poly$(h)$ time samplable distribution over $\F^{h}$ such that for any $\mathbf{x} \in \cB_q(\bzero,t)$, the following holds:
	\begin{equation} \label{eq:dense_prob}
	\Pr_{\mathbf{s} \sim \mathcal{D}} \Bigg[\mathbf{x} \in \bT\Big(\mathcal{B}_h(\bs,r) \cap \bL\mathbb{F}^m_2\Big)\Bigg] \geqs \delta.
	\end{equation}
\end{definition}

We would like to note that \SCC are closely related to Locally Dense Codes (see \cite{DMS03} or \cite{Mic14} for the definition). A key difference in our definition is that the code is tailored towards covering sparse vectors which allows us greater flexibility in the parameters. This is in direct contrast to previous constructions which were designed to cover entire subspaces, with parameters which end up having strong dependencies on the ambient dimension. Consequently, they are not directly applicable to the parameterized setting. 

Next, we show the existence of \SCC  for a certain range of parameters. Informally, the existence of \SCC follows from the existence of codes near the sphere-packing bound, such as BCH codes.

\begin{lemma} \label{lem:ldc}
	For any $q,t \in \mathbbm{N}$ and any $\varepsilon > 0$, there exist $d,r,h,m \in \mathbb{N}$, two linear maps $\bL \in \F^{h \times m}$ and $\bT \in \F^{q \times h}$, and a poly$(h)$ time samplable distribution $\cD$ over $\F^{h}$ such that $(\mathbf{L},\mathbf{T},\mathcal{D})$ is a \SCC with parameters $\left(m,q,t,d,r,\frac{1}{d^{d/2}}\right)$. Additionally, the following holds:
	\begin{itemize}
		\item $r \leqs (\nicefrac{1}{2} + \varepsilon) d$,
		\item $r, d \leqs O(t/\varepsilon)$ and $h, m \leqs \poly(q, t,\nicefrac{1}{\varepsilon})$,
		\item $\bL$ and $\bT$ can be computed in poly($q, t,\nicefrac{1}{\varepsilon}$) time for all $\varepsilon > 0$.
	\end{itemize}
\end{lemma}

\begin{proof}
	Let $d = 2\lceil t/\varepsilon \rceil + 1$ and let $r = \left(\frac{d-1}{2}\right) + t$. Let $h$ be the smallest positive integer such that $h + 1$ is a power of two and that $h \geqs \max\{2q, d\}$. Finally, we set $m = h - \left(\frac{d - 1}{2}\right)\log(h + 1)$. It is clear that the chosen parameters satisfy the first two conditions, i.e., $r \leqs (\nicefrac{1}{2} + \varepsilon) d$	 and $r, d \leqs O(t/\varepsilon)$, and $h, m \leqs \poly(q, t,\nicefrac{1}{\varepsilon})$.

	Let $\mathbf{L}$ be the generator matrix of the $[h,m,d]_2$ linear code as given by Theorem~\ref{thm:bch}. Without loss of generality, we assume that the code is systematic on the first $m$ coordinates. The linear map $\mathbf{T}$ is defined as
	$\mathbf{T} \overset{\rm def}{=} \Big[ \Id_q \enskip\vzero_{q\times h-q}\Big]$ i.e., it is the matrix which projects onto the first $q$ coordinates. The distribution $\mathcal{D}$ is given as follows: we set the first $q$ coordinates of the vector to $0$ and each of the remaining $(h - q)$ coordinates is sampled uniformly and independently at random from $\F$. It is clear that $\mathbf T$ is computed in $\poly(q,t,\nicefrac{1}{\varepsilon})$ time and $\mathcal D$ is a $\poly(h)$ time samplable distribution over $\F^h$. From Theorem~\ref{thm:bch}, we also have that $\mathbf L$ can be computed in $\poly(h)=\poly(q,t,\nicefrac{1}{\varepsilon})$ time.
	
	It remains to show that for our choices of matrices $\mathbf{L},\mathbf{T}$ and distribution $\mathcal{D}$, equation \ref{eq:dense_prob} holds for any fixed choice of $\mathbf{x} \in \cB_q(\bzero,t)$. Fix a vector $\mathbf{x} \in \mathcal{B}_{q}(\mathbf{0},t)$ and define the set $\mathcal{C}= \Big\{ {\bf z} \in \mathbb{F}^{h - q} \Big| {\bf x} \circ {\bf z} \in \mathbf{L}\mathbb{F}^m \Big\}$. Since the code generated by $\mathbf{L}$ is systematic on the first $m$ coordinates, we have that $|\mathcal{C}| = 2^{m - q}$. 
Moreover, since the code generated by $\bL$ has distance $d$, we have that every distinct pair of vectors $\mathbf{z}_1,\mathbf{z}_2 \in \mathcal{C}$ are at least $d$-far from each other (i.e. $\|\bz_1 - \bz_2\|_0 \geqs d$). 
	
	Recall that $r - t = \frac{d - 1}{2}$. Therefore, for any distinct pair of vectors $\mathbf{z}_1,\mathbf{z}_2 \in \mathcal{C}$, the sets $\cB_{h - q}(\mathbf{z}_1,r-t)$ and  $\cB_{h - q}(\mathbf{z}_2,r-t)$ are disjoint. Hence the number of vectors in the union of $(r-t)$-radius Hamming balls around every $\mathbf{z} \in \mathcal{C}$ is
	\begin{align*}
	2^{m - q}\left\lvert\mathcal{B}_{h - q}\left(\mathbf{0}, \frac{d-1}{2} \right)\right\rvert
	\geqs 2^{m - q}{h - q \choose \frac{d-1}{2} }
	\geqs 2^{m - q}{h/2 \choose \frac{d-1}{2} }
	\geqs 2^{m - q}\Big(\frac{h}{d - 1}\Big)^{\frac{d-1}{2}} 
	\end{align*}	
	On the other hand, $|\F^{h - q}| = 2^{h-q} = 2^{m - q}(h + 1)^{\frac{d-1}{2}}$. Hence, with probability at least $\left(\frac{h}{(d - 1)(h + 1)}\right)^{\frac{d-1}{2}} \geqs \frac{1}{d^{d/2}}$, a vector $\mathbf{s}'$ sampled uniformly from $\mathbb{F}^{h - q}$ lies in $\cB_{h - q}(\bp', r - t)$ for some vector $\mathbf{p}' \in\mathcal{C}$. Fix such a vector $\mathbf{s}'$ and consider the vectors $\mathbf{s} = \mathbf{0}_q \circ \mathbf{s}'$ and $\mathbf{p} = \mathbf{x} \circ  \mathbf{p}'$. Then by construction,
	\begin{equation*}
	\|\mathbf{s} - \mathbf{p}\|_0 = \|\mathbf{s}' - \mathbf{p}'\|_0 + \|\mathbf{x}\|_0 \leqs (r - t) + t = r,
	\end{equation*} 
	 and $\mathbf{T}\mathbf{p} = \mathbf{x}$. To complete the proof, we observe that the distribution of the vector $\mathbf{s}$ (as constructed from $\mathbf{s}'$) is identical to $\mathcal{D}$, and hence the claim follows. 
\end{proof}

\subsection{Reducing \sncp\ to \mdp} \label{sec:main-red}

In this subsection, we state and prove the FPT reduction from the $\sncp$ problem to the $\mdp$ problem. It uses the basic template of the reduction from \cite{DMS03}, which is modified to work in combination with \SCC.

\begin{lemma}							\label{lem:sncp-to-mdp}
For any constants $\gamma' > 2$ and $\gamma \geqs 1$ such that $\gamma < \frac{2\gamma'}{2 + \gamma'}$, there is a randomized FPT reduction from $\sncp_{\gamma'}$ to $\mdp_{\gamma}$.
\end{lemma}

\begin{proof}
Let $(\bB, \by, t)$ be the input for $\sncp_{\gamma}$ where $\bB \in \F^{n \times q}$, $\by\in\F^n$, and $t$ is the parameter. Moreover, let $\varepsilon > 0$ be a sufficiently small constant such that $\gamma < \frac{2\gamma'}{2 + (1 + 2\varepsilon)\gamma'}$. Let $d, r, h, m \in \N, \bL \in \F^{h \times m}, \bT \in \F^{q \times h}$ and $\cD$ be as in Lemma~\ref{lem:ldc}. Pick $a', b' \in \N$ such that
\begin{align}
\frac{\gamma}{\gamma' - \gamma} < \frac{a'}{b'} < \frac{(d/r) - \gamma}{\gamma}.
\label{cond:ab-pick}
\end{align} 
To see that such $a'$ and $b'$ exists, observe the following: $$ \zeta=\frac{(d/r) - \gamma}{\gamma} - \frac{\gamma}{\gamma' - \gamma} \geqs  \frac{2\gamma' - \gamma(2 + (1 + 2\varepsilon)\gamma'+2\varepsilon\gamma)}{(1 + 2\varepsilon)\gamma(\gamma' - \gamma)} > \frac{2\varepsilon\gamma}{(1+2\varepsilon)\cdot (\gamma'-\gamma)}> 0.$$ Moreover, we can always choose $a'$ and  $b'$ so that they are at most $2/\zeta = O(1)$.

Let $a = a'r$ and $b = b't$. Note that, condition~(\ref{cond:ab-pick}) implies that $\gamma(at + br) < \min\{\gamma'at, bd\}$. We produce an instance $(\bA, k)$ for $\mdp_{\gamma}$ by first sampling $\bs \sim \cD$. Then, we set $k = at + br$ and
\begin{align*}
\bA =
\begin{bmatrix}
\vone_a \otimes \bB \bT \bL & - \vone_a \otimes \by \\
\vone_b \otimes \bL & - \vone_b \otimes \bs
\end{bmatrix}
\in \F^{(an + bh) \times (m + 1)}.
\end{align*}

{\bf Parameter Dependencies.} Notice that $k = at + br = a'rt + b'rt = O(rt)$. Moreover, from Lemma~\ref{lem:ldc}, we have that $r \leqs O(t)$. Hence, we can conclude that $k = O(t^2)$ (note that $t$ is the parameter of the input instance $(\bB, \by, t)$ of $\sncp_{\gamma'}$).

{\bf Completeness.} Suppose that $(\bB, \by, t)$ is a YES instance of $\sncp_{\gamma'}$. That is, there exists $\bx \in \cB_q(\bzero, t)$ such that $\|\bB\bx - \by\|_0 \leqs t$. Now, from Lemma~\ref{lem:ldc}, with probability at least $1/d^{d/2}$, we have $\bx \in \bT(\cB_m(\bs, r) \cap \bL \F^m)$. This is equivalent to the following: there exists $\bz' \in \F^m$ such that $\bT\bL\bz' = \bx$ and $\|\bL\bz' - \bs\|_0 \leqs r$. Conditioned on this event, we can pick $\bz = \bz' \circ 1 \in \F^{m + 1}$, which yields
\begin{align*}
\|\bA \bz\|_0 = a\|\bB\bT\bL\bz - \by\|_0 + b\|\bL\bz - \bs\|_0 = a\|\bB\bx - \by\|_0 + b\|\bL\bz - \bs\|_0 \leqs at + br.
\end{align*}

In other words, with probability at least $1/d^{d/2}$, $(\cA, k)$ is a YES instance of $\mdp_{\gamma}$ as desired.

{\bf Soundness.} Suppose that $(\bB, \by, t)$ is a NO instance of $\sncp_{\gamma'}$. We will show that, for all non-zero $\bz \in \F^{m + 1}$, $\|\bA\bz\|_0 > \gamma(at + br)$; this implies that $(\bA, k)$ is a NO instance of $\mdp_{\gamma}$.

To show this let us consider two cases, based on the last coordinate $\bz_{m + 1}$ of $\bz$. Let $\bz' = [\bz_1 \cdots \bz_{m }] \in \F^{m }$ be the vector consisting of the first $m$ coordinates of $\bz$.

If $\bz_{m + 1} = 0$, then $\|\bA\bz\|_0 = a\|\bB\bT\bL\bz'\|_0 + b\|\bL\bz'\|_0 \geqs b\|L\bz'\|_0 \geqs bd$, where the last inequality comes from the fact that $\bL$ is a generator matrix of a code of distance $d$ (and that $\bz\neq\mathbf 0$). Finally, recall that our parameter selection implies that $bd > \gamma(at + br)$, which yields the desired result for this case.

On the other hand, if $\bz_{m + 1} = 1$, then $\|\bA\bz\|_0 = a\|\bB\bT\bL\bz' - \by\|_0 + b\|\bL\bz' - \bs\|_0 \geqs a\|\bB\bT\bL\bz' - \by\|_0 \geqs \gamma'at$, where the last inequality simply follows from our assumption that $(\bB, \by, t)$ is a NO instance of $\sncp_{\gamma'}$. Finally, recall that our parameter selection guarantees that $\gamma'at > \gamma(at + br)$. This concludes the proof.
\end{proof}

\paragraph{Gap Amplification.} Finally, the above gap hardness can be boosted to any constant using the now standard technique of tensoring the code (c.f. \cite{DMS03},\cite{AK14}) using  the following lemma:

\begin{proposition}[E.g. \cite{DMS03}] \label{prop:gap-amplification}
Given two linear codes $C_1 \subseteq \F^m$ and ${C}_2 \subseteq \F^n$, let ${C}_1 \otimes {C}_2 \subseteq \mathbb{F}^{m \times n}$ be the tensor product of ${C}_1$ and ${C}_2$. Then $d({C}_1 \otimes {C}_2) = d({C}_1)d({C}_2)$.
\end{proposition}

We briefly show how the above proposition can be use to amplify the gap. Consider a $\mdp_\gamma$ instance $(\bA,k)$ where $\bA \in \F^{m \times n}$. Let $C \subseteq \F^m$ be the linear code generated by it. Let $C^{\otimes 2} = C \otimes C$ be the tensor product of the code with itself, and let $\bA^{\otimes 2}$ be its generator matrix. By the above proposition, if $(\bA,k)$ is a YES instance, then $d(C^{\otimes 2}) \leqs k^2$. Conversely, if $(\bA,k)$ is a NO instance, then $d(C^{\otimes 2}) \geqs \gamma^2k^2$. Therefore $(\bA^{\otimes 2},k^2)$ is a $\mdp_{\gamma^2}$ instance. Hence, for any $\alpha \in \mathbb{R}_+$, repeating this argument $\lceil\log_\gamma \alpha \rceil$-number of times gives us an FPT reduction from $k$-$\mdp_\gamma$ to $k^{2\lceil\log_\gamma \alpha\rceil}$-$\mdp_\alpha$. We have thereby completed our proof of Theorem~\ref{thm:MDPmain}.

\newcommand{\bU}{\mathbf{U}}
\newcommand{\bW}{\mathbf{W}}
\newcommand{\cL}{\mathcal{L}}

\section{Parameterized Intractability of Shortest Vector Problem}\label{sec:svp}

We begin by stating the intractability of $\svp$, which is the main result of this section:

\begin{theorem}[FPT Inapproximability of $\svp$]\label{thm:SVPmain}
Assuming PIH, any $p > 1$, there exists constant $\gamma_p > 1$ (where $\gamma_p$ depends on $p$), such that there is no randomized FPT algorithm for $\svp_{p,\gamma_p}$.
\end{theorem}

For $p = 2$, the gap can be boosted to any constant factor, as stated by the following:

\begin{theorem}[Constant Inapproximability of $\svp_{2,\gamma}$]\label{thm:SVPmain-boost}
	Assuming PIH, for all constants $\gamma \geqs 1$, there is no randomized FPT algorithm for $\svp_{2,\gamma}$.
\end{theorem}

The proof of Theorem~\ref{thm:SVPmain} is essentially the same as that of Khot~\cite{Khot05}, with only a small change in parameter selection. We devote Subsection~\ref{subsec:khot} to this proof. Theorem~\ref{thm:SVPmain-boost} follows by gap amplification via tensor product of lattices. Unlike tensor products of codes, the length of the shortest vector of the tensor product of two lattices is not necessarily equal to the product of the length of the shortest vector of each lattice. Fortunately, Haviv and Regev~\cite{HR07} showed that a strengthened notion of soundness from~\cite{Khot05}, which will be described soon, allows gap amplification via tensoring in the $\ell_2$ norm. This is explained in more detail in Section~\ref{sec:svp-gap-amp}.

\subsection{Following Khot's Reduction: Proof of Theorem~\ref{thm:SVPmain}} \label{subsec:khot}

The proof of Theorem~\ref{thm:SVPmain} goes through the following FPT inapproximability of $\snvp_{p,\eta}$. 

\begin{theorem}[FPT Inapproximability of $\snvp$]\label{thm:CVPmain}
Assuming PIH, for any $\gamma, p \geqs 1$, there is no FPT algorithm which on input $(\bB, \by, t)$ where $\bB \in \Z^{n \times q}, \by \in \Z^n$ and $t \in \N$, can distinguish between
\begin{itemize}
	\item (YES) there exists $\bx \in \Z^q$ such that $\bB \bx - \by \in \{0,1\}^n$ and $\|\bB \bx - \by \|^p_p \leqs t$.
	\item (NO) for all choices of $\bx \in \Z^q$ and $w \in \Z \setminus \{0\}$, we have $\|\bB \bx - w \cdot \by \|^p_p > \gamma \cdot t$
\end{itemize} 
\end{theorem}

Since the proof of the above theorem is identical to the FPT reduction of $\csp$ to $\sncp$, we defer it to Appendix~\ref{sec:csp-to-snvp}. Note that the reduction has stronger guarantees than a usual reduction from $\csp_\epsilon$ to $\snvp_{p,\eta}$ in both the YES and NO cases. In the YES case, the witness ot the $\snvp$ instance is a $\{0,1\}$-vector. Furthermore, in the NO case of the above lemma, the lower bound guarantee is in terms of the Hamming norm. These properties are used crucially in the analysis for  the reduction from $\snvp_{p,\eta}$ to $\svp_{p,\gamma}$. The proof of Theorem \ref{thm:SVPmain} also requires the following definition of annoying vectors from~\cite{Khot05}.

\begin{definition}[Annoying Vectors]				\label{def:annoying-vectors}
	For a parameter $l \in \mathbb{N}$, we say that a vector $\bz \in \Z^u$ for some $u \in \N$ is an \emph{Annoying Vector} in the $\ell_p$ norm, if it violates all of the following conditions:
	\begin{itemize}
		\item[1.] The Hamming norm of $\bz$ is at least $l$.
		\item[2.] All coordinates of $\bz$ are even, and its Hamming norm is at least $l/2^p$.
		\item[3.] All coordinates of $\bz$ are even, and it has at least one coordinate of magnitude at least $l^{10l}$. 
	\end{itemize}
\end{definition}

As shown in~\cite{Khot05,HR07}, lattices without annoying vectors (for an appropriate value of $l$) are ``well behaved'' with respect to tensor products of lattices. This notion will be useful later, when the lattice tensor product is used to amplify the gap to any constant for $p = 2$ (Section \ref{sec:svp-gap-amp}). 

The second step is a randomized FPT reduction from $\snvp_{p,\eta}$ to $\svp_{p,\gamma}$, as stated by the following lemma. For succinctness, we define a couple of additional notations: let $\cL(\bA)$ denote the lattice generated by the matrix $\bA \in \Z^{n \times m}$, i.e., $\cL(\bA) = \{\bA\bx \mid \bx \in \Z^m\}$, and let $\lambda_p(\cL)$ denote the length (in the $\ell_p$ norm) of the shortest vector of the lattice $\cL$, i.e., $\lambda_p(\cL) = \underset{\bzero \ne \by \in \cL}{\min} \|\by\|_p$.

\begin{lemma}				\label{lem:snvp-to-svp}
	Fix $p > 1$, and let $\eta \geqs 1$ be such that $\frac12 + \frac{1}{2^p} + \frac{(2^p + 1)}{\eta} < 1$. Let $(\bB,\by,t)$ be a $\snvp_{p,\eta}$ instance, as given by Theorem~\ref{thm:CVPmain}. Then, there is a randomized FPT reduction from $\snvp_{p,\eta}$ instance $(\bB,\by,t)$ to $\svp_{p,\gamma}$ instance $(\bB_{\rm svp},\gamma^{-1}_pl)$ with $l = \eta \cdot t$ such that
	\begin{itemize}
		\item $\bB_{\rm svp} \in \Z^{(n + h + g + 1) \times (q + h + g + 2)}$ where $m = l^{O(l)}\cdot{\rm poly}(n,q)$ and $h = {\rm poly}(l,n,q)$.
		\item If $(\bB,\by,t)$ is a YES instance (from Theorem~\ref{thm:CVPmain}), then with probability $0.8$, $\lambda_p(\cL(\bB_{\rm svp}))^p \leqs \gamma^{-1}_p l$.
		\item If $(\bB,\by,t)$ is a NO instance, then with probability $0.9$, for all choices of $\bx \in \Z^{q + h + g + 2} \setminus \{\bzero\}$, the vector $\bB_{\rm svp}\bx$ is not annoying. In particular, we have $\lambda_p(\cL(\bB_{\rm svp}))^p \geqs l$.
	\end{itemize}
	Here $\gamma_p := \frac{1}{\frac12 + (2^p+1)/\eta + 1/2^p}$ is strictly greater than $1$ by our choice of $\eta$. 
\end{lemma}

Combining the above lemma with Theorem~\ref{thm:CVPmain} gives us Theorem \ref{thm:SVPmain}. 

The reduction in the above lemma is nearly identical to the one from \cite{Khot05}. We devote the rest of this subsection to describing this reduction and proving Lemma \ref{lem:snvp-to-svp}. The rest of this section is organized as follows. In Section \ref{sec:bch-lattic}, we define the BCH lattice, which is the key gadget used in the reduction. Using the BCH lattice and the $\snvp_{p, \eta}$ instance, we construct the intermediate lattice in Section \ref{sec:int-lattice}. The intermediate lattice serves to blow up the number of ``good vectors'' for the YES case, while controlling the number of annoying vectors for the NO case. In particular, this step ensures that the number of good vectors (Lemma \ref{lem:good-vectors}) far outnumber the number of annoying vectors (Lemma \ref{lem:annoying-vectors}). Finally, in Section \ref{sec:final-lattice} we compose the intermediate lattice with a random homogeneous constraint (sampled from an appropriate distribution), to give the final $\svp_{p,\gamma}$ instance. The additional random constraint is used to annihilate all annoying vectors in the NO case, while retaining at least one good vector in the YES case.

{\noindent{\bf Setting}: For the rest of the section, we fix $(\bB,\by,t)$ to be a $\snvp_{p,\eta}$ instance (as given by Theorem \ref{thm:CVPmain}) and we set $l := \eta \cdot t$.}

\subsubsection{The BCH Lattice gadget}				\label{sec:bch-lattic}

We begin by defining the BCH lattices which is the key gadget used in the reduction. Given parameters $l,h \in \mathbb{N}$ where $h$ is a power of $2$, the BCH lattice is defined by the generator $\bB_{\rm BCH}$ defined as
\begin{align*}
	\bB_{\rm BCH} =
	\begin{bmatrix}
		{\rm Id}_{h} &  \mathbf{0}_{h \times g} \\
		Q\cdot\mathbf{P}_{\rm BCH} & 2Q\cdot{\rm Id}_{g}
	\end{bmatrix}
	\in \mathbb{Z}^{(h + g) \times (h + g)}.
\end{align*}
\noindent where $\mathbf{P}_{\rm BCH} \in \{0,1\}^{g \times h}$ is the parity check matrix of a $[h,h-g,l]_2$-BCH code (i.e., $g$ is the codimension of the code whose parity check matrix is ${\bf P}_{\rm BCH}$), and $Q = l^{10l}$. By our choice of parameters it follows that $g \leqs \frac{l}{2}\log h$ (see Theorem \ref{thm:bch}). The following lemma summarizes the key properties of BCH lattices, as defined above.

\begin{lemma}				\label{lem:bch-lattice}
	Let $\mathbf{B}_{\rm BCH} \in \mathbb{Z}^{(h + g) \times ( h + g)}$ be the BCH lattice constructed as above. Let $ r = \Big(\frac12 + \frac{1}{2^p} + \frac{1}{\eta}\Big)\eta t$. Then there exists a randomized polynomial time algorithm, which with probability at least $0.99$, returns vector $\mathbf{s} \in \mathbb{Z}^{h+g}$ such that the following holds. There exists at least $\frac{1}{100}2^{-g}{h \choose r}$ distinct coefficient vectors $\mathbf{z} \in \mathbb{Z}^{h + g}$ such that the vector $\big(\mathbf{B}_{\rm BCH} \mathbf{z} - \mathbf{s}\big)$ is a $\{0,1\}$-vector of Hamming weight exactly $r$.
\end{lemma}

We skip the proof of the above lemma, since it is identical to proof of Lemma 4.3 in \cite{Khot05}. Note that this is in fact even weaker than Khot's lemma, since we do not impose a bound on $\|\mathbf{z}\|_p$.

\subsubsection{The Intermediate Lattice}			\label{sec:int-lattice}

We now define the intermediate lattice. Let $(\bB,\by,t)$ be an instance of $\snvp_{p,\eta}$, where $\bB \in \Z^{n \times q}$. The intermediate lattice $\bB_{\rm int}$ is constructed as follows. Let $l = \eta t$. Let $h$ be the smallest power of 2 such that $h \geqs \max\{2^{10p+10}n, (100l)^{100\eta l}\}$, and let $\bB_{\rm BCH}$ be constructed as above. Then 
\begin{align*}
	\bB_{\rm int} =
	\begin{bmatrix}
		2\bB &  \mathbf{0}_{ n \times  (h+g)} &  2\by  \\
		\mathbf{0}_{(h + g) \times q} &  \bB_{\rm BCH} &  \bs
	\end{bmatrix}
	\in \Z^{(n + h + g) \times (q + h + g + 1)}.
\end{align*}
where $\mathbf{s} \in \Z^{h + g}$ is the vector given by Lemma \ref{lem:bch-lattice}. The following lemmas bound the good and annoying vectors in the YES and NO cases respectively.

\begin{lemma}				\label{lem:good-vectors}
   Let $(\mathbf{B},\mathbf{y},t)$ be a YES instance, and let $\mathbf{B}_{\rm int}$ be the corresponding intermediate lattice. Then, with probability at least $0.99$, there exists at least distinct ${h^{r}}\Big({100 h^{l/2}l^{l}}\Big)^{-1}$ vectors $\mathbf{x} \in \mathbb{Z}^{q + h + g + 1}$ such that $\mathbf{B}_{\rm int}\mathbf{x}$ are distinct $\{0,1,2\}$-vectors with $\ell_p$ norm at most $(2^p t + r)^{1/p}$.
\end{lemma}
\begin{proof}
	Let $\tbx \in \mathbb{Z}^{q}$ be vector such that ${\bf B}\tbx - {\bf y}$ is a $\{0,1\}$-vector of Hamming weight at most $t$. From Lemma \ref{lem:bch-lattice}, there exists at least $2^{-g}{h \choose r}/100$ distinct vectors ${\bf z} \in \mathbb{Z}^{h + g}$ such that ${\bf B}_{\rm BCH}{\bf z} - {\bf s}$ is a $\{0,1\}$-vector of Hamming weight exactly $r$. For each such distinct coefficient vector, consider the vector $\bx = \tbx \circ \bz \circ -1$. By construction, it follows that ${\bf B}_{\rm int}\bx = (2{\bf B}\tbx - 2{\bf y}) \circ({\bf B}_{\rm BCH}{\bf z} - {\bf s})$ is a $\{0,1,2\}$-vector and $\|{\bf B}_{\rm int}\bx\|_p^p = 2^p\|{\bf B}\tbx - {\bf y}\|_p^p + \|{\bf B}_{\rm BCH}{\bf z} - {\bf s}\|_p^p\leqs 2^pt + r$. Since the number of such vectors is at least the number of distinct coefficient vectors $\mathbf{z}$, it can be lower bounded by
	\begin{equation*}
	\frac{1}{100} \cdot 2^{-g}{h \choose r} \geqs \frac{1}{100} \cdot 2^{-\frac{l}{2}\log h} {h \choose r} \geqs \frac{1}{100} \cdot \frac{h^r}{r^r h^{l/2}} \geqs \frac{1}{100} \cdot \frac{h^r}{l^l h^{l/2}}.
	\end{equation*}
	\noindent Finally, observe that each $\bz$ produces different $\bB_{\rm BCH}\bz$ and hence all $\bB_{\rm int}\bv$'s are distinct.
\end{proof}	
	
\begin{lemma}				\label{lem:annoying-vectors}
	Let $(\mathbf{B},\mathbf{y},t)$ be a NO instance, and let $\mathbf{B}_{\rm int}$ be the corresponding intermediate lattice. Then the number of annoying vectors for $\mathbf{B}_{\rm int}$ is at most $10^{-5}{h^{r}}\Big({100 h^{l/2}{l}^{l}}\Big)^{-1}$.
\end{lemma}

\begin{proof}
	Consider $\bx \in \mathbb{Z}^{q + h + g + 1}$ such that $\bB_{\rm int}\mathbf{x}$ is an annoying vector. We partition the vector and write it as $\bx = \bx_1 \circ \bx_2 \circ x$ where $\bx_1 \in \Z^q$, $\bx_2 \in \Z^{m + h}$ and $x \in \Z$. Using this, we can express $\bB_{\rm int}\bx$ as $\bB_{\rm int}\mathbf{x} = (2\bB\bx_1 - 2x\cdot\by)\circ(\bB_{\rm BCH}\bx_2 - x\cdot \mathbf{s})$.
	
	Suppose $x \neq 0$. Since $(\bB,\by,t)$ is a NO instance, the vector $[2\bB \enskip {\bf 0} \enskip 2\by]\cdot\bx = 2\bB\bx_1 - 2x\cdot\by$ has Hamming weight at least $l$. This would imply that $\bB_{\rm int}\bx$ cannot by annoying, which is a contradiction.  
	
	Henceforth we can assume that $x = 0$; note that we now have $\bB_{\rm int}\mathbf{x} = (2\bB\bx_1) \circ (\bB_{\rm BCH}\bx_2)$. We further claim that, if $\bB_{\rm int}\mathbf{x}$ is annoying, then all of its coordinates must be even. To see that this is the case, let us consider any $\bx$ such that $\bB_{\rm int}\bx$ has at least one odd coordinate. From $\bB_{\rm int}\mathbf{x} = (2\bB\bx_1) \circ (\bB_{\rm BCH}\bx_2)$, $\bB_{\rm BCH}\bx_2$ must have at least one odd coordinate. Let us further write $\bx_2$ as $\bx_2 = \bw_1 \circ \bw_2$ where $\bw_1 \in \Z^m$ and $\bw_2 \in \Z^h$. Notice that $\bB_{\rm BCH} \bx_2 = \bw_1 \circ (Q(\bP_{\rm BCH}\bw_1 - 2\bw_2))$. Since every coordinate of $\bB_{\rm BCH} \bx_2$ must be less than $Q$ in magnitude, it must be the case that $\bP_{\rm BCH}\bw_1 - 2\bw_2 = \vzero$. In other words, $(\bw_1 \mod 2)$ is a codeword of the BCH code. However, since the code has distance $l$, this means that, if $\bw_1$ has at least one odd coordinate, it must have at least $l$ odd (non-zero) coordinates. This contradicts to our assumption that $\bB_{\rm int}\mathbf{x}$ is annoying.
		 
	Therefore, $\bB_{\rm int}\bx$ must have less than $l/2^p$ non-zero coordinates, all of which are bounded in magnitude by ${l}^{10l}$. Hence, we can bound the total number of annoying vectors by
	\begin{eqnarray*}
	\big(2{l}^{10l} + 1\big)^{l/2^p}{{n + h + g}\choose {\lfloor \frac{l}{2^p} \rfloor}} \leqs (2l)^{10l^2} (n + h + g)^{l/2^p}
	\leqs (2l)^{10l^2} ((1 + l/2^p)h)^{l/2^p}
	\leqs e(2l)^{10l^2}h^{l/2^p}
	\end{eqnarray*}
\noindent where the second-to-last step holds since $g \leqs \frac{l}{2}\log h \leqs \frac{l h}{2^{p+1}}$ and $n \leqs \frac{h}{2^{p+1}}$. On the other hand,
	\begin{equation*}
	\frac{h^{r}}{h^{l/2}{l}^{l}} = \frac{h^{\big(\frac12 + \frac1\eta + \frac{1}{2^p}\big)l}}{h^{l/2}{l}^{l}} = h^{l/2^p} (h/l^\eta)^{l/\eta} \geqs h^{l/2^p} (100l)^{99l^2} \geqs 10^5 \left(e(2l)^{10l^2}h^{l/2^p}\right).
	\end{equation*}
\noindent which follows from our choice of $h$. Combining the bounds completes the proof.
\end{proof}

\subsubsection{The $\svp_{p,\gamma}$-instance}					\label{sec:final-lattice}

Finally, we construct $\bB_{\rm svp}$ from $\bB_{\rm int}$ by adding a random homogeneous constraint as in \cite{Khot05}. For ease of notation, let $N_g$ denote the lower bound on the number of distinct coefficient vectors guaranteed by Lemma \ref{lem:good-vectors} in the YES case. Similarly, let $N_a$ denote the upper bound on the number of annoying vectors as given in Lemma \ref{lem:annoying-vectors}. Combining the two Lemmas we have $N_g \geqs 10^5N_a$, which will be used crucially in the construction and analysis of the final lattice.
\paragraph{Construction of the Final Lattice}: Let $\rho$ be any prime number in{\footnote{Note that the density of primes in this range is at least $1/\log N_g = 1/r\log h$. Therefore, a random sample of size $O(r\log h)$ in this range contains a prime with high probability. Since we can test primality for any $\rho \in \Big[10^{-4}N_g, 10^{-2}N_g\Big]$ in FPT time, this gives us an FPT algorithm to sample such a prime number efficiently .}} $\Big[10^{-4}N_g, 10^{-2}N_g\Big]$. Furthermore, let ${\bf r} \overset{\rm unif}{\sim} [0,\rho-1]^{n + h + g}$ be a uniformly sampled lattice point. We construct ${\bf B}_{\rm svp}$ as 
\begin{align*}
	\bB_{\rm svp} =
	\begin{bmatrix}
		\bB_{\rm int} &  {0}  \\
		D\cdot{\bf r}\bB_{\rm int} &  D\cdot \rho 
	\end{bmatrix}
	\in \Z^{(n + h + g + 1) \times (q + h + g + 2)}.
\end{align*}

\noindent where $D = l^{10l}$. In other words, this can be thought of as adding a random linear constraint to the intermediate lattice. The choice of parameters ensures that with good probability, in the YES case, at least one of the good vectors ${\bf x} \in\mathbb{Z}^{q + h + g + 1}$ evaluates to $0$ modulo $\rho$ on the random constraint, and therefore we can pick $u \in \Z$ such that ${\bf B}_{\rm svp}({{\bf x} \circ u}) = (\bB_{\rm int}\bx) \circ 0$ still has small $\ell_p$ norm. On the other hand, since $N_a \ll N _g$, with good probability, all of annoying vectors evaluate to non-zeros, and hence will contribute a coordinate of magnitude $D = l^{10l}$. This intuition is formalized below.


\begin{proof}[Proof of Lemma~\ref{lem:snvp-to-svp}]
	Let $\bB_{\rm svp}$ be the corresponding final lattice of $(\bB,\by,t)$ as described above. Observe that given the $\snvp_{p,\eta}$-instance $(\bB,\by,t)$, we can construct $\bB_{\rm svp}$ in $t^{O(t)}\cdot{\rm poly}(n,q)$-time.

	Suppose that $(\bB,\by,t)$ is a NO instance and fix a vector ${\bf x}\circ u \in \Z^{q + h + g + 2}$. If $\bx$ is not an annoying vector w.r.t. the intermediate lattice $\bB_{\rm int}$, then it follows that $\bx\circ u$ is also not an annoying vector w.r.t. the final lattice $\bB_{\rm svp}$. Therefore, it suffices to look at vectors $\bx \circ u$ such that $\bx$ is annoying w.r.t. $\bB_{\rm int}$. By construction, $(\bB_{\rm svp}(\bx \circ u))_{(n + h + g + 1)} = 0$ implies that $ {\bf r}\bB_{\rm int}\bx = u\cdot \rho$, which in turn is equivalent to ${\bf r}\bB_{\rm int}\bx \equiv 0 \mod \rho$. Since, $\bx_i < \rho$ for every $i \in [q + h + g + 1]$, this event happens with probability exactly $1/\rho$ over the random choice of $\mathbf{r}$ (for any arbitrary $u \in \mathbb{Z}$). Therefore by union bound, for every choice of annoying vector $\bx$ and coefficient $u$, we have ${\bf r}\bB_{\rm int}\bx \neq u\cdot \rho$, with probability at least $1 - N_a/\rho \geqs 0.9$. And therefore, with probability at least $0.9$, for every annoying $\bx$, the last coordinate of $\bB_{\rm svp}(\bx \circ u)$ has magnitude at least $D$.
	
	Next, suppose that $(\bB,\by,t)$ is a YES instance. We condition on the event that there exists at least $N_g$ good vectors as guaranteed by Lemma \ref{lem:good-vectors}. Consider any two distinct vectors $\bx_1,\bx_2 \in \Z^{q + h + g + 1}$ provided by Lemma \ref{lem:good-vectors}. Since $\bB_{\rm int}\bx_1$ and $\bB_{\rm int}\bx_2$ are distinct $\{0, 1, 2\}$-vectors, they are pairwise independent modulo $\rho > 2$. Therefore, instantiating Lemma 5.8 from \cite{Khot05} with the lower bound on the number of good vectors $N_g$, and our choice of $\rho$, it follows that with probability at least $0.9$, there exists a good vector ${\bf x}$ such that ${\bf r}\bB_{\rm int}\bx \equiv 0 \mod \rho$. Therefore, by union bound, with probability at least $0.8$ (over the randomness of Lemma \ref{lem:good-vectors} and the choice of $\mathbf{r},\rho$), there exists $\bx$ such that, for some $u \in \Z$, $\|\bB_{\rm svp}(\bx \circ u)\|_p^p = \|\bB_{\rm int}\bx\|_p^p \leqs 2^p t + r = \gamma_p^{-1}l$, which concludes the proof.
\end{proof}

\subsection{Gap Amplification for $\svp_{2,\gamma}$}	\label{sec:svp-gap-amp}

As in the case of $\mdp$, we can use tensor product of lattices to boost the hardness of $\svp_{2,\gamma}$ to any constant factor. The tensor product of lattices is defined similarly as in the case of linear codes. Given $\bA \in \Z^{m \times m'}$, let $\mathcal{L}(\bA)$ be the lattice generated by $\bA$. Given matrices $\bA \in \Z^{m \times m'}$ and $\bB \in \Z^{n \times n'}$ the tensor product of $\mathcal{L}(\bf A)$ and $\mathcal{L}(\bB)$ is defined as 
\begin{equation}
\mathcal{L}(\bA) \otimes \mathcal{L}(\bB) = \{ \bA\mathbf{X}\bB^\top | \mathbf{X} \in \Z^{m' \times n'}\}
\end{equation}
For brevity, let $\mathcal{L}_1 = \mathcal{L}(\bf A)$ and $\mathcal{L}_2 = \mathcal{L}(\bB)$. The following lemma from \cite{HR07} summarizes the properties of tensor product instances of $\svp_{2,\gamma}$.

\begin{lemma}				\label{lem:lattice-gap-amp}
	Let $\mathcal{L}_1$ and $\mathcal{L}_2$ be as above. Then   we have that $\lambda_2(\mathcal{L}_1 \otimes \mathcal{L}_2) \leqs \lambda_2(\mathcal{L}_1)\cdot\lambda_2 (\mathcal{L}_2)$. Furthermore, suppose every vector $\by \in \mathcal{L}_1$ satisfies at least one of the following conditions:
	\begin{itemize}
		\item The vector $\by$ has Hamming norm at least $l$.
		\item All coordinates of $\by$ are even and $\by$ has Hamming norm at least $l/4$.
		\item All coordinates of $\by$ are even and $\|\by\|_\infty \geqs l^{4l}$.
	\end{itemize}
	Then, $\lambda_2(\mathcal{L}_1 \otimes \mathcal{L}_2) \geqs l^{1/2}\cdot\lambda_2(\mathcal{L}_2)$
\end{lemma}

It is easy to see that given a $\svp_{2,\gamma}$ instance $(\bA,k)$ (as in Lemma \ref{lem:snvp-to-svp}), the tensor product gives an FPT reduction to $\svp_{2,\gamma^2}$. Indeed, if $(\bA,k)$, then $\lambda_2(\mathcal{L}(\bA) \otimes \mathcal{L}(\bA))^2 \leqs k^2$. Conversely, if $(\bA, k)$ is a NO instance, then setting $l = \gamma k$ in Lemma \ref{lem:lattice-gap-amp}, we get $\lambda_2(\mathcal{L}(\bA)  \otimes \mathcal{L}(\bA))^2 \geqs l^2 = \gamma^2k^2$. Therefore, by repeated tensoring, we have proved Theorem~\ref{thm:SVPmain-boost}.

We remark here that, while there is an analogous statement as in Lemma~\ref{lem:lattice-gap-amp} for $p \ne 2$, the third case requires $\|\by\|_\infty$ to not only be large in terms of $l$, but also in terms of the dimension of the lattice. However, to the best of our knowledge, this can only be achieved when $l$ is polynomial in the dimension. This is indeed the reason we fail to amplify the gap to all constants for $p \ne 2$.
\section{Conclusion and Open Questions}
\label{sec:open}
In this work, we have shown the parameterized inapproximability of $k$-Minimum Distance Problem ($k$-\MDP) and $k$-Shortest Vector Problem ($k$-\SVP) in the $\ell_p$ norm for every $p > 1$ based on the Parameterized Inapproximability Hypothesis (PIH), which in turns follows from the Gap Exponential Time Hypothesis (Gap-ETH). While our results give an evidence of intractability of these problems, there are still many questions that remain open. First and foremost, it is still open whether the hardness of both problems can be based on more standard assumptions, such as ETH or $\W[1] \ne \FPT$. On this front, we would like to note that the only reason we need PIH (or Gap-ETH) is to arrive at the inapproximability of the non-homogeneous variants of the problems (Theorems~\ref{thm:MLDbase} and \ref{thm:CVPmain}), which is needed for us even if we want to only rule out exact FPT algorithms for $k$-\MDP\ and $k$-\SVP. Hence, if one could prove the hardness of approximation for these problems under weaker assumptions, then the inapproximability of $k$-\MDP\ and $k$-\SVP\ would still follow.

Another obvious question from our work is whether $k$-\SVP\ in the $\ell_1$ norm is in FPT. Khot's reduction unfortunately does not work for $\ell_1$; indeed, in the work of Haviv and Regev~\cite{HR07}, they arrive at the hardness of approximating \SVP\ in the $\ell_1$ norm by embedding \SVP\ instances in $\ell_2$ to instances in $\ell_1$ using an earlier result of Regev and Rosen~\cite{RR06}. The Regev-Rosen embedding inherently does not work in the FPT regime either, as it produces non-integral lattices. Similar issue applies to an earlier hardness result for \SVP\ on $\ell_1$ of~\cite{Mic00}, whose reduction produces irrational bases.

An additional question regarding $k$-\SVP\ is whether we can prove hardness of approximation for \emph{every} constant factor for $p \ne 2$. As described earlier, the gap amplification techniques of~\cite{Khot05,HR07} require the distance $k$ to be dependent on the input size $nm$, and hence are not applicable for us. To the best of our knowledge, it is unknown whether this dependency is necessary. If they are indeed required, it would also be interesting to see whether other different techniques that work for our settings can be utilized for gap amplification instead of those from~\cite{Khot05,HR07}.

Furthermore, the Minimum Distance Problem can be defined for linear codes in $\mathbb{F}_p$ for any larger field of size $p > 2$ as well. It turns out that our result does not rule out FPT algorithms for $k$-\MDP\ over $\mathbb{F}_p$ with $p > 2$. The issue here is that, in our proof of existence of Sparse Covering Codes (Lemma~\ref{lem:ldc}), we need the co-dimension of the code to be small compared to its distance. In particular, the co-dimension $h - m$ has to be at most $(d/2 + O(1))\log_p h$ where $d$ is the distance. While the BCH code over binary alphabet satisfies this property, we are not aware of any linear codes that satisfy this for larger fields. It is an intriguing open question to determine whether such codes exist, or whether the reduction can be made to work without existence of such codes.

Since the current reductions for both $k$-\MDP\ and $k$-\SVP\ are randomized, it is still an intriguing open question whether we can find deterministic reductions from PIH (or Gap-ETH) to these problems. As stated in the introduction, even in the non-parameterized setting, \NP-hardness of \SVP\ through deterministic reductions is not known. On the other hand, \MDP\ is known to be \NP-hard even to approximate under deterministic reductions; in fact, even the Dumer~\etal's reduction~\cite{DMS03} that we employ can be derandomized, as long as one has a deterministic construction for Locally Dense Codes~\cite{CW12,Mic14}. In our settings, if one can deterministically construct Sparse Covering Codes (i.e. derandomize Lemma~\ref{lem:ldc}), then we would also get a deterministic reduction for $k$-\MDP.

Finally, another interesting research direction is to attempt to prove more concrete running time lower bounds for $k$-\MDP\ and $k$-\SVP. For instance, it is easy to see that $k$-\MDP\ can be solved (exactly) in $N^{O(k)}$ time, where $N = nm$ is the input size. On the other hand, while not stated explicitly above, it is simple to check that our proof implies that $k$-\MDP\  cannot be solved (or even approximated) in time $N^{o(k^c)}$ for some small constant $c > 0$, assuming Gap-ETH. Would it be possible to improve this running time lower bound to the tight $N^{o(k)}$? Similar questions also apply to $k$-\SVP.

\subsection*{Acknowledgment}

We are grateful to Ishay Haviv for providing insights on how the gap amplification for $p \ne 2$ from~\cite{HR07} works. Pasin would like to thank Danupon Nanongkai for introducing him to the $k$-Even Set problem and for subsequent useful discussions.
\bibliographystyle{alpha}
\bibliography{main}

\newcommand{\etalchar}[1]{$^{#1}$}
\begin{thebibliography}{DFVW99}

\bibitem[ABSS97]{ABSS97}
Sanjeev Arora, L{\'{a}}szl{\'{o}} Babai, Jacques Stern, and Z.~Sweedyk.
\newblock The hardness of approximate optima in lattices, codes, and systems of
  linear equations.
\newblock {\em J. Comput. Syst. Sci.}, 54(2):317--331, 1997.

\bibitem[AD97]{AD97}
Mikl{\'{o}}s Ajtai and Cynthia Dwork.
\newblock A public-key cryptosystem with worst-case/average-case equivalence.
\newblock In {\em STOC}, pages 284--293, 1997.

\bibitem[Ajt96]{Ajt96}
Mikl{\'{o}}s Ajtai.
\newblock Generating hard instances of lattice problems (extended abstract).
\newblock In {\em STOC}, pages 99--108, 1996.

\bibitem[Ajt98]{Ajt98}
Mikl{\'{o}}s Ajtai.
\newblock The shortest vector problem in $\ell_2$ is {NP}-hard for randomized
  reductions (extended abstract).
\newblock In {\em STOC}, pages 10--19, 1998.

\bibitem[AK14]{AK14}
Per Austrin and Subhash Khot.
\newblock A simple deterministic reduction for the gap minimum distance of code
  problem.
\newblock {\em {IEEE} Trans. Information Theory}, 60(10):6636--6645, 2014.

\bibitem[App17]{App17}
Benny Applebaum.
\newblock Exponentially-hard gap-{CSP} and local {PRG} via local hardcore
  functions.
\newblock In {\em FOCS}, pages 836--847, 2017.

\bibitem[AS18]{AD17}
Divesh Aggarwal and Noah Stephens{-}Davidowitz.
\newblock (gap/s)eth hardness of {SVP}.
\newblock In {\em STOC}, 2018.
\newblock To appear.

\bibitem[BGGS16]{BGGS16}
Arnab Bhattacharyya, Ameet Gadekar, Suprovat Ghoshal, and Rishi Saket.
\newblock On the hardness of learning sparse parities.
\newblock In {\em ESA}, pages 11:1--11:17, 2016.

\bibitem[BGS17]{BGS17}
Huck Bennett, Alexander Golovnev, and Noah Stephens{-}Davidowitz.
\newblock On the quantitative hardness of {CVP}.
\newblock In {\em FOCS}, pages 13--24, 2017.

\bibitem[BMvT78]{BMT78}
Elwyn~R. Berlekamp, Robert~J. McEliece, and Henk C.~A. van Tilborg.
\newblock On the inherent intractability of certain coding problems (corresp.).
\newblock {\em {IEEE} Trans. Information Theory}, 24(3):384--386, 1978.

\bibitem[BR60]{BR60}
R.~C. Bose and Dwijendra~K. Ray{-}Chaudhuri.
\newblock On a class of error correcting binary group codes.
\newblock {\em Information and Control}, 3(1):68--79, 1960.

\bibitem[BS08]{BS08}
Eli Ben{-}Sasson and Madhu Sudan.
\newblock Short {PCP}s with polylog query complexity.
\newblock {\em {SIAM} J. Comput.}, 38(2):551--607, 2008.

\bibitem[CCK{\etalchar{+}}17]{CCKLMNT17}
Parinya Chalermsook, Marek Cygan, Guy Kortsarz, Bundit Laekhanukit, Pasin
  Manurangsi, Danupon Nanongkai, and Luca Trevisan.
\newblock From {Gap-ETH} to {FPT}-inapproximability: Clique, dominating set,
  and more.
\newblock In {\em FOCS}, pages 743--754, 2017.

\bibitem[CFHW17]{CFHW17}
Marek Cygan, Fedor~V. Fomin, Danny Hermelin, and Magnus Wahlstr{\"{o}}m.
\newblock Randomization in parameterized complexity (dagstuhl seminar 17041).
\newblock {\em Dagstuhl Reports}, 7(1):103--128, 2017.

\bibitem[CFJ{\etalchar{+}}14]{CFJKLMPS14}
Marek Cygan, Fedor Fomin, Bart~MP Jansen, Lukasz Kowalik, Daniel Lokshtanov,
  D{\'a}niel Marx, Marcin Pilipczuk, and Saket Saurabh.
\newblock Open problems for fpt school 2014.
\newblock 2014.

\bibitem[CFK{\etalchar{+}}15]{CFKLMPPS15}
Marek Cygan, Fedor~V. Fomin, Lukasz Kowalik, Daniel Lokshtanov, D{\'{a}}niel
  Marx, Marcin Pilipczuk, Michal Pilipczuk, and Saket Saurabh.
\newblock {\em Parameterized Algorithms}.
\newblock Springer, 2015.

\bibitem[CFM17]{CFM17}
Rajesh Chitnis, Andreas~Emil Feldmann, and Pasin Manurangsi.
\newblock Parameterized approximation algorithms for directed steiner network
  problems.
\newblock {\em CoRR}, abs/1707.06499, 2017.

\bibitem[CN99]{CN99}
Jin{-}yi Cai and Ajay Nerurkar.
\newblock Approximating the {SVP} to within a factor $(1+1/\text{dim}^{\xi})$
  is {NP}-hard under randomized reductions.
\newblock {\em J. Comput. Syst. Sci.}, 59(2):221--239, 1999.

\bibitem[CW12]{CW12}
Qi~Cheng and Daqing Wan.
\newblock A deterministic reduction for the gap minimum distance problem.
\newblock {\em {IEEE} Trans. Information Theory}, 58(11):6935--6941, 2012.

\bibitem[DF99]{DF99}
Rodney~G. Downey and Michael~R. Fellows.
\newblock {\em Parameterized Complexity}.
\newblock Monographs in Computer Science. Springer, 1999.

\bibitem[DF13]{DF13}
Rodney~G. Downey and Michael~R. Fellows.
\newblock {\em Fundamentals of Parameterized Complexity}.
\newblock Texts in Computer Science. Springer, 2013.

\bibitem[DFVW99]{DFVW99}
Rodney~G. Downey, Michael~R. Fellows, Alexander Vardy, and Geoff Whittle.
\newblock The parametrized complexity of some fundamental problems in coding
  theory.
\newblock {\em {SIAM} J. Comput.}, 29(2):545--570, 1999.

\bibitem[DGMS07]{DGMS07}
Erik~D. Demaine, Gregory Gutin, D{\'{a}}niel Marx, and Ulrike Stege.
\newblock 07281 open problems -- structure theory and {FPT} algorithmcs for
  graphs, digraphs and hypergraphs.
\newblock In {\em Structure Theory and {FPT} Algorithmics for Graphs, Digraphs
  and Hypergraphs, 08.07. - 13.07.2007}, 2007.

\bibitem[Din02]{Din02}
Irit Dinur.
\newblock Approximating $\text{SVP}_{\infty}$ to within almost-polynomial
  factors is {NP}-hard.
\newblock {\em Theor. Comput. Sci.}, 285(1):55--71, 2002.

\bibitem[Din07]{D07}
Irit Dinur.
\newblock The {PCP} theorem by gap amplification.
\newblock {\em J. {ACM}}, 54(3):12, 2007.

\bibitem[Din16]{D16}
Irit Dinur.
\newblock Mildly exponential reduction from gap {3SAT} to polynomial-gap
  label-cover.
\newblock {\em ECCC}, 23:128, 2016.

\bibitem[DKRS03]{DKRS03}
Irit Dinur, Guy Kindler, Ran Raz, and Shmuel Safra.
\newblock Approximating {CVP} to within almost-polynomial factors is {NP}-hard.
\newblock {\em Combinatorica}, 23(2):205--243, 2003.

\bibitem[DM18]{DM18}
Irit Dinur and Pasin Manurangsi.
\newblock {ETH}-hardness of approximating 2-{CSP}s and directed steiner
  network.
\newblock In {\em ITCS}, pages 36:1--36:20, 2018.

\bibitem[DMS03]{DMS03}
Ilya Dumer, Daniele Micciancio, and Madhu Sudan.
\newblock Hardness of approximating the minimum distance of a linear code.
\newblock {\em {IEEE} Trans. Information Theory}, 49(1):22--37, 2003.

\bibitem[FGMS12]{FGMS12}
Michael~R. Fellows, Jiong Guo, D{\'{a}}niel Marx, and Saket Saurabh.
\newblock Data reduction and problem kernels (dagstuhl seminar 12241).
\newblock {\em Dagstuhl Reports}, 2(6):26--50, 2012.

\bibitem[GKS12]{GKS12}
Petr~A. Golovach, Jan Kratochv{\'{\i}}l, and Ondrej Such{\'{y}}.
\newblock Parameterized complexity of generalized domination problems.
\newblock {\em Discrete Applied Mathematics}, 160(6):780--792, 2012.

\bibitem[GMSS99]{GMSS99}
Oded Goldreich, Daniele Micciancio, Shmuel Safra, and Jean{-}Pierre Seifert.
\newblock Approximating shortest lattice vectors is not harder than
  approximating closest lattice vectors.
\newblock {\em Inf. Process. Lett.}, 71(2):55--61, 1999.

\bibitem[Gol06]{Gol06}
Oded Goldreich.
\newblock On promise problems: {A} survey.
\newblock In {\em Theoretical Computer Science, Essays in Memory of Shimon
  Even}, pages 254--290, 2006.

\bibitem[Hoc59]{H59}
Alexis Hocquenghem.
\newblock Codes correcteurs d’erreurs.
\newblock {\em Chiffres}, 2:147–156, September 1959.

\bibitem[HR07]{HR07}
Ishay Haviv and Oded Regev.
\newblock Tensor-based hardness of the shortest vector problem to within almost
  polynomial factors.
\newblock In {\em STOC}, pages 469--477, 2007.

\bibitem[IP01]{IP01}
Russell Impagliazzo and Ramamohan Paturi.
\newblock On the complexity of k-{SAT}.
\newblock {\em J. Comput. Syst. Sci.}, 62(2):367--375, 2001.

\bibitem[IPZ01]{IPZ01}
Russell Impagliazzo, Ramamohan Paturi, and Francis Zane.
\newblock Which problems have strongly exponential complexity?
\newblock {\em J. Comput. Syst. Sci.}, 63(4):512--530, 2001.

\bibitem[Joh90]{J90}
D.~S. Johnson.
\newblock Handbook of theoretical computer science.
\newblock volume A (Algorithms and Complexity), chapter 2, A catalog of
  complexity classes, pages 67--161. Elseveir, 1990.

\bibitem[Kho05]{Khot05}
Subhash Khot.
\newblock Hardness of approximating the shortest vector problem in lattices.
\newblock {\em J. ACM}, 52(5):789--808, 2005.

\bibitem[KLM18]{KLM18}
{Karthik {C. S.}}, Bundit Laekhanukit, and Pasin Manurangsi.
\newblock On the parameterized complexity of approximating dominating set.
\newblock In {\em STOC}, 2018.
\newblock To appear.

\bibitem[Len83]{Len83}
Hendrik~Willem Lenstra.
\newblock Integer programming with a fixed number of variables.
\newblock {\em Math. Oper. Res.}, 8(4):538--548, 1983.

\bibitem[Lin15]{Lin15}
Bingkai Lin.
\newblock The parameterized complexity of \emph{k}-{Biclique}.
\newblock In {\em SODA}, pages 605--615, 2015.

\bibitem[LLL82]{LLL}
Arjen~Klaas Lenstra, Hendrik~Willem Lenstra, and L{\'a}szl{\'o} Lov{\'a}sz.
\newblock Factoring polynomials with rational coefficients.
\newblock {\em Mathematische Annalen}, 261(4):515--534, 1982.

\bibitem[LRSZ17]{LRSZ17}
Daniel Lokshtanov, M.~S. Ramanujan, Saket Saurabh, and Meirav Zehavi.
\newblock Parameterized complexity and approximability of directed odd cycle
  transversal.
\newblock {\em CoRR}, abs/1704.04249, 2017.

\bibitem[Maj17]{Maj17}
Ruhollah Majdoddin.
\newblock Parameterized complexity of {CSP} for infinite constraint languages.
\newblock {\em CoRR}, abs/1706.10153, 2017.

\bibitem[Man17]{M17}
Pasin Manurangsi.
\newblock Almost-polynomial ratio {ETH}-hardness of approximating densest
  $k$-subgraph.
\newblock In {\em STOC}, pages 954--961, 2017.

\bibitem[Mic00]{Mic00}
Daniele Micciancio.
\newblock The shortest vector in a lattice is hard to approximate to within
  some constant.
\newblock {\em {SIAM} J. Comput.}, 30(6):2008--2035, 2000.

\bibitem[Mic01]{Mic01}
Daniele Micciancio.
\newblock The hardness of the closest vector problem with preprocessing.
\newblock {\em {IEEE} Trans. Information Theory}, 47(3):1212--1215, 2001.

\bibitem[Mic12]{Mic12}
Daniele Micciancio.
\newblock Inapproximability of the shortest vector problem: Toward a
  deterministic reduction.
\newblock {\em Theory of Computing}, 8(1):487--512, 2012.

\bibitem[Mic14]{Mic14}
Daniele Micciancio.
\newblock Locally dense codes.
\newblock In {\em CCC}, pages 90--97, 2014.

\bibitem[MR09]{MR-survey}
Daniele Micciancio and Oded Regev.
\newblock Lattice-based cryptography.
\newblock In {\em Post-quantum cryptography}, pages 147--191. Springer, 2009.

\bibitem[MR16]{MR16}
Pasin Manurangsi and Prasad Raghavendra.
\newblock A birthday repetition theorem and complexity of approximating dense
  {CSP}s.
\newblock {\em CoRR}, abs/1607.02986, 2016.

\bibitem[NV10]{NV10}
Phong~Q. Nguyen and Brigitte Vall{\'{e}}e, editors.
\newblock {\em The {LLL} Algorithm - Survey and Applications}.
\newblock Information Security and Cryptography. Springer, 2010.

\bibitem[Reg03]{Reg03}
Oded Regev.
\newblock New lattice based cryptographic constructions.
\newblock In {\em STOC}, pages 407--416, 2003.

\bibitem[Reg05]{Reg05}
Oded Regev.
\newblock On lattices, learning with errors, random linear codes, and
  cryptography.
\newblock In {\em STOC}, pages 84--93, 2005.

\bibitem[Reg06]{Reg06}
Oded Regev.
\newblock Lattice-based cryptography.
\newblock In {\em CRYPTO}, pages 131--141, 2006.

\bibitem[Reg10]{Reg10}
Oded Regev.
\newblock The learning with errors problem (invited survey).
\newblock In {\em CCC}, pages 191--204, 2010.

\bibitem[RR06]{RR06}
Oded Regev and Ricky Rosen.
\newblock Lattice problems and norm embeddings.
\newblock In {\em STOC}, pages 447--456, 2006.

\bibitem[Ste93]{Ste93}
Jacques Stern.
\newblock Approximating the number of error locations within a constant ratio
  is {NP}-complete.
\newblock In {\em AAECC}, pages 325--331, 1993.

\bibitem[Var97a]{Var97a}
Alexander Vardy.
\newblock Algorithmic complexity in coding theory and the minimum distance
  problem.
\newblock In {\em STOC}, pages 92--109, 1997.

\bibitem[Var97b]{Var97b}
Alexander Vardy.
\newblock The intractability of computing the minimum distance of a code.
\newblock {\em {IEEE} Trans. Information Theory}, 43(6):1757--1766, 1997.

\bibitem[vEB81]{VEB}
Peter van Emde-Boas.
\newblock {\em Another {NP}-complete partition problem and the complexity of
  computing short vectors in a lattice}.
\newblock Report. Department of Mathematics. University of Amsterdam.
  Department, Univ., 1981.

\end{thebibliography}

\appendix
\section{A Simple Proof of PIH from Gap-ETH} \label{app:gap-eth}

In this section, we provide a simple proof of PIH assuming Gap-ETH. The proof is folklore and is known in the community, although we are not aware of it being fully written down anywhere. The main technical ingredient is the following subexponential time reduction from 3SAT to 2CSP:

\begin{lemma}
\label{red:gap-eth}
For any $\varepsilon > 0$ and $\Delta, k \in \N$ such that $k \geqs 2$, there exists a randomized reduction that, given any 3SAT instance $\Phi$ on $n$ variables where each variable appears in at most $\Delta$ clauses, produces a 2CSP instance $\Gamma = (G = (V, E), \Sigma, \{C_{uv}\}_{(u, v) \in E})$ such that
\begin{itemize}
\item (YES) If $\sat(\Phi) = 1$, then $\val(\Gamma) = 1$.
\item (NO) If $\sat(\Phi) < 1 - \varepsilon$, then $\val(\Gamma) < 1 - \frac{\varepsilon}{3000 \Delta^4}$ with high probability.
\item (Running Time) The reduction runs in time $2^{O_{\Delta,\varepsilon}(n/k)}$, and,
\item (Parameter) The number of variables of $\Gamma$ is equal to $k$, i.e., $|V| = k$.
\end{itemize}
\end{lemma}

Note that, for Gap-ETH, it can be assumed without loss of generality that each variable of the 3CNF formula appears in at most $O(1)$ number of clauses (see e.g.~\cite[page 21]{MR16}). Observe that, since the reduction runs in time $2^{O(n/k)}$; $\Gamma$ is of size at most $2^{O(n/k)}$. This yields the following corollary, which in fact provides a running time lower bound which is even stronger than PIH:

\begin{corollary}
Assuming Gap-ETH, there exists $\varepsilon > 0$ such that, for any computable function $T$, there is no $T(k) \cdot |\Sigma|^{o(k)}$ algorithm for $\csp_{\varepsilon}$.
\end{corollary}

We note here that the running time lower bound in the above lemma is asymptotically optimal, because we can solve 2CSP exactly in $\poly(k) \cdot |\Sigma|^k$ time by just enumerating all possible assignments and outputting the one that satisfies the maximum number of constraints.

We now turn our attention to the proof of Lemma~\ref{red:gap-eth}. 

\begin{proof}[Proof of Lemma~\ref{red:gap-eth}]
The reduction is simple and is in fact a special case of the reduction from~\cite{DM18}. Let $\cX$ and $\cC$ denote the set of variables and the set of clauses of $\Phi$ respectively. We assume w.l.o.g. that each variable appears in at least one clause. For each $C \in \cC$, we use $\var(C)$ to denote the variables appearing in $C$. Furthermore, for each $S \subseteq \cC$, let $\var(S) = \cup_{C \in S} \var(S)$. For convenience, we will also assume that $m$ is divisible by $k$ where $m = |\cC|$ is the number of clauses.

Randomly partition $\cC$ into $k$ parts $S_1, \dots, S_k$ each of size $m/k$. We can then define $\Gamma$ as follows:
\begin{itemize}
\item The vertex set $V$ is $[k]$ and the constraint graph $G$ is the complete graph over $V$.
\item Let $\Sigma = 2^{3(m/k)}$. For each $i \in [k]$, we associate elements of $\Sigma$ with the partial assignments $f_i: \var(S_i) \to \{0, 1\}$ that satisfies all clauses in $S_i$. In other words, we view $\Sigma$ as $\Sigma_i = \{f_i: \var(S_i) \to \{0, 1\} \mid f_i \text{ satisfies all } C \in S_i\}$.
\item For each edge $\{i, j\} \subseteq [k]$, we define the constraint $C_{ij}$ to contain all $f_i \in \Sigma_i$ and $f_j \in \Sigma_j$ that agree with each other, i.e., $C_{ij} = \{(f_i, f_j) \in \Sigma_i \times \Sigma_j \mid \forall x \in \var(S_i) \cap \var(S_j), f_i(x) = f_j(x)\}$.
\end{itemize}

It is clear that the reduction runs in $2^{O(m/k)} = 2^{O(\Delta n / (\varepsilon k))}$ time and $|V| = k$. Further, in the YES case where there exists an assignment $f: \cX \to \{0, 1\}$ that satisfies all the clauses of $\Phi$, we can simply pick $f_i$ to be $f|_{\var(S_i)}$, the restriction of $f$ on $\var(S_i)$. Clearly, $\psi$ is a valid assignment for $\Gamma$ and it satisfies all the constraints. Hence, $\val(\Gamma) = 1$.

To show the NO case of the reduction, we will need an additional property of a random partition $S_1, \cdots, S_k$, which is formally stated below and proved later:

\begin{proposition}\label{prop:appprob}
With high probability, for every $i \ne j$, we have $\var(S_i) \cap \var(S_j) < \frac{1000n\Delta^3}{k^2}$.
\end{proposition}

Now, conditioned on $\var(S_i) \cap \var(S_j) < \frac{1000n\Delta^3}{k^2}$ for every $i \ne j \in [k]$, we will prove the NO case by showing the contrapositive of the statement. Suppose that there exists an assignment $\psi: V \to \Sigma$ such that $\val(\psi) \geqs 1 - \varepsilon'$ where $\varepsilon' = \frac{\varepsilon}{3000 \Delta^4}$. For notational convenience, let $f_i = \psi(i)$. Observe that $\val(\psi) \geqs 1 - \varepsilon'$ is simply equivalent to $\Pr_{\{i, j\} \subseteq [k]}[f_i \nsim f_j] \leqs \varepsilon'$ where $f_i \nsim f_j$ denote $f_i(x) \ne f_j(x)$ for some $x \in \var(S_i) \cap \var(S_j)$.

We call a variable $x \in \cX$ \emph{globally consistent} if $f_i(x)$ is equal for every $i \in [k]$ such that $\var(S_i)$ contains $x$. We claim that the number of globally consistent variables is at least $(1 - 1000\varepsilon'\Delta^3)n$. To see this, let us count the number of $(x, \{i, j\}) \in \cX \times \binom{[k]}{2}$ such that $x \in \var(S_i) \cap \var(S_j)$ and $f_i(x) \ne f_j(x)$. Notice that such tuple exists only for $\{i, j\}$ such that $f_i \nsim f_j$, and there are only $\varepsilon' \binom{k}{2}$ such pairs $\{i, j\}$'s. Recall also that $\var(S_i) \cap \var(S_j) \leqs \frac{1000n\Delta^3}{k^2}$ for any $i \ne j \in [k]$. As a result, the number of such $(x, \{i, j\})$'s is at most $\varepsilon' \binom{k}{2} \cdot \frac{1000n\Delta^3}{k^2} \leqs 1000 \varepsilon' \Delta^3 n$. On the other hand, each variable $x \in \cX$ that is not globally consistent must contribute to at least one such pair; hence, at most $1000\varepsilon'\Delta^3 n$ variables are not globally consistent.

Let us define a global assignment $f: \cX \to \{0, 1\}$ as follows: for each $x \in \cX$, pick an arbitrary partition $S_i$ such that $x \in \var(S_i)$ and let $f(x) = f_i(x)$. Observe that, for every clause $C$ such that all its variables are globally consistent, $C$ is satisfied by $f$; this is simply because, if $C$ is in partition $S_j$, then $f(x) = f_j(x)$ for all $x \in \var(C)$ due to global consistency of $x$ and $f_j$ must satisfy $C$ from our definition of $\Gamma$. As a result, $f$ satisfies all but at most $1000\varepsilon'\Delta^4 n \leqs 3000\varepsilon'\Delta^4 m$ clauses, which means that $\val(\Phi) \geqs 1 - 3000\varepsilon'\Delta^4 = 1 - \varepsilon$ as desired.
\end{proof}

\begin{proof}[Proof of Proposition~\ref{prop:appprob}]
For convenience, let $\ell = \frac{1000n\Delta^3}{k^2}$.
Let us fix a pair $i \ne j$. We will show that $\var(S_i) \cap \var(S_j) < \ell$ with high probability; taking union bound over all $\{i, j\}$'s yields the desired statement.

To bound $\Pr[\var(S_i) \cap \var(S_j) < \ell]$, let us define one more notation: we call a set $T \subseteq \cX$ of variables \emph{independent} if no two variables in $T$ appears in the same clause. We claim that, for any subset $T'$ of at least $\ell$ variables, there exists an independent subset $T \subseteq T'$ of at least $\lceil \ell / (2\Delta + 1) \rceil$ variables. This is true because we can use the following algorithm to find $T$: as long as $T'$ is not empty, pick an arbitrary variable $x \in T'$, put $x$ into $T$, and remove $x$ and all variables that share at least one clause with $x$ from $T'$. Since each time we add an element to $T$, we remove at most $2\Delta + 1$ elements from $T'$; we can conclude that $|T'| \geqs \lceil \ell / (2\Delta + 1) \rceil$ as desired.

Hence, to show that $\var(S_i) \cap \var(S_j) < \ell$ w.h.p., it suffices to show that $\var(S_i) \cap \var(S_j)$ does not contain an independent set of variable of size $\lceil \ell / (2\Delta + 1) \rceil$ w.h.p.

Let $r = \lceil \ell / (2\Delta + 1) \rceil$; notice that $r \geqs \frac{100 n \Delta^2}{k^2}$. Consider any independent subset of variables $T \subseteq \cX$ of size $r$. We will upper bound the probability that $T$ is contained in $\var(S_i) \cap \var(S_j)$. To do so, first observe that the events $T \subseteq \var(S_i)$ and $T \subseteq \var(S_j)$ are negatively correlated, i.e., 
\begin{align} \label{eq:tmp1}
\Pr[T \subseteq \var(S_i) \cap \var(S_j)] \leqs \Pr[T \subseteq \var(S_i)]\Pr[T \subseteq \var(S_j)] = \left(\Pr[T \subseteq \var(S_i)]\right)^2.
\end{align}

Observe that since two distinct $x, x' \in T$ do not appear in the same clause, the events $x \in \var(S_i)$ and $x' \in \var(S_i)$ are negatively correlated. That is, if $T = \{x_{p_1}, \dots, x_{p_r}\}$, then we can bound $\Pr[T \subseteq \var(S_i)]$ as follows:
\begin{align} \label{eq:tmp2}
\Pr[T \subseteq \var(S_i)] = \prod_{j \in [r]} \Pr[x_{p_j} \in S_i \mid x_{p_1}, \dots, x_{p_{j - 1}} \in S_i] \leqs \prod_{j \in [r]} \Pr[x_{p_j} \in S_i].
\end{align}

Now, consider any $x \in T$. Since $x$ is contained in at most $\Delta$ clauses, we have 
\begin{align} \label{eq:tmp3}
\Pr[x \in S_i] \leqs 1 - \frac{\binom{m - \Delta}{m/k}}{\binom{m}{m/k}} \leqs 1 - \left(1 - \frac{\Delta}{m - m/k}\right)^{m/k} \leqs \frac{\Delta(m/k)}{m - m/k} = \frac{\Delta}{k - 1}.
\end{align}
where the last inequality follows from Bernoulli's inequality.

Combining (\ref{eq:tmp1}), (\ref{eq:tmp2}) and (\ref{eq:tmp3}), we have
\begin{align*}
\Pr[T \subseteq \var(S_i) \cap \var(S_j)] \leqs \left(\frac{\Delta}{k - 1}\right)^{2r}.
\end{align*}

Finally, by taking union bound over all $r$-size independent sets of variables, the probability that $\var(S_i) \cap \var(S_j)$ contains an $r$-size independent sets of variables is at most
\begin{align*}
\binom{n}{r} \left(\frac{\Delta}{k - 1}\right)^{2r} \leqs \left(\frac{e n}{r}\right)^{r} \left(\frac{\Delta}{k - 1}\right)^{2r} = \left(\frac{e n \Delta^2}{r(k - 1)^2}\right)^{r} \leqs (1/2)^r = o(1),
\end{align*}  
which concludes our proof.
\end{proof}
\newcommand{\bC}{\mathbf{C}}
\section{FPT Inapproximability of $\snvp$}			\label{sec:csp-to-snvp}

In this section, we will prove the inapproximability of $\snvp$ (Theorem~\ref{thm:CVPmain}). The reduction goes through the following intermediate problem $\lvec$ (analogous to the $\gapvec$):

\begin{framed}
	$\eta$-Gap Lattice Vector Sum ($\lvec_{\eta}$)
	
	{\bf Input: } A matrix $\bA \in \Z^{m \times n}$, a vector $\by \in \{0,1\}^m$ and a parameter $k \in \N$.
	
	{\bf Parameter: } $k$
	
	{\bf Question: } Distinguish between the following two cases:
	\begin{itemize}
		\item (YES) There exists $\bx \in \{0,1\}^{n}$ such that $\|\bx\|_0 \leqs k$ and $\bA \bx = \by$.
		\item (NO) For all choices of $\bx \in \Z^n$ such that $\|\bx\|_0 \leqs \eta\cdot k$ and $r \in \Z \setminus\{0\}$, we have $\bA\bx \neq r\cdot\by$.
	\end{itemize}
\end{framed}

The above problem is well defined for any $\eta \geqs 1$. The reduction for Theorem~\ref{thm:CVPmain} goes through the following steps. First, we use Lemma \ref{lem:csp-to-lvec} to give an FPT reduction from $\csp_\epsilon$ to $\lvec_{1 + \epsilon/3}$. Next, we use the gap amplification operation (Lemma \ref{lem:lvec-gap-amp}) to give a FPT reduction from $\lvec_{1 + \epsilon/3}$ to all $\lvec_{\eta}$ for all $\eta \ge 1$. Finally, we use Lemma \ref{lem:lvec-to-snvp} to reduce $\lvec_\eta$ to $\snvp_{p,\eta}$ (for any $p \geqs 1$).

\begin{lemma} \label{lem:csp-to-lvec}
	For all $\epsilon > 0$, there is an FPT reduction from $\csp_\epsilon$ to $\lvec_{1 + \epsilon/3}$.
\end{lemma}
\begin{proof}
	Let $\Gamma = (G := (V,E),\Sigma,\{C_{uv}\}_{\{(u,v) \in E\}})$ be the $\csp_\epsilon$ instance. For clarity, we assume that $G$ is directed. From $\Gamma$, we construct $\bA \in \{-1,0,1\}^{m \times n}, \by \in \{0, 1\}^m$ where $m = |V| + |E| + 2|E||\Sigma|$ and $n = |V||\Sigma| + \sum_{(u,v) \in E}|C_{uv}|$ exactly the same as in the proof of Lemma \ref{lem:csp-to-gapvec} with only one difference: we set $\bA_{(e, \sigma_0, 0), (e, \sigma_0, \sigma_1)}$ and $\bA_{(e, \sigma_1, 1), (e, \sigma_0, \sigma_1)}$ to be -1 instead of 1 as in Lemma~\ref{lem:csp-to-gapvec}.

	We set the sparsity parameter to be $k = |V| + |E|$. Note that the matrix $\bA$ and the vector $\by$ here are $\{-1, 0,1\}$-objects over $\Z$ instead of $\F$. Additionally, by the construction of $\bA$, for every  $\bx \in \Z^n$, the vector $\bA\bx$ must satisfy equations (\ref{eq:row-1}) and (\ref{eq:row-2}) from Section \ref{sec:csp-to-gapvec} whereas, due to our change, equation (\ref{eq:row-3}) now becomes 
\begin{align*}
(\bA\bx)_{(e, \sigma, b)} = \bx_{(u_b, \sigma)} - \sum_{(\sigma_0, \sigma_1) \in C_e \atop \sigma_b = \sigma} \bx_{(e, \sigma_0, \sigma_1)}.
\end{align*}
	 We now argue completeness and soundness of the reduction.
	
	\paragraph{Completeness.} If $\Gamma$ is a YES instance, let $\psi:V \mapsto \Sigma$ be the satisfying labeling. We construct $\bx \in \{0,1\}^n$ as follows: we set $\bx_{(u,\psi(u))} = 1$ for every vertex $u \in V$. Additionally, for every $e = (u,v) \in E$, we set $\bx_{((u,v),\psi(u),\psi(v))} = 1$, and all the remaining coordinates of $\bx$ are set to zero. It is easy to see that $\bx$ as constructed is $(|V| + |E|)$-sparse, and satisfies $\bA\bx = \by$. 
	
	\paragraph{Soundness.} Let $\bx \in \Z^n$ be such that $\|\bx\|_0 \leqs \big(1 + \epsilon/3\big)k$ and $\bA \bx = r\cdot\by$ for some $r \neq 0$. As before, for every vertex $u \in V$, we construct the set $S_u :{=}\{\sigma | \bx_{(u,\sigma)} \neq 0\}$. Similarly, for every edge $e \in E$, we construct the set $T_e := \{(\sigma_0,\sigma_1) | \bx_{(e,\sigma_0,\sigma_1)} \neq 0\}$. By construction it follows that for every $u \in V$, $(\bA\bx)_u = y_u = r \neq 0$ and therefore $|S_u| \geqs 1$. Similarly, for every edge $e \in E$ we have $|T_e| \geqs 1$. Furthermore, we define $E_\uni := \{e \in E| |T_e| = 1\}$. Since the arguments used in Equations (\ref{eq:start})-(\ref{eq:end}) (from proof of Lemma \ref{lem:csp-to-gapvec}) depend only on the fact that $\|\bx\|_0 \leqs (1 + \epsilon/3)k$, again we have $|E_\uni| \geqs (1 - \epsilon)|E|$. 
	
	Now, consider any labeling $\psi:V \mapsto \Sigma$ such that $\psi(u) \in S_u$ for every $u \in V$. We shall show that any such labeling must satisfy every edge $e \in E_\uni$. Towards that end, we fix an edge $e = (u,v) \in E_\uni$ and let $T_e = \{(\sigma^*_0,\sigma^*_1)\}$. Since $|T_e| = 1$, using equation \eqref{eq:row-2} we have 
	\begin{equation*}
	r = \by_e = \sum_{(\sigma_0,\sigma_1) \in T_e} \bx_{(e,\sigma_0,\sigma_1)} = \bx_{(e,\sigma^*_0,\sigma^*_1)} 
	\end{equation*}	

	 Observe that $S_u = \{\sigma^*_0\}$; this follows immediately from substituting (the new) equation \eqref{eq:row-3} with $\sigma = \sigma^*_0, b = 0$ and from $\by_{(e,\sigma^*_0,0)} = 0$. Therefore, $\psi(u) = \sigma^*_0$ and, similarly, $\psi(v) = \sigma^*_1$. Furthermore, since $(\sigma^*_0,\sigma^*_1) \in C_{(u,v)}$, $\psi$ satisfies the edge $e$. Since, this happens for every $e \in E_\uni$, the labeling $\psi$ satisfies every constraint in $E_\uni$. This completes the proof.
\end{proof}

We now give the lemma for the gap amplification step.

\begin{lemma}					\label{lem:lvec-gap-amp}
	For any $c \in \N$ and $\eta \geqs 1$, there exists an FPT reduction from $\lvec_{\eta}$ to $\lvec_{\eta^c} $.
\end{lemma}	

\begin{proof}
	As in Section \ref{sec:gap-amplification}, we define a composition operation $\oplus_{\cL}$ on $\lvec$ instances. Given $\lvec_{\eta_1}$ instance $(\bA,\by_1,k_1)$ (where $\bA \in \Z^{u \times v}, \by_1 \in \Z^u$) and $\lvec_{\eta_2}$ instance $(\bB,\by_2,k_2)$ (where $\bB \in \Z^{u' \times v'}, \by_2 \in \Z^{u'}$), we denote the composed instance as $(\bC,\by,k') = (\bA,\by_1,k_1)  \oplus_{\cL}  (\bB,\by_2,k_2)$, with $k' = k_2 + k_1k_2$. The matrix $\bC \in \Z^{(u' + uv') \times (v' + vv')}$ and the vector $\by \in \Z^{u' + uv'}$ is constructed exactly as in Section \ref{sec:gap-amplification}. We argue the completeness and soundness properties of the operator.
	
	\paragraph{Completeness}: Let $(\bA,\by_1,k_1)$ and $(\bB,\by_2,k_2)$ be YES instances. Then there exists vector $\bx_1 \in \{0,1\}^v$ such that $\|\bx_1\|_0 \leqs k_1$ and $\bA\bx_1 = \by_1$. Similarly there exists vector $\bx_2 \in \{0,1\}^{v'}$ such that $\|\bx_2\|_0 \leqs k_2$ and $\bB\bx_2 = \by_2$. Borrowing notation from Section \ref{sec:gap-amplification}, for any $\bx \in \Z^n$, we define the sub-vectors $\bx^0,\bx^1,\ldots ,\bx^{v'}$. We construct the vector $\bx \in \{0,1\}^{v' + vv'}$ as follows. We set $\bx^0 = \bx_2$ and for every $i \in [v']$ we set $ \bx^i = x^0_i\cdot\bx_1$. It is easy to check that $\bx$ as constructed satisfies $\|\bx\|_0 \leqs k_2 + k_1k_2$ and $\bC\bx = \by$. 
	
	\paragraph{Soundness}: Let $(\bA,\by_1,k_2)$ and $(\bB,\by_2,k_2)$ be NO instances. As in the the proof of Lemma \ref{lem:gap_step}, we consider the row-blocks $S_0,S_1,\ldots,S_{v'}$ which form an identical partition of $[u' + uv']$. Let $\bx \in \Z^{v' + vv'}$ be any vector such that $\bC \bx = r\cdot\by$ for some $r \in \Z \setminus \{0\}$. Since $\bx^0$ satisfies the constraints along $S_0$, we have $\bB \bx^0 = r\cdot\by_2$, and therefore $\|\bx^0\|_0 > \eta_2\cdot k_2$. Fix any $i \in [v']$ such that $x^0_i \neq 0$. Then the sub-vector $\bx^i$ must satisfy the constraints along row-block $S_i$ i.e., $\bA \bx^i = rx^0_i.\by_1$ and therefore, we have $\|\bx^i\|_0 > \eta_1k_1$. Since, this happens for every such choice of $i \in [v']$ such that $x^0_i \neq 0$, we have $\|\bx\|_0 > \eta_2k_2 + \eta_1\eta_2k_1k_2$.

	Therefore, given $\lvec_\eta$ instance $(\bA,\by,k)$, we can construct the $i$-wise composed instance $(\bA^{(i)},\by^{(i)},k(k+1)^{i-1}) = (\bA,\by,k) \oplus_\cL (\bA^{(i-1)},\by^{(i-1)},k(k + 1)^{i-2})$. Using the above arguments, it is easy to check that $(\bA^{(i)},\by^{(i)},k(k+1)^{i-1})$ is a $\lvec_{\eta^{3i/2}}$ instance. Setting $i \geqs \lceil 3c/2 \rceil$ gives us the required gap. 
\end{proof}

Finally, the following lemma completes the reduction to $\snvp$. 

\begin{lemma}					\label{lem:lvec-to-snvp}	
	For all choices of $\eta \geqs 1$ and $p \geqs 1$, there is an FPT reduction from $k$-$\lvec_\eta$ to $k$-$\snvp_{p,\eta}$. Furthermore, the resulting $\snvp_{p,\eta}$ instance $(\bA',\by',k)$ (where $\bA' \in \Z^{n' \times m}$ and $\by' \in \Z^{n'}$) satisfies the following properties:
	\begin{itemize}
		\item For the YES case, there exists $\bx \in \Z^m$ such that $(\bA' \bx - \by')$ is a $\{0,1\}$-vector and $\|\bA' \bx - r\cdot\by' \|_0 \leqs k$.
		\item For the NO case, for all $\bx \in \Z^m$ and $r \in \Z \setminus \{0\}$, we have $\|\bA' \bx - r\cdot\by' \|_0 > \eta \cdot k$
	\end{itemize}
\end{lemma}	
	 \begin{proof}
	 	Let $(\bA,\by,k)$ be a $\lvec_\eta$ instance where $\bA \in \Z^{n \times m}$ and $\by \in \Z^{n}$. We construct the matrix $\bA' \in \Z^{(\lceil \eta k + 1\rceil n + m) \times m}$ and $\by' \in \Z^{\lceil \eta k + 1\rceil n + m}$ identical to the proof of Lemma \ref{lem:gapvec-to-sncp}, and claim that $(\bA',\bx',k)$ is a $\snvp_\eta$ instance with the additional properties as guaranteed by the lemma.
	 	
	 	Indeed, if $(\bA,\by,k)$ is a YES instance, then there exists $\bx \in \{0,1\}^m$ such that $\|\bx\|_0 \leqs k$ such that $\bA\bx = \by$. Consequently, $\bA'\bx - \by'$ is $0$ on the first $\lceil \eta k + 1\rceil n$ coordinates, and it is just $\bx$ on the last $m$ coordinates. Conversely, if $(\bA,\by,k)$ is a NO instance, then as before we consider two cases. If $\bx \in \Z^m$ such that $\|\bx\|_0 > \eta k$, then it is easy to see that $\bA'\bx - \by'$ has Hamming weight greater than $\eta k$ on the last $m$-coordinates. On the other hand if $\|\bx\|_0 \leqs \eta k$, then we know that $\bA \bx \neq r\cdot\by$ for any $r \in \Z \setminus \{0\}$ Therefore $\|\bA\bx - r \cdot \by\|_0 \geqs 1$ which implies that $\|\bA'\bx - r \cdot \by'\|_0 \geqs \eta k$. 
	 	
	 	Finally, since in the YES case the witness is a $\{0,1\}$-vector and in the NO case the guarantees are in the Hamming norm, the completeness and soundness properties are satisfied by $(\bA',\by',k)$ for any value of $p \geqs 1$.
	\end{proof}

	\noindent{\bf Putting Things Together}: Given a $\csp_\epsilon$-instance $\Gamma$, using Lemma \ref{lem:csp-to-lvec} we construct a $\lvec_{1 + \epsilon/3}$-instance. Using Lemma \ref{lem:lvec-gap-amp}, we reduce $\lvec_{1 + \epsilon/3}$ to $\lvec_\eta$ for any choice of $\eta \geqs 1 + \epsilon/3$. Finally, using Lemma \ref{lem:lvec-to-snvp}, we reduce the $\lvec_\eta$ instance to $\snvp_{p,\eta}$ instance. Since all the steps in the reduction are FPT, this completes the proof.

\end{document}